\documentclass[10pt,a4paper]{article}
\usepackage[a4paper,textwidth=144mm,top=25mm,bottom=25mm]{geometry}

\usepackage[utf8]{inputenc}

\usepackage{graphicx}
\usepackage{multicol,multirow}
\usepackage{amsmath,amssymb,amsfonts}
\usepackage{amsthm}
\usepackage{rotating}
\usepackage{appendix}
\usepackage{ifpdf}
\usepackage[T1]{fontenc}
\usepackage{newtxtext}
\usepackage{newtxmath}
\usepackage{textcomp}
\usepackage{xcolor}
\usepackage{bm}
\usepackage{booktabs}
\usepackage{enumitem}
\usepackage{array}
\usepackage{tabularx}
\usepackage{caption}
\usepackage{float}
\usepackage{listings}
\usepackage[style=apa,backend=biber,natbib=true]{biblatex}
\usepackage[colorlinks,allcolors=blue]{hyperref}
\usepackage{pdflscape}

\theoremstyle{definition}
\newtheorem{theorem}{Theorem}
\newtheorem{definition}{Definition}
\numberwithin{equation}{section}

\newcommand{\onlineresource}[1]{%
  \clearpage
  \setcounter{section}{0}%
  \setcounter{table}{0}%
  \setcounter{figure}{0}%
  \renewcommand{\thesection}{S\arabic{section}}%
  \renewcommand{\thetable}{S\arabic{table}}%
  \renewcommand{\thefigure}{S\arabic{figure}}%
  \section*{\centering Online Resource #1}%
  \addcontentsline{toc}{section}{Online Resource #1}%
  \markboth{Online Resource #1}{Online Resource #1}%
}

\title{Efficient Bayesian Estimation of Dynamic Structural Equation Models via State Space Marginalization}
\author{Øystein Sørensen\\[0.4em]
  \small\itshape Department of Psychology, University of Oslo, Oslo, Norway}
\date{}

\begin{document}

\maketitle

\begin{abstract}
Dynamic structural equation models (DSEMs) combine time-series modeling of within-person processes with hierarchical modeling of between-person differences and differences between timepoints, and have become very popular for the analysis of intensive longitudinal data in the social sciences. An important computational bottleneck has, however, still not been resolved: whenever the underlying process is assumed to be latent and measured by one or more indicators per timepoint, currently published algorithms rely on inefficient brute-force Markov chain Monte Carlo sampling which scales poorly as the number of timepoints and participants increases and results in highly correlated samples. The main result of this paper shows that the within-level part of any DSEM can be reformulated as a linear Gaussian state space model. Consequently, the latent states can be analytically marginalized using a Kalman filter, allowing for highly efficient estimation via Hamiltonian Monte Carlo. This makes estimation of DSEMs computationally tractable for much larger datasets---both in terms of timepoints and participants---than what has been previously possible. We demonstrate the proposed algorithm in simulation experiments and with an application example, showing it can be orders of magnitude more efficient than standard Metropolis-within-Gibbs approaches.
\end{abstract}


\vspace{0.5cm}

\section{Introduction}

Intensive longitudinal data, characterized by a large number of observations per participant observed over a relatively short time span \citep{hamaker_at_2018}, have become very common in the social sciences over the last decade. In a seminal paper, \citet{asparouhov_dynamic_2018} introduced dynamic structural equation modeling (DSEM), a comprehensive framework for analyzing intensive longitudinal data which combines multilevel modeling, time-series modeling, structural equation modeling (SEM), and time-varying effects modeling. \citet{asparouhov_dynamic_2018} also proposed a Metropolis-within-Gibbs algorithm, which is implemented in the proprietary \textsc{Mplus} software \citep{mplus}.

Although several methodological papers demonstrate how to use DSEM or related frameworks with latent variables \citep{holtmannModelingWithinlevelLatent2025,andriamiaranaRequirementsNonlinearDynamic2023,falehDynamicLatentClass2026,kelavaForecastingIntraindividualChanges2022,oh_incorporating_2023}, most applications of DSEM published to date use sum scores or raw observations, leaving the SEM part of DSEM underutilized despite the well-known advantages of working with latent variable formulations \citep{mcneishThinkingTwiceSum2020}. Several factors may contribute to this. First of all, these models are demanding to specify, and until recently applied researchers had few tutorials or ready-made implementations to draw on. That situation is improving, through tutorial treatments accompanied by model code \citep{holtmannModelingWithinlevelLatent2025} and through packages such as the \textsc{R} \citep{rcoreteamLanguageEnvironmentStatistical2025} package \textsc{mlts} \citep{koslowskiMltsMultilevelLatent2024}, which fits multilevel latent time-series models, including DSEM. A second factor, and the one this paper addresses, is that the Markov chain Monte Carlo (MCMC) algorithm of \citet{asparouhov_dynamic_2018} suffers from a computational bottleneck when latent traits are included. The two factors reinforce one another: the sample sizes reported as necessary for these models are precisely those at which the bottleneck becomes severe, with \citet{andriamiaranaRequirementsNonlinearDynamic2023} requiring at least $N=50$ participants and $T=25$ timepoints,\footnote{In the rest of the paper, whenever the symbols $N$ and $T$ are used, they denote the number of participants and the number of timepoints, respectively.} and \citet{holtmannModelingWithinlevelLatent2025} recommending at least $100$ timepoints for at least $100$ persons in random-effects factor models. The computational challenge is well-known in the literature: in the context of psychometric network estimation, \citet[p. 212]{epskamp_psychometric_2020} states that hierarchical modeling of multivariate latent traits from time series data is ``computationally very expensive'', and authors have even suggested abandoning Bayesian methods altogether and instead rely on maximum likelihood estimation \citep{sakalauskas_technique_2024}.

The number of parameters to sample in a single iteration of \citet{asparouhov_dynamic_2018}'s MCMC algorithm is proportional to $N \cdot T$ in the presence of latent traits but only $N + T$ in the absence of latent traits. \citet[p. 386]{asparouhov_dynamic_2018} state that latent traits in the within-level model ``cannot be generated simultaneously in an efficient manner'', i.e., that the latent trait for each participant at each timepoint needs to be sampled independently. However, as we show, the need to sample these latent traits can be bypassed entirely by using a Kalman filter \citep{kalman_new_1960} to integrate them out analytically, replacing stochastic sampling with routine linear algebra. Forward-filtering backward-sampling \citep{carterGibbsSamplingState1994, fruhwirth-schnatterDataAugmentationDynamic1994} has been used in this literature, although for a different purpose: \citet{kelavaForecastingIntraindividualChanges2022} integrate such an algorithm with the Gibbs sampler in \textsc{JAGS} in order to obtain forecast distributions for within-level latent factors alongside the parameter draws. The algorithm of \citet{asparouhov_dynamic_2018} could in principle be adapted to use it differently, as a blocked update generating each participant's latent trajectory simultaneously, but embedding such an update in a Gibbs sampler still suffers from severe autocorrelation between the trajectories and the hierarchical participant-specific and timepoint-specific parameters. Alternatively, integrating the traits out to create a collapsed Gibbs sampler \citep{liuCollapsedGibbsSampler1994} avoids this autocorrelation, but it would require re-running the Kalman filter each time a parameter is sampled, and hence scale terribly. The same computational concern applies to classical Metropolis-Hastings algorithms.

Hamiltonian Monte Carlo \citep{Neal2011}, and particularly the No-U-Turn Sampler (NUTS) \citep{hoffman_no_2014}, is however ideally suited to marginalizing over latent variables using a Kalman filter. Since NUTS ensures high acceptance ratios even when proposing thousands of parameters in a single accept-reject step, the Kalman filter only needs to be run a single time for each full iteration. This reduces the number of parameters to be sampled in each iteration to the order of $N+T$, also in the presence of latent traits. A similar computational algorithm using the Kalman-Bucy filter \citep{kalman_new_1961} has been proposed in the context of continuous-time structural equation models \citep{Driver2018}. The \textsc{dynr} package \citep{ouWhatsDynrPackage2019} also uses the Kalman filter for estimating linear and nonlinear dynamic models, but without hierarchical modeling. More broadly, latent variable formulations of dynamic factor models have a long history in the statistical and biostatistical literature, for single multivariate time series \citep{cuiGeneralizedDynamicFactor2014} and, more recently, for samples of individuals \citep{liuMixedEffectsStateSpaceModels2011, tranModelingLocalDependence2021, abbottContinuousTimeDynamicFactor2025}, with the multilevel dynamic factor model of \citet{songAnalyzingMultipleMultivariate2014} closest in spirit to the models considered here. What distinguishes the DSEM setting is that a measurement model is combined with dynamics that vary randomly across participants, which is what makes the number of latent variables grow as $\mathcal{O}(N \cdot T)$ and hence makes the marginalization we propose worth performing. We note that, contrary to what is frequently stated in the applied literature, the DSEM algorithm implemented in \textsc{Mplus} does not use a Kalman filter: missing observations and unequally spaced intervals are handled by data augmentation, with the missing values forming an additional block of the Gibbs sampler \citep[p.~365]{asparouhov_dynamic_2018}.

The paper proceeds as follows. In Section \ref{sec:Background}, we first give necessary background on DSEM, and then on linear Gaussian state space models (LG-SSMs) and the Kalman filter. In Section \ref{sec:SSM_DSEM}, we present our theorem which shows that the within-level part of a DSEM can be equivalently formulated as an LG-SSM and then show how this can be used to perform efficient Bayesian estimation using marginalization. In Section \ref{sec:Simulations}, we present the results of simulation experiments for a wide range DSEMs, which demonstrate that the proposed algorithm in some cases is orders of magnitude more efficient than Metropolis-within-Gibbs and in other cases is comparable. In Section \ref{sec:Application} we present an application example, using a dataset from \citet{rowlandMindfulnessAffectNetworkDensity2020} measuring positive and negative affect six times per day over a period of 40 days in 125 individuals. We discuss the results and possible extensions of the framework in Section \ref{sec:Discussion}.

\section{Background}
\label{sec:Background}

\subsection{Dynamic Structural Equation Models}

A $U$-dimensional response vector $\bm{y}_{it}$ is assumed measured in $N$ participants ($i=1,2,\dots,N$) at $T$ timepoints ($t=1,2,\dots,T$). Unequally spaced timepoints or an unequal number of timepoints per participant can be dealt with by introducing missing observations.\footnote{Truly continuous-time modeling \citep{Driver2018,oudContinuousTimeState2000} is in general not feasible with DSEM, however, since this would require introducing an extremely dense grid with a very high degree of missingness.} Each observation is assumed decomposable into a part $\bm{y}_{3,t}$ which captures systematic variation between timepoints, a part $\bm{y}_{2,i}$ which captures systematic variation between participants, and a part $\bm{y}_{1,it}$ which captures all remaining variability. Consequently, we can write
\begin{equation}
    \label{eq:ResponseDecomposition}
    \bm{y}_{it} = \bm{y}_{1,it} + \bm{y}_{2,i} + \bm{y}_{3,t}.
\end{equation}
Next, we assume that each term on the right-hand side of \eqref{eq:ResponseDecomposition} is an indirect reflective measurement of some latent variable. We use the term indirect because all terms on the right-hand side of \eqref{eq:ResponseDecomposition} are latent variables; only $\bm{y}_{it}$ is directly observed.

The $V_{3}$-dimensional latent variables $\bm{\eta}_{3,t}$ vary between timepoints and are related to $\bm{y}_{3,t}$ through the between-timepoints model
\begin{equation}
    \begin{aligned}
        \bm{y}_{3,t}    & = \bm{\nu}_{3} + \bm{\Lambda}_{3} \bm{\eta}_{3,t} + \bm{K}_{3} \bm{X}_{3,t} + \bm{\epsilon}_{3,t} \\
        \bm{\eta}_{3,t} & = \bm{\alpha}_{3} + \bm{B}_{3}\bm{\eta}_{3,t} + \bm{\Gamma}_{3}\bm{X}_{3,t} + \bm{\xi}_{3,t}.
    \end{aligned}
    \label{eq:BetweenTimepointModel}
\end{equation}
In \eqref{eq:BetweenTimepointModel}, $\bm{\nu}_{3}$ and $\bm{\alpha}_{3}$ are intercept terms, $\bm{\Lambda}_{3}$ is a loading matrix, $\bm{B}_{3}$ is a strictly lower triangular matrix for regression between latent variables, $\bm{X}_{3,t}$ is a vector of covariates varying between timepoints, and $\bm{K}_{3}$ and $\bm{\Gamma}_{3}$ are matrices with regression coefficients. The vector of residuals $\bm{\epsilon}_{3,t}$ is assumed normally distributed around zero with covariance matrix $\bm{\Sigma}_{3}$, $\bm{\epsilon}_{3,t} \sim \mathcal{N}(\bm{0}, \bm{\Sigma}_{3})$. The vector of disturbances $\bm{\xi}_{3,t}$ is also normally distributed around zero, $\bm{\xi}_{3,t} \sim \mathcal{N}(\bm{0}, \bm{\Psi}_{3})$. The assumption that $\bm{B}_{3}$ is strictly lower triangular ensures that relationships among latent variables are recursive (see also \citet[p. 385]{asparouhov_dynamic_2018}). The same rationale applies when we assume below that $\bm{B}_{2}$ in \eqref{eq:BetweenIndividualModel} and $\bm{B}_{1,0it}$ and $\bm{R}_{0it}$ in \eqref{eq:WithinLevelModel} are strictly lower triangular.

Next, the $V_{2}$-dimensional latent variables $\bm{\eta}_{2,i}$ vary between participants and are related to $\bm{y}_{2,i}$ through the between-participants model
\begin{equation}
    \begin{aligned}
        \bm{y}_{2,i}    & = \bm{\nu}_{2} + \bm{\Lambda}_{2} \bm{\eta}_{2,i} + \bm{K}_{2}\bm{X}_{2,i} + \bm{\epsilon}_{2,i} \\
        \bm{\eta}_{2,i} & = \bm{\alpha}_{2} + \bm{B}_{2} \bm{\eta}_{2,i} + \bm{\Gamma}_{2} \bm{X}_{2,i} + \bm{\xi}_{2,i},
    \end{aligned}
    \label{eq:BetweenIndividualModel}
\end{equation}
where the components have similar interpretations as in \eqref{eq:BetweenTimepointModel} except that the covariates $\bm{X}_{2,i}$ and noise terms $\bm{\epsilon}_{2,i} \sim \mathcal{N}(\bm{0}, \bm{\Sigma}_{2})$ and $\bm{\xi}_{2,i} \sim \mathcal{N}(\bm{0}, \bm{\Psi}_{2})$ now vary between individuals and not between timepoints.

The within-level model relates the $V_{1}$-dimensional latent variables $\bm{\eta}_{1,it}$ varying both between participants and timepoints to $\bm{y}_{1,it}$ through a timeseries model with lag $L \geq 0$,
\begin{equation}
    \begin{aligned}
        \bm{y}_{1,it}    & = \bm{\nu}_{1,it} + \bm{\Lambda}_{1,it}(L) \bm{\eta}_{1,it} + \bm{R}_{it}(L) \bm{y}_{1,it} + \bm{K}_{1,it} \bm{X}_{1,it} + \bm{\epsilon}_{1,it} \\
        \bm{\eta}_{1,it} & = \bm{\alpha}_{1,it} + \bm{B}_{1,it}(L) \bm{\eta}_{1,it} + \bm{Q}_{it}(L) \bm{y}_{1,it} + \bm{\Gamma}_{1,it} \bm{X}_{1,it} + \bm{\xi}_{1,it},
    \end{aligned}
    \label{eq:WithinLevelModel}
\end{equation}
where $\bm{X}_{1,it}$ is a vector of covariates varying between timepoints and individuals, and $\bm{K}_{1,it}$ and $\bm{\Gamma}_{1,it}$ are the corresponding matrices with regression coefficients, and $\bm{\nu}_{1,it}$ and $\bm{\alpha}_{1,it}$ are intercepts. In \eqref{eq:WithinLevelModel} we have
\begin{align*}
     & \bm{\Lambda}_{1,it}(L) = \sum_{l=0}^{L} \bm{\Lambda}_{1,lit} L^{l}, \quad \bm{R}_{it}(L) = \sum_{l=0}^{L} \bm{R}_{lit} L^{l}, \quad  \bm{B}_{1,it}(L) = \sum_{l=0}^{L} \bm{B}_{1,lit} L^{l}, \quad \bm{Q}_{it}(L) = \sum_{l=0}^{L} \bm{Q}_{lit} L^{l},
\end{align*}
where $L$ is a lag operator such that $L^{k} \bm{x}_{t} = \bm{x}_{t-k}$. Inside these definitions, $\bm{\Lambda}_{1,lit}$ is the factor loading matrix relating the lag-$l$ latent variables to measurements, $\bm{R}_{lit}$ and $\bm{B}_{1,lit}$ for $l>0$ are autoregression matrices, $\bm{R}_{0it}$ and $\bm{B}_{1,0it}$ are strictly lower triangular matrices for regression between components of $\bm{y}_{1,it}$ and $\bm{\eta}_{1,it}$, respectively, and the matrix $\bm{Q}_{lit}$ relates the lag-$l$ measurements to the latent variables. We assume that the measurement noise $\bm{\epsilon}_{1,it}$ and process noise $\bm{\xi}_{1,it}$ are normally distributed, $\bm{\epsilon}_{1,it} \sim \mathcal{N}(\bm{0}, \bm{\Sigma}_{1,it})$ and $\bm{\xi}_{1,it} \sim \mathcal{N}(\bm{0}, \bm{\Psi}_{1,it})$. Note that \citet[p. 361]{asparouhov_dynamic_2018} use lagged terms also for $\bm{X}_{1,it}$, but as they point out their model can equivalently be written as in \eqref{eq:WithinLevelModel}.

As the notation in \eqref{eq:WithinLevelModel} suggests, the parameters $\bm{\nu}_{1,it}$, $\bm{\Lambda}_{1,it}(L)$, $\bm{R}_{it}(L)$, $\bm{K}_{1,it}$, $\bm{\alpha}_{1,it}$, $\bm{B}_{1,it}(L)$, $\bm{Q}_{it}(L)$, and $\bm{\Gamma}_{1,it}$, as well as the covariance matrices $\bm{\Sigma}_{1,it}$ and $\bm{\Psi}_{1,it}$, are allowed to vary both between timepoints and participants. These parameters are additively composed from the elements of $\bm{\eta}_{2,i}$ and $\bm{\eta}_{3,t}$, and potentially transformed by some function $f(\cdot)$. For example, if the parameter $s$ is the square root of an element on the diagonal of $\bm{\Sigma}_{1,it}$, i.e., an individual-specific residual standard deviation, it could be represented as $s = \exp(\eta_{2,iq} + \eta_{3,tr})$, where $\eta_{2,iq}$ and $\eta_{3,tr}$ are the $q$th and $r$th elements of $\bm{\eta}_{2,i}$ and $\bm{\eta}_{3,t}$, respectively, for some indices $q$ and $r$, and the exponential transformation ensures that the standard deviation is positive. In most applications, additional constraints will be needed to identify the model. For example, it will typically not be possible to estimate both the measurement-level intercepts ($\bm{\nu}_{3}$, $\bm{\nu}_{2}$, and $\bm{\nu}_{1,it}$) and the corresponding latent-variable intercepts ($\bm{\alpha}_{3}$, $\bm{\alpha}_{2}$, and $\bm{\alpha}_{1,it}$), and one might fix one factor loading per item to $1$ in the loading matrices ($\bm{\Lambda}_{3}$, $\bm{\Lambda}_{2}$, $\bm{\Lambda}_{1,lit}$). However, rather than deriving generic identifiability conditions, we suggest that model identification is best assessed on a case-by-case basis.

Due to the flexibility of our proposed algorithm, we do not make any assumptions about prior distributions at this point, as the choice of priors should ideally be made in the context of a specific research problem \citep{schadPrincipledBayesianWorkflow2020}. For the full DSEM framework, to the best of our knowledge no alternatives to Metropolis-within-Gibbs have been proposed. For example, \citet{li_fitting_2022} compare \textsc{Stan} \citep{carpenter_stan_2017}, \textsc{JAGS} \citep{plummer2003jags}, and \textsc{Mplus} implementations for DSEM, the first of which uses NUTS \citep{hoffman_no_2014}, but they allow only manifest variables $\bm{y}_{1,it}$ and no latent traits $\bm{\eta}_{1,it}$ in the within-level model.

\subsection{State Space Models and Kalman Filters}

Rather than using \eqref{eq:WithinLevelModel}, we could formulate the within-level model as an LG-SSM \citep{durbin_time_2012} where for each individual $i$,
\begin{equation}
    \begin{aligned}
        \bm{y}_{1,it}               & = \bm{Z}_{it} \tilde{\bm{\eta}}_{1,it} + \bm{d}_{it} + \bm{v}_{it}, \quad \bm{v}_{it} \sim \mathcal{N}\left(\bm{0}, \bm{H}_{it}\right)                      \\
        \tilde{\bm{\eta}}_{1,i,t+1} & = \bm{T}_{it} \tilde{\bm{\eta}}_{1,it} + \bm{c}_{it} + \bm{w}_{it}, \quad \bm{w}_{it} \sim \mathcal{N}\left(\bm{0}, \bm{W}_{it}\right), \quad t=1,2,\dots,T \\
        \tilde{\bm{\eta}}_{1,i1}    & = \bm{\xi}_{i}, \quad \bm{\xi}_{i} \sim \mathcal{N}\left(\bm{a}_{i1}, \bm{P}_{i1}\right),
    \end{aligned}
    \label{eq:BasicStateSpaceModel}
\end{equation}
where $\tilde{\bm{\eta}}_{1,it}$ is a state vector, $\bm{a}_{i1}$ and $\bm{P}_{i1}$ are the mean vector and covariance matrix in the Gaussian prior for the initial state, $\bm{Z}_{it}$, $\bm{H}_{it}$, $\bm{T}_{it}$, and $\bm{W}_{it}$ are matrices, $\bm{d}_{it}$ and $\bm{c}_{it}$ are mean vectors, and $\bm{v}_{it}$ and $\bm{w}_{it}$ are noise terms.

An advantage of the state space formulation is that the state dynamics is Markovian; whatever value the lag $L$ takes, $\tilde{\bm{\eta}}_{1,i,t+1}$ only depends on $\tilde{\bm{\eta}}_{1,it}$ and $\bm{y}_{1,it}$ only depends on $\tilde{\bm{\eta}}_{1,it}$. This lets us use the Kalman filter \citep{kalman_new_1960} to compute the exact posterior distribution of the state variables. For each timepoint $t=2,3,\dots,T$ we have the recursion
\begin{equation}
    \begin{aligned}
        \bm{a}_{i,t|t-1} & = \bm{T}_{i,t-1} \bm{a}_{i,t-1|t-1} + \bm{c}_{i,t-1}                                                 \\
        \bm{P}_{i,t|t-1} & = \bm{T}_{i,t-1} \bm{P}_{i,t-1|t-1} \bm{T}_{i,t-1}^{T} + \bm{W}_{i,t-1}                              \\
        \bm{v}_{it}      & = \bm{y}_{1,it} - \bm{Z}_{it} \bm{a}_{i,t|t-1} - \bm{d}_{it}                                         \\
        \bm{F}_{it}      & = \bm{Z}_{it} \bm{P}_{i,t|t-1} \bm{Z}_{it}^{T} + \bm{H}_{it}                                         \\
        \bm{a}_{i,t|t}   & = \bm{a}_{i,t|t-1} + \bm{P}_{i,t|t-1} \bm{Z}_{it}^{T} \bm{F}_{it}^{-1} \bm{v}_{it}                   \\
        \bm{P}_{i,t|t}   & = \bm{P}_{i,t|t-1} - \bm{P}_{i,t|t-1} \bm{Z}_{it}^{T} \bm{F}_{it}^{-1} \bm{Z}_{it} \bm{P}_{i,t|t-1},
    \end{aligned}
    \label{eq:BasicKalmanFilter}
\end{equation}
where for some $s \leq t$ we have $P(\tilde{\bm{\eta}}_{1,it} | \bm{y}_{1,is}) = \mathcal{N}(\bm{a}_{i,t|s}, \bm{P}_{i,t|s})$, i.e., $\bm{a}_{i,t|s}$ is the posterior mean of the latent state at time $t$ after observing data up to time $s$ and $\bm{P}_{i,t|s}$ is the posterior covariance. For timepoint $t=1$ we use the prior initialization $\bm{a}_{i,1|0}=\bm{a}_{i1}$ and $\bm{P}_{i,1|0} = \bm{P}_{i1}$ and perform only the last four updates in \eqref{eq:BasicKalmanFilter}, to obtain $\bm{v}_{i1}$, $\bm{F}_{i1}$, $\bm{a}_{i,1|1}$, and $\bm{P}_{i,1|1}$. The marginal likelihood, integrating over $\tilde{\bm{\eta}}_{1,it}$, is now given analytically by\footnote{With missing data $UNT$ in the first term is replaced by $\sum_{i=1}^{N}\sum_{t=1}^{T}U_{it}$, where $U_{it} \geq 0$ is the number of observed indicators from participant $i$ at timepoint $t$.}
\begin{equation}
    \label{eq:BasicStateSpaceLogLikelihood}
    \mathcal{L} = - \frac{UNT}{2} \log 2\pi - \sum_{i=1}^{N}  \sum_{t=1}^{T} \left\{ \frac{1}{2} \log \left| \bm{F}_{it} \right| + \frac{1}{2} \left( \bm{y}_{1,it} - \hat{\bm{y}}_{1,it} \right)^{T} \bm{F}_{it}^{-1} \left( \bm{y}_{1,it} - \hat{\bm{y}}_{1,it} \right)\right\} ,
\end{equation}
where $\hat{\bm{y}}_{1,it} = \bm{Z}_{it} \bm{a}_{i,t|t-1} + \bm{d}_{it}$.

\section{A State Space Formulation of DSEM}
\label{sec:SSM_DSEM}

In this section we show how the within-level model \eqref{eq:WithinLevelModel} of the DSEM can be reformulated as a state space model on the form \eqref{eq:BasicStateSpaceModel}. This lets us replace sampling of the within-level latent variables with the Kalman filter \eqref{eq:BasicKalmanFilter} and plug the log likelihood \eqref{eq:BasicStateSpaceLogLikelihood} directly into the accept-reject step of the NUTS algorithm.

\subsection{Augmented State Space Formulation}

Before stating our main result, we define the coefficient extraction operator $[\cdot]_{k}$ such that $[\bm{\Phi}_{it}(L)]_k \equiv \bm{D}_{k,it}$ for any polynomial $\bm{\Phi}_{it}(L) = \sum_{k} \bm{D}_{k,it} L^k$.

\begin{theorem}
    \label{th:Theorem1}
    For maximum lag $L \ge 1$, the within-level model \eqref{eq:WithinLevelModel} is exactly equivalent to the state space model
    \begin{equation}
        \label{eq:StateSpaceModelTheorem1}
        \begin{aligned}
            \bm{y}_{1,it}               & = \bm{Z}_{it} \tilde{\bm{\eta}}_{1,it}                                                                                                  \\
            \tilde{\bm{\eta}}_{1,i,t+1} & = \bm{T}_{it} \tilde{\bm{\eta}}_{1,it} + \bm{c}_{it} + \bm{w}_{it}, \quad \bm{w}_{it} \sim \mathcal{N}\left(\bm{0}, \bm{W}_{it}\right),
        \end{aligned}
    \end{equation}
    where the augmented state vector is $\tilde{\bm{\eta}}_{1,it} = [\bm{\eta}_{1,i,t}^{T}, \dots, \bm{\eta}_{1,i,t-L+1},\bm{y}_{1,i,t}, \dots, \bm{y}_{1,i,t-L+1}]^{T}$ and the measurement matrix extracts the contemporaneous observation, $\bm{Z}_{it} = [\bm{0}_{U \times L V_1} ~ \bm{I}_{U \times U} ~ \bm{0}_{U \times (L-1)U}]$. The transition matrix is
    \begin{equation*}
        \bm{T}_{it} =
        \begin{bmatrix}
            \bm{T}_{it}^{(1,1)}                                                            & \bm{T}_{it}^{(1,2)}                                                      \\
            \begin{bmatrix} \bm{I}_{(L-1)V_1} & \bm{0}_{(L-1)V_1 \times V_1} \end{bmatrix} & \bm{0}                                                                   \\
            \bm{T}_{it}^{(3,1)}                                                            & \bm{T}_{it}^{(3,2)}                                                      \\
            \bm{0}                                                                         & \begin{bmatrix} \bm{I}_{(L-1)U} & \bm{0}_{(L-1)U \times U} \end{bmatrix}
        \end{bmatrix} ,
    \end{equation*}
    where the block-columns for $k=1,\dots,L$ are defined by
    \begin{align*}
        \{\bm{T}_{it}^{(1,1)}\}_{k} & = \bm{M}_{1,i,t+1} [\bm{\mathcal{P}}^{\eta}_{i,t+1}(L)]_k, &
        \{\bm{T}_{it}^{(1,2)}\}_{k} & = \bm{M}_{1,i,t+1} [\bm{\mathcal{P}}^{y}_{i,t+1}(L)]_k       \\
        \{\bm{T}_{it}^{(3,1)}\}_{k} & = \bm{N}_{1,i,t+1} [\bm{\mathcal{Q}}^{\eta}_{i,t+1}(L)]_k, &
        \{\bm{T}_{it}^{(3,2)}\}_{k} & = \bm{N}_{1,i,t+1} [\bm{\mathcal{Q}}^{y}_{i,t+1}(L)]_k.
    \end{align*}
    The matrices $\bm{M}_{1,i,t+1}$ and $\bm{N}_{1,i,t+1}$ are defined in \eqref{eq:M1it} and \eqref{eq:N1it}, the composite polynomials $\bm{\mathcal{P}}^{\eta}_{i,t+1}(L)$, $\bm{\mathcal{P}}^{y}_{i,t+1}(L)$, $\bm{\mathcal{Q}}^{\eta}_{i,t+1}(L)$, and $\bm{\mathcal{Q}}^{y}_{i,t+1}(L)$ are defined in \eqref{eq:P_poly} and \eqref{eq:Q_poly}, and the intercept vector $\bm{c}_{it}$ and covariance matrix $\bm{W}_{it}$ are defined in \eqref{eq:c_intercept} and \eqref{eq:w_process_noise}.
\end{theorem}

\begin{proof}
    A proof is given in Appendix \ref{app:Proof}.
\end{proof}

\subsection{Bayesian Estimation via Marginalization}

The state space formulation derived in Theorem \ref{th:Theorem1} allows for a fundamental change in the estimation strategy compared to the standard implementation \citep{asparouhov_dynamic_2018}. By utilizing the Kalman filter, we can analytically integrate out the latent states $\tilde{\bm{\eta}}_{1,it}$. Let $\bm{\Theta}$ denote the set of all static model parameters. The marginal likelihood of the data given $\bm{\Theta}$, $\bm{\eta}_{2} = \{\bm{\eta}_{2,i}\}_{i=1}^{N}$ and $\bm{\eta}_{3} = \{\bm{\eta}_{3,t}\}_{t=1}^{T}$, $P(\bm{y} | \bm{\Theta}, \bm{\eta}_{2}, \bm{\eta}_{3})$, is computed exactly via the recursive update equations \eqref{eq:BasicKalmanFilter} and the summation in \eqref{eq:BasicStateSpaceLogLikelihood}. Targeting the marginal posterior
\begin{equation}
    \label{eq:TargetPosterior}
    P(\bm{\Theta}, \bm{\eta}_{2}, \bm{\eta}_{3} | \bm{y}) \propto P(\bm{y} | \bm{\Theta}, \bm{\eta}_{2}, \bm{\eta}_{3}) P(\bm{\Theta}, \bm{\eta}_{2}, \bm{\eta}_{3}),
\end{equation}
where $\bm{y}$ denotes the complete vector of all observations, reduces the dimension of the sampling problem from $\mathcal{O}(N \cdot T \cdot V_1 + N \cdot V_{2} + T \cdot V_{3} + |\bm{\Theta}|)$ with standard MCMC to $\mathcal{O}(N \cdot V_{2} + T \cdot V_{3} + |\bm{\Theta}|)$. However, neither \citet{asparouhov_dynamic_2018}'s Metropolis-within-Gibbs sampler nor a random walk Metropolis algorithm would be able to utilize this. In the former, parameters are sampled in blocks conditionally on all other parameters, and the Kalman filter and marginal likelihood calculation would have to be repeated for each new parameter block, i.e., a large number of times per iteration. Random walk Metropolis, on the other hand, allows updating all parameters targeted by \eqref{eq:TargetPosterior} in a single accept-reject step. However, due to the high dimensionality of the problem the acceptance rate would approach zero, rendering the algorithm useless.

Hamiltonian Monte Carlo \citep{Neal2011}, on the other hand, circumvents both of these problems by utilizing the exact gradient of the marginalized log posterior, the logarithm of \eqref{eq:TargetPosterior}. By simulating Hamiltonian dynamics it proposes distant states that theoretically conserve Hamiltonian energy, yielding an acceptance probability of 1. While discrete-time leapfrog integration introduces minor numerical errors, NUTS \citep{hoffman_no_2014} adaptively tunes its step size to bound this error, maintaining high acceptance rates regardless of dimensionality. This gradient-guided exploration makes joint updating of all static parameters and between-level latent variables highly efficient.

The general algorithmic scheme is as follows:

\begin{enumerate}
    \item \textbf{Initialization:} Initialize parameters $\bm{\Theta}^{(0)}$, $\bm{\eta}_{2}^{(0)}$, and $\bm{\eta}_{3}^{(0)}$.
    \item \textbf{Transition:} At iteration $k$, propose a new set of parameters $\{\bm{\Theta}^{*}, \bm{\eta}_{2}^{*}, \bm{\eta}_{3}^{*}\}$ using the NUTS proposal mechanism. This involves repeatedly evaluating the marginal likelihood \eqref{eq:BasicStateSpaceLogLikelihood} by running the Kalman filter recursion \eqref{eq:BasicKalmanFilter}, once for each leapfrog step.
    \item \textbf{Likelihood Computation:} Conditional on $\{\bm{\Theta}^{*}, \bm{\eta}_{2}^{*}, \bm{\eta}_{3}^{*}\}$, run the Kalman filter recursion \eqref{eq:BasicKalmanFilter} over $i=1,\dots,N$ and $t=1, \dots, T$ to compute the marginal log-likelihood using \eqref{eq:BasicStateSpaceLogLikelihood}.
    \item \textbf{Acceptance:} Accept or reject $\{\bm{\Theta}^{*}, \bm{\eta}_{2}^{*}, \bm{\eta}_{3}^{*}\}$ based on the Metropolis ratio, which involves the likelihood computed in Step 3 and the prior $P(\bm{\Theta}, \bm{\eta}_{2}, \bm{\eta}_{3})$.
    \item \textbf{Repeat:} Repeat steps 2-4 until convergence.
\end{enumerate}

Posterior estimates of the latent states $\tilde{\bm{\eta}}_{1,it}$ can be recovered \textit{post hoc} by running the Kalman smoother using the posterior samples of $\bm{\Theta}$, $\bm{\eta}_{2}$, and $\bm{\eta}_{3}$, without affecting the efficiency of the parameter estimation. If some $\bm{y}_{1,it}$ is missing, steps 3--6 of the Kalman filter are skipped and the mean and covariance of the latent state are carried forward, i.e., $\bm{a}_{i,t|t} = \bm{a}_{i,t|t-1}$ and $\bm{P}_{i,t|t}=\bm{P}_{i,t|t-1}$. This analytically marginalizes over the missing data without needing to sample it as a parameter as in standard MCMC.

To ensure efficient exploration of the posterior a non-centered parameterization for the between-level latent variables may be useful \citep{papaspiliopoulos_general_2007}, and we have implemented this in all examples of the next section. However, whether a centered or non-centered parametrization should be preferred depends on the geometry of the particular problem \citep{betancourt_hamiltonian_2015}.

\section{Simulation Experiments}
\label{sec:Simulations}

We present simulation experiments comparing the proposed algorithm---hereafter denoted ``NUTS-Kalman''---with a Metropolis-within-Gibbs\footnote{For brevity, this sampler is denoted ``Gibbs'' in the figure legends.} sampler which resembles that described in \citet{asparouhov_dynamic_2018}. As we do not have an \textsc{Mplus} license we have implemented this sampler in \textsc{JAGS (4.3.2)} using \textsc{R (4.4.2)} \citep{rcoreteamLanguageEnvironmentStatistical2025} and the packages \textsc{rjags (4-17)} \citep{plummerRjagsBayesianGraphical2025} and \textsc{jagsUI (1.6.3)} \citep{kellnerJagsUIWrapperRjags2026}. A distinction between our sampler and \citet{asparouhov_dynamic_2018} is that we constrain the autoregressive parameters to yield stationary solutions. Since this breaks conjugacy, we resort to Metropolis-Hastings---via slice sampling \citep{neal_slice_2003} implemented in \textsc{JAGS}---for a few more parameters than what they do. It follows that an actual \textsc{Mplus} implementation would plausibly mix marginally better per iteration than our baseline does, and the efficiency factors reported below should therefore be read as upper bounds on the gain relative to \textsc{Mplus} rather than as estimates of it. Still, for all simulation experiments considered, between 94.5\% and 98.5\% of the parameters in our implementation are updated via Gibbs sampling, and crucially, all of the $\mathcal{O}(N \cdot T)$ within-level latent variables are updated this way. We refer to Online Resource 2 for a detailed analysis of our \textsc{JAGS} implementation. NUTS-Kalman was implemented in \textsc{Stan (2.39.0)} using the \textsc{cmdstanr (0.9.0)} package \citep{gabryCmdstanrInterfaceCmdStan2025}. In order to separate the effect of using the Kalman filter from potential effects of using NUTS and \textsc{Stan}, we also implemented a full NUTS algorithm which samples over all latent states, which we denote ``NUTS-Joint''.

For all simulation experiments, note that both Metropolis-within-Gibbs and NUTS target the correct posterior. That is, when run for a sufficient number of iterations they will all get arbitrarily close to the true posterior distributions. The important practical question is whether a sufficiently good approximation can be obtained in reasonable time. Parameter values sampled by MCMC are in general auto-correlated, and the effective sample size (ESS) measures how many independent samples they correspond to. Hence, in order to obtain a sufficiently low Monte Carlo error to accurately estimate a posterior quantity of interest, the ESS needs to be sufficiently high \citep{zitzmann_going_2019}. Similar to other papers on Bayesian estimation \citep{bierkens_zig-zag_2019,prangle_lazy_2016} we thus compare algorithms in terms of their efficiency as measured by ESS per time unit.

An important point is that the ESS output by the \textsc{jagsUI} package is based on estimating ESS independently in each chain by fitting an autoregressive model using the \textsc{coda} package \citep{plummerCODAConvergenceDiagnosis2006} and then summing over all chains. This approach can be overoptimistic, e.g., if the posterior is bimodal and one chain gets stuck in each mode \citep{vehtariRankNormalizationFoldingLocalization2021}. \textsc{Stan} uses a more sophisticated approach based on the estimators proposed by \citet{vehtariRankNormalizationFoldingLocalization2021}, and we computed this also for the \textsc{JAGS} output using the \textsc{posterior} package \citep{Burkneretal2025}. This both ensures that we use state-of-the-art estimators for ESS and that comparisons are based on standardized quantities. Following \citet{vehtariRankNormalizationFoldingLocalization2021} we distinguish between bulk-ESS for diagnosing the sampling efficiency in the bulk of the posterior (e.g., for estimating posterior means) and the tail-ESS defined as the minimum ESS for 5\% and 95\% percentiles for diagnosing the sampling efficiency in the tails of the posterior (e.g., for estimating posterior intervals). The potential scale reduction factor $\hat{R}$ was also computed using the method proposed by \citet{vehtariRankNormalizationFoldingLocalization2021} for all algorithms. While \citet[p. 297]{Gelmanetal2003} suggested that maximum values of $\hat{R}$ over all parameters below $1.10$ are acceptable, \citet[p. 518]{vatsRevisitingGelmanRubin2021} have argued that ``a cutoff of $\hat{R} \leq 1.1$ is much too high to yield reasonable estimates of target quantities'', and an $\hat{R} \leq 1.01$ threshold has instead been suggested \citep{vehtariRankNormalizationFoldingLocalization2021}.

Making a fair comparison between MCMC algorithms in a large-scale simulation study is challenging, since setting the tuning parameters manually for each of several thousands of Monte Carlo sample is not feasible. To address this issue, for each simulation setting we ran pilots with 10 Monte Carlo samples prior to the full simulations, in which convergence was assessed and various tuning parameters were tested. For the \textsc{JAGS} implementation, an initial adaptation phase is used to set initial slice width for the parameters that are updated with slice sampling, and the burn-in also needs to be set. Here we used visual inspection of trace plots to ensure that the burn-in period was set as low as possible while still ensuring approximate convergence. For the \textsc{Stan} implementations, we varied the target acceptance probability and maximum tree depth during warm-up, and found the most liberal values which gave close to zero probability of divergent transitions after warm-up. For all algorithms, we set the total number of samples after burn-in/warm-up quite high, to ensure that the actual wall time was dominated by sampling (approximately) the true posterior. Since we compare algorithms in terms of sampling efficiency and time to reach a given ESS, the actual wall time is not directly relevant. A detailed description of the pilot runs for all simulation settings is provided in Online Resource 1. 

All simulations were run on the [Blinded] computing cluster hosted by [Institution Blinded]. All simulation results can be reproduced using code openly available in our OSF repository.\footnote{\url{https://osf.io/p754m/overview?view_only=f99de4145ddf49c9a76f0e38b5e42223}} The code used for these simulation experiments was developed with the assistance of Google's Gemini 3.1 Pro and Claude Opus 5. All AI-assisted code was thoroughly reviewed, modified, and independently verified by the author prior to execution. Supplementary simulation results are provided in Online Resource 4. In particular, Figures S5, S9, S11, S14, and S16 of Online Resource 4 show that our \textsc{Stan} and \textsc{JAGS} implementations yield similar posterior means and standard deviations, suggesting that they target the same posterior distributions.\footnote{Figure S16 shows some disagreement between NUTS-Joint and the two other algorithms, which happened because NUTS-Joint had serious convergence issues for this particular model, as further described in Section \ref{sec:CrossClassifiedVAR1}.} Further details about prior distributions, data-generating parameter values, and algorithmic implementations, are provided in Online Resource 3.

\subsection{Multilevel Scalar Latent AR(1) Model}
\label{sec:LatentScalarAR1}

We start by considering a univariate latent AR(1) model, with parameters varying between participants but not between timepoints. The decomposition \eqref{eq:ResponseDecomposition} becomes $y_{it} = y_{1,it} + y_{2,i}$, there is no between-timepoint model \eqref{eq:BetweenTimepointModel} and the between-participant model \eqref{eq:BetweenIndividualModel} is
\begin{equation*}
    \begin{aligned}
        y_{2,i}         & = \eta_{2,i1}                                                                                     \\
        \bm{\eta}_{2,i} & = \bm{\alpha}_{2} + \bm{\xi}_{2,i}, \quad \bm{\xi}_{2,i} \sim \mathcal{N}(\bm{0}, \bm{\Psi}_{2}),
    \end{aligned}
\end{equation*}
where $\eta_{2,i1}$ is the first component of $\bm{\eta}_{2,i}$ and $\bm{\Psi}_{2} = \text{diag}(\tau_{\nu}^{2}, \tau_{\phi}^{2}, \tau_{\log\sigma}^{2}, \tau_{\log\psi}^{2})$. The within-level model \eqref{eq:WithinLevelModel} is
\begin{equation*}
    \begin{aligned}
        y_{1,it}    & = \eta_{1,it} + \epsilon_{1,it}, \quad \epsilon_{1,it} \sim \mathcal{N}(0, \sigma_{2,i}^{2})   \\
        \eta_{1,it} & = \phi_{2,i} \eta_{1,i,t-1} + \xi_{1,it}, \quad \xi_{1,it} \sim \mathcal{N}(0,\psi_{2,i}^{2}),
    \end{aligned}
\end{equation*}
which is what \citet{asparouhov_dynamic_2018} call the measurement error AR(1) model. The autoregressive parameter $\phi_{2,i}$ and the standard deviations $\sigma_{2,i}$ and $\psi_{2,i}$ in the within-level model are related to the between-level latent variable $\bm{\eta}_{2,i}$ through the transformations $\phi_{2,i} = \tanh(\eta_{2i,2})$, $\sigma_{2,i} = \exp(\eta_{2i,3})$, and $\psi_{2,i} = \exp(\eta_{2i,4})$. The latter two ensure that the standard deviations are positive, and the first ensures stationarity by constraining the autoregressive coefficients to lie between $-1$ and $1$. The vector $\bm{\alpha}_{2} = (\nu, \phi, \log \sigma, \log \psi)^{T}$ contains the population means of the intercepts $\eta_{2,i1}$ and the between-level latent variables $\eta_{2i,2}$, $\eta_{2i,3}$, and $\eta_{2i,4}$, which we denote $\nu$, $\phi$, $\log \sigma$, and $\log \psi$, respectively. Although only one measurement is available per timepoint, the latent state $\eta_{1,it}$ is still unobserved and the model retains $\mathcal{O}(N \cdot T)$ within-level latent variables; what identifies the measurement error variance $\sigma_{2,i}^{2}$ separately from the process noise variance $\psi_{2,i}^{2}$ is the autoregression, which makes $\xi_{1,it}$ persist across timepoints.

We generated 100 random datasets from the model for each of six data dimension settings, with $N=50,100$ and $T=50,100,200$. For each generated dataset we ran both NUTS-Kalman and NUTS-Joint for 1,000 warm-up\footnote{To be consistent with \textsc{Stan} terminology, we use the term warm-up rather than burn-in when referring to \textsc{Stan} implementations, while we use the term burn-in for JAGS implementations.} iterations followed by 10,000 sampling iterations in each of four parallel chains, yielding a total of 40,000 posterior samples. Based on the piloting, the target average acceptance probability during warm-up was set to $0.80$ for NUTS-Kalman and $0.95$ for NUTS-Joint, with both having maximum tree depth of $10$. The Metropolis-within-Gibbs sampler was run for 100,000 iterations in each of four parallel chains after an adaptation phase of 2,000 iterations and no additional burn-in, yielding a total of 400,000 posterior samples.

The eight population level parameters in $\bm{\alpha}_{2}$ and $\bm{\Psi}_{2}$ were considered the parameters of interest, and for each run of an algorithm bulk-ESS and tail-ESS were computed as the minimum over these, whereas $\hat{R}$ was computed as the maximum. According to the $\hat{R}<1.01$ threshold, NUTS-Kalman and Metropolis-within-Gibbs succeeded for almost 100\% of the samples, whereas NUTS-Joint failed in about half of the samples. Loosening the threshold to $\hat{R}<1.05$, NUTS-Joint succeeded in most of the samples (see Figure S1 of Online Resource 4).

Figure \ref{fig:eg1-univariate-latent-ar1-efficiency} shows the efficiency of the three algorithms, measured in terms of bulk-ESS per minute. In all six cases considered, NUTS-Kalman was between 14 and 45 times more efficient than Metropolis-within-Gibbs. NUTS-Kalman was also between 14 and 30 times more efficient than NUTS-Joint, confirming that the efficiency gain was not due to \textsc{Stan} or NUTS in and of itself, but rather the use of the Kalman filter to marginalize over the within-level latent variables. The corresponding plot for efficiency in terms of tail-ESS is almost identical (see Figure S2 of Online Resource 4). Figures S3 and S4 of Online Resource 4 show histograms of bulk-ESS and tail-ESS, respectively, for all eight parameters of interest. They reveal that all three algorithms had comparable efficiency for estimating the population-level intercept term, the first element of $\bm{\alpha}_{2}$, whereas the order-of-magnitude advantage of NUTS-Kalman shown in Figure \ref{fig:eg1-univariate-latent-ar1-efficiency} stems from Metropolis-within-Gibbs and NUTS-Joint being inefficient for the remaining three elements of $\bm{\alpha}_{2}$, i.e., the mean of the autoregressive coefficient and of the log standard deviations of the measurement noise and process noise. NUTS-Kalman was also more efficient for the variance parameters in $\bm{\Psi}_{2}$.

\begin{figure}[t]
    \centering
    \includegraphics[width=\linewidth]{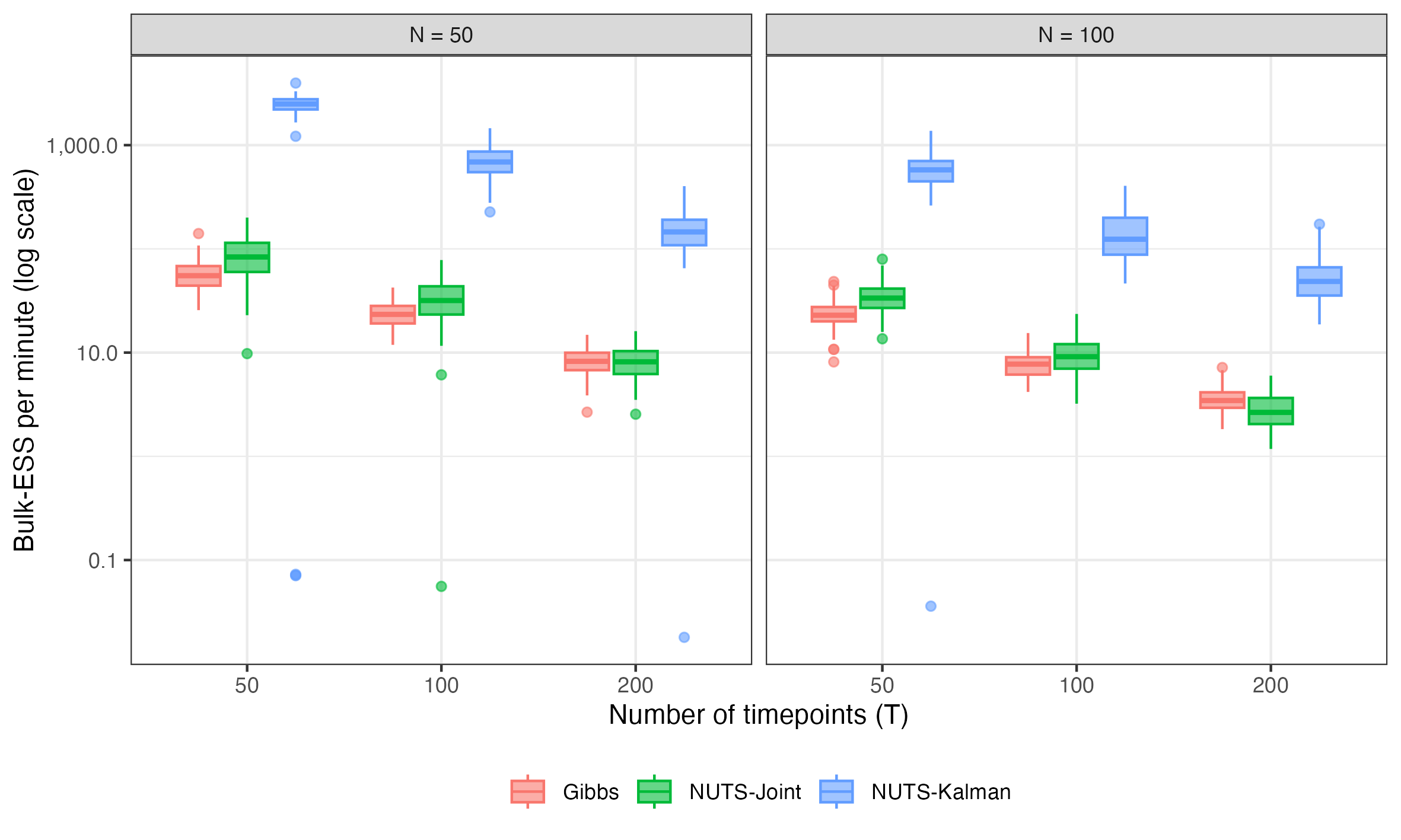}
    \caption{Box plots of efficiency for the three algorithms compared for the multilevel scalar latent AR(1) model}
    \label{fig:eg1-univariate-latent-ar1-efficiency}
\end{figure}

For a practitioner doing Bayesian data analysis, an ESS of around 400 is typically considered a sensible threshold for computing summaries like posterior means \citep{veenman_bayesian_2024}, around 1,000 may be required to obtain reliable posterior intervals \citep{zitzmann_going_2019}, while 10,000 has been suggested as a threshold for high-precision inference \citep{kruschke_doing_2015}. Table \ref{tab:eg1-univariate-latent-ar1-time} shows the projected time required to reach these targets for the three algorithms, for the larger data setting with 100 participants and 200 timepoints.

\begin{table}
    \centering
    \caption{\label{tab:eg1-univariate-latent-ar1-time}Median time (minutes) to reach target bulk-ESS for $N=100$ and $T=200$}
    \centering
    \begin{tabular}[t]{lrrr}
        \toprule
        Method & Target=400 & Target=1,000 & Target=10,000\\
        \midrule
        Gibbs & 115 & 289 & 2899\\
        NUTS-Joint & 150 & 377 & 3771\\
        NUTS-Kalman & 8 & 20 & 206\\
        \bottomrule
    \end{tabular}
\end{table}

\subsection{Multilevel One-Factor Multi-Indicator AR(1) Model}
\label{sec:MultiIndicatorAR1}

We next simulated a model for which the observation $\bm{y}_{it}$ was a vector of size $U=3$, while the underlying latent trait was still scalar. Equation \eqref{eq:ResponseDecomposition} is now $\bm{y}_{it} = \bm{y}_{1,it} + \bm{y}_{2,i}$, there is no between-timepoint model \eqref{eq:BetweenTimepointModel}, and the between-participant model \eqref{eq:BetweenIndividualModel} is
\begin{align*}
    \bm{y}_{2,i}    & =
    \begin{bmatrix}
        \eta_{2,i1} & \eta_{2,i2} & \eta_{2,i3}
    \end{bmatrix}^{T}                                                                             \\
    \bm{\eta}_{2,i} & = \bm{\alpha}_{2} + \bm{\xi}_{2,i}, \quad \bm{\xi}_{2,i} \sim \mathcal{N}(\bm{0}, \bm{\Psi}_{2}),
\end{align*}
where $\bm{\Psi}_{2} = \text{diag}(\tau_{\nu}^{2}, \tau_{\nu}^{2}, \tau_{\nu}^{2}, \tau_{\phi}^{2}, \tau_{\log \sigma}^{2}, \tau_{\log \psi}^{2})$. The within-level model \eqref{eq:WithinLevelModel} is
\begin{align*}
    \bm{y}_{1,it} & =
    \begin{bmatrix}
        1 & \lambda_{1} & \lambda_{2}
    \end{bmatrix}^{T} \eta_{1,it}
    + \bm{\epsilon}_{1,it}, \quad \bm{\epsilon}_{1,it} \sim \mathcal{N}(\bm{0}, \bm{\Sigma}_{2,i})                \\
    \eta_{1,it}   & = \phi_{2,i} \eta_{1,i,t-1} + \xi_{1,it} \quad \xi_{1,it} \sim \mathcal{N}(0,\psi_{2,i}^{2}),
\end{align*}
where $\bm{\Sigma}_{2,i} = \text{diag}(\sigma_{2,i}^{2}, \sigma_{2,i}^{2}, \sigma_{2,i}^{2})$. The first three elements of $\bm{\alpha}_{2}$ are now population means for each of the three measured indicators, while the remaining three elements are related to $\phi_{2,i}$, $\sigma_{2,i}$, and $\psi_{2,i}$ through $\tanh(\cdot)$ and $\exp(\cdot)$ transformations exactly as in the previous section. The factor loadings $\lambda_{1}$ and $\lambda_{2}$ are population-level parameters. We ran the simulations for two data dimensions: one with twice as many participants as there were timepoints per participant, $N=100$ and $T=50$, and one with the opposite, $N=50$ and $T=100$. Based on the pilot runs described in Online Resource 1, the tuning parameters of all algorithms were set to the same values as in Section \ref{sec:LatentScalarAR1}.

Figure \ref{fig:eg2-three-indicator-ar1-effiency} shows the efficiency of the three algorithms, measured in terms of bulk-ESS per minute (left) and tail-ESS per minute (right). In the two simulation settings, NUTS-Kalman was between $2.9$ and $3.3$ times more efficient than Metropolis-within-Gibbs in terms of median bulk-ESS per minute and $1.8$ times more efficient in terms of median tail-ESS per minute. NUTS-Kalman was also between $5.7$ and $8.2$ times more efficient than NUTS-Joint (between $3.9$ and $4.7$ for tail-ESS). Table \ref{tab:eg2-three-indicator-ar1-time} shows the median time to reach a target ESS for each algorithms in the two settings considered. Note that in the $N=100, T=50$ setting there was a single Monte Carlo sample for which NUTS-Kalman was extremely inefficient, as shown by the blue dots in Figure \ref{fig:eg2-three-indicator-ar1-effiency}.

\begin{figure}
    \centering
    \includegraphics[width=\linewidth]{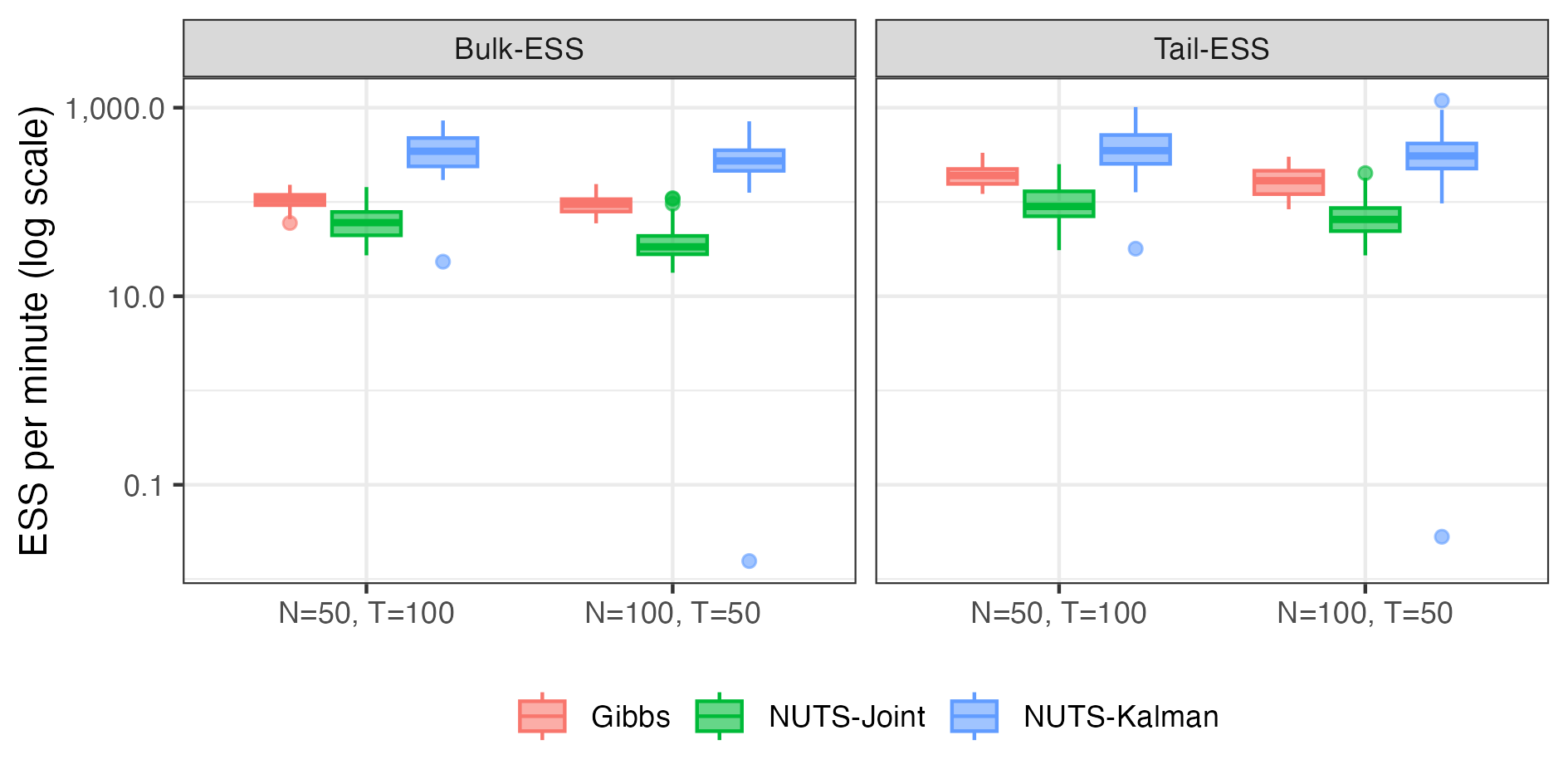}
    \caption{Box plots of efficiency for the three algorithms compared for the multilevel one-factor three-indicator latent AR(1) model}
    \label{fig:eg2-three-indicator-ar1-effiency}
\end{figure}

\begin{table}
\centering
\caption{\label{tab:eg2-three-indicator-ar1-time}Median time (minutes) to reach target bulk-ESS for the three-indicator AR(1) model}
\centering
\begin{tabular}[t]{rrlrrr}
\toprule
N & T & Method & Target=400 & Target=1,000 & Target=10,000\\
\midrule
100 & 50 & Gibbs & 4.2 & 10.6 & 106.0\\
100 & 50 & NUTS-Joint & 12.0 & 29.9 & 299.0\\
100 & 50 & NUTS-Kalman & 1.5 & 3.7 & 36.7\\
50 & 100 & Gibbs & 3.8 & 9.4 & 94.4\\
50 & 100 & NUTS-Joint & 6.6 & 16.6 & 165.7\\
\addlinespace
50 & 100 & NUTS-Kalman & 1.2 & 2.9 & 29.0\\
\bottomrule
\end{tabular}
\end{table}

\subsection{Multilevel Trivariate VAR(1) Model}
\label{sec:TrivariateAR1}

We next extended the model further to include a three-dimensional latent state vector ($V_{1}=3$), with each latent trait being measured by either one indicator ($U=3$), two indicators ($U=6$), or three indicators ($U=9$). Throughout this section we used $N=50$ participants observed at $T=100$ timepoints. Equation \eqref{eq:ResponseDecomposition} is still $\bm{y}_{it} = \bm{y}_{1,it} + \bm{y}_{2,i}$, there is no between-timepoint model \eqref{eq:BetweenTimepointModel}, and the between-participant model \eqref{eq:BetweenIndividualModel} is
\begin{equation*}
    \begin{aligned}
        \bm{y}_{2,i}    & =
        \begin{bmatrix}
            \eta_{2,i1} & \eta_{2,i2} & \dots & \eta_{2,iU}
        \end{bmatrix}^{T}                                                                     \\
        \bm{\eta}_{2,i} & = \bm{\alpha}_{2} + \bm{\xi}_{2,i}, \quad \bm{\xi}_{2,i} \sim \mathcal{N}(\bm{0}, \bm{\Psi}_{2}).
    \end{aligned}
\end{equation*}
The within-level model is
\begin{align*}
    \bm{y}_{1,it}    & = \bm{\Lambda}_{2}  \bm{\eta}_{1,it} + \bm{\epsilon}_{1,it} , \quad \bm{\epsilon}_{1,it} \sim \mathcal{N}(\bm{0}, \bm{\Sigma}_{2,i}) \\
    \bm{\eta}_{1,it} & = \bm{\Phi}_{2,i} \bm{\eta}_{1,i,t-1} + \bm{\xi}_{1,it}, \quad \bm{\xi}_{1,it} \sim \mathcal{N}(\bm{0}, \bm{\Psi}_{2,i}),
\end{align*}
where $\bm{\Lambda}_{2}$ was such that we had an independent cluster model, with each item reflecting only a single latent variable. We ensured stationarity of the simulated datasets by rejecting any $\bm{\Phi}_{2,i}$ with maximum modulus of its eigenvalues larger than $0.95$. We assumed for simplicity that the $U \times U$ matrix $\bm{\Sigma}_{2,i}$ had identical values $\sigma_{2,i}^{2}$ along the diagonal and similarly that the $3 \times 3$ matrix $\bm{\Psi}_{2,i}$ has identical values $\psi_{2,i}^{2}$ along the diagonal. The vector $\bm{\alpha}_{2}$ is now of size $U+11$ and contains the means of the $U$ measured items, the 9 arctanh-transformed elements of $\bm{\Phi}_{2,i}$, and the means of $\log \sigma_{2,i}$ and $\log \psi_{2,i}$. The $(U+11) \times (U+11)$ covariance matrix $\bm{\Psi}_{2}$ contains the corresponding variances along its diagonal, and has zero off-diagonal entries.

Based on the pilot runs described in Online Resource 1, we again set the algorithm parameters identically to Sections \ref{sec:LatentScalarAR1} and \ref{sec:MultiIndicatorAR1}, except that in the case of a single indicator per latent variables, $U=3$, the target acceptance rate during warm-up for NUTS-Joint was $0.99$ to avoid divergent transitions. By default, \textsc{cmdstanr} initializes parameters by drawing from a uniform distribution over $[-2,2]$ on the unconstrained parameter space. This led to divergence into non-stationary regions during the warm-up phase, and for this reason we explicitly initialized the population means of the autoregressive coefficients to $0$ and the log-standard deviations to $0.05$. In contrast, \textsc{JAGS} initializes parameters by drawing from the prior distributions, and since we used somewhat informative priors, this turned out to be sufficient to ensure convergence during burn-in (see Section S3.3 of Online Resource 3 for details).

In our implementation, we projected the $U$ observed indicators onto a $V_{1}=3$ dimensional compressed state space prior to running the Kalman filter, following \citet{jungbackerLikelihoodbasedDynamicFactor2015}. Due to this, we only had to do a $3 \times 3$ matrix inversion in every step rather than $U \times U$, which in the worst case would be $9 \times 9$ (see Section S3.2 of Online Resource 3).

Figure \ref{fig:eg3-trivariate-efficiency} shows the efficiency (left) and maximum $\hat{R}$ (right) for each number of indicator variables. It is clear from the plot that NUTS-Kalman outperforms Metropolis-within-Gibbs in all settings, and that the margin is largest in the low-information $K=1$ case. To be specific, NUTS-Kalman is $9.6$ times more efficient than Metropolis-within-Gibbs in the $K=1$ case, and $1.6$ and $2.0$ times more efficient in the $K=2$ and $K=3$ case, respectively. In the single-indicator setting, Metropolis-within-Gibbs and NUTS-Joint also have poor $\hat{R}$ values, indicating problems with convergence. In contrast, in the two other settings, all algorithms have $\hat{R}$ values very close to 1.

\begin{figure}
    \centering
    \includegraphics[width=\linewidth]{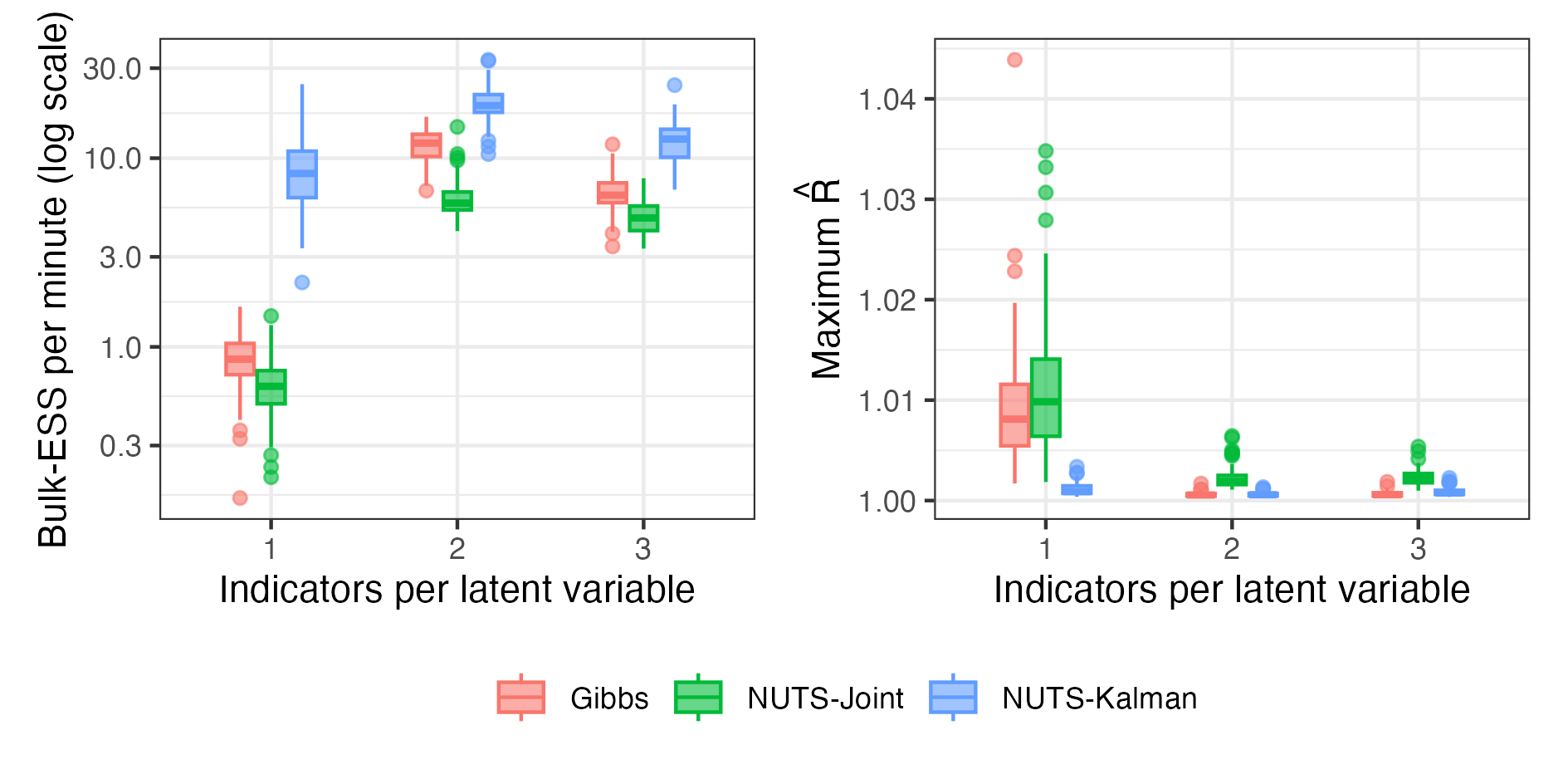}
    \caption{The left box plot shows efficiency plotted against indicators per latent variable and the right plot shows $\hat{R}$ plotted against indicators per latent variable for the trivariate VAR(1) model}
    \label{fig:eg3-trivariate-efficiency}
\end{figure}

\subsection{Multilevel Latent AR(p) Model}
\label{sec:ARp_model}

In the next simulation, we extended the model from Section \ref{sec:LatentScalarAR1} to an AR($p$) model, with the lag $p$ taking values $p=1,2,3,4$. The within-level model is now
\begin{align*}
     & y_{1,it} = \eta_{1,it} + \epsilon_{1,it}, \quad \epsilon_{1,it} \sim \mathcal{N}(0, \sigma_{2,i}^{2})                     \\
     & \eta_{1,it} = \sum_{l=1}^{p}\phi_{2,li} \eta_{1,i,t-l} + \xi_{1,it}, \quad \xi_{1,it} \sim \mathcal{N}(0,\psi_{2,i}^{2}),
\end{align*}
while the remaining components as well as all other aspects related to the simulations were identical to Section \ref{sec:LatentScalarAR1} except for the data-generating parameter values of the autoregression coefficients, which had means $0.4$, $0.2$, $0.1$, and $0.05$ for lags $1,2,3,4$, respectively, and standard deviations $\tau_{\phi} = 0.05$. During data generation, we repeatedly redrew the individual-level coefficients until $\sum_{l=1}^{p} |\phi_{2,li}| < 0.95$, thereby ensuring that each participant's process was stationary. We simulated 100 Monte Carlo samples for each of four settings, with maximum lags $L$ equal to $1,2,3,4$. The number of participants $N=100$ and number of timepoints $T=50$ were fixed. Based on the pilot runs described in Online Resource 1, we set the simulation parameters of NUTS-Kalman and Metropolis-within-Gibbs identically to the previous sections, while for NUTS-Joint we set the target acceptance during warm-up to $0.99$ for all lags and the maximum treedepth to $10$ for $L=1,2,3$ and $15$ for lag $L=4$, to avoid divergent transitions after warm-up.

The left plot in Figure \ref{fig:eg4-latent-arp-efficiency} shows how the efficiency, measured by bulk-ESS per minute, depends on the lag. While NUTS-Kalman remains far more efficient than Metropolis-within-Gibbs across all conditions, its relative advantage decreases as the lag increases, being $61$, $33$, $19$, and $13$ times more efficient for $L=1, 2, 3$, and $4$, respectively. Compared to NUTS-Joint, NUTS-Kalman was $36$, $17$, $11$, and $11$ times more efficient for $L=1,2,3,4$, respectively. The right plot in Figure \ref{fig:eg4-latent-arp-efficiency} shows box plots of the maximum $\hat{R}$ across the parameters of interest: NUTS-Kalman clearly has satisfactory $\hat{R}$ values in all cases, while the other two algorithms occasionally have trouble converging. Note, however, that with a single exception, all algorithms have $\hat{R}$ below the more lenient $1.05$ threshold.

\begin{figure}
    \centering
    \includegraphics[width=\linewidth]{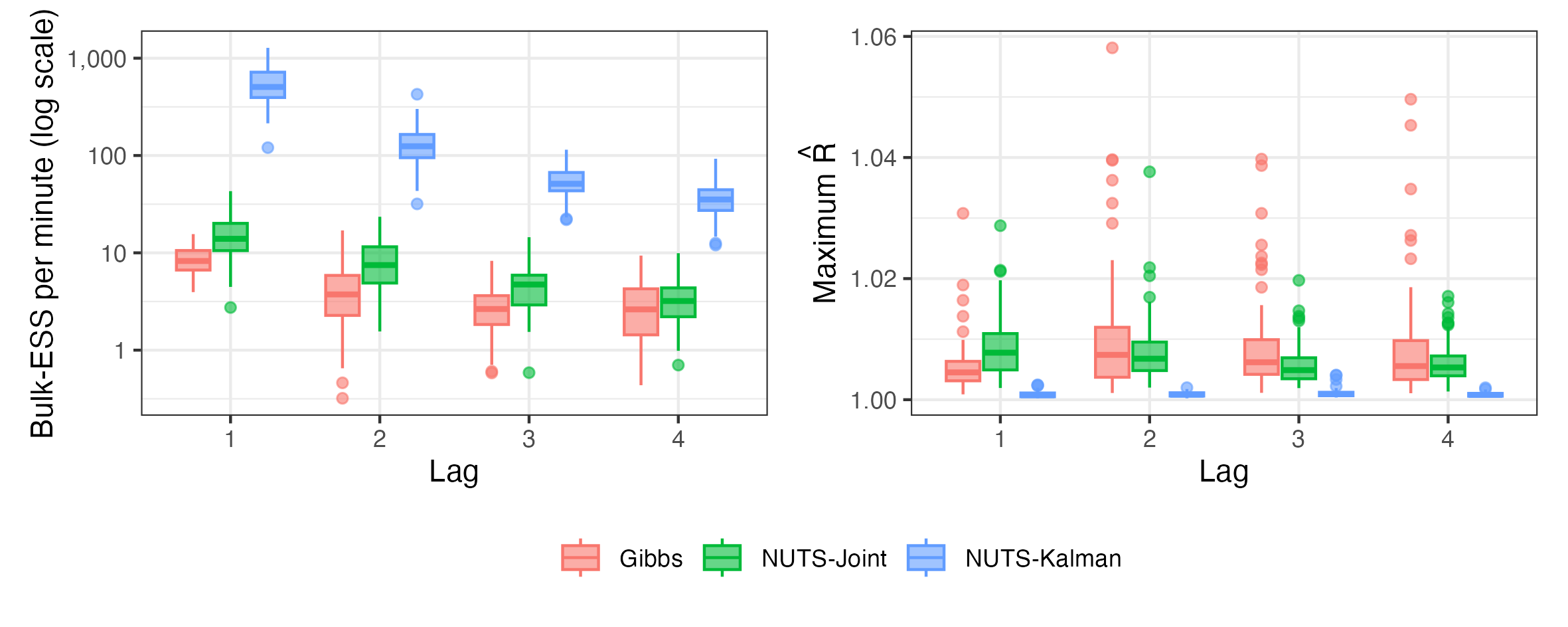}
    \caption{The left box plot shows efficiency plotted against lag, and the right plot shows $\hat{R}$ plotted against lag for the AR(p) model}
    \label{fig:eg4-latent-arp-efficiency}
\end{figure}

\subsection{Cross-Classified VAR(1) Model}
\label{sec:CrossClassifiedVAR1}

We finally considered a model inspired by the one proposed by \citet{mcneish_measurement_2021} for assessing measurement invariance in DSEM. The measurement $\bm{y}_{it}$ was four-dimensional. The first three elements of $\bm{y}_{it}$ are items measuring the latent trait perseverance, and the fourth element is a manifest variable measuring restraint. The response decomposition \eqref{eq:ResponseDecomposition} is $\bm{y}_{it} = \bm{y}_{1,it} + \bm{y}_{2,i} + \bm{y}_{3,t}$ and the between-timepoints model \eqref{eq:BetweenTimepointModel} is
\begin{equation}
    \begin{aligned}
        \bm{y}_{3,t}    & =
        \begin{bmatrix}
            \eta_{3,t1} & \eta_{3,t2} & \eta_{3,t3} & 0
        \end{bmatrix}^{T}                                                       \\
        \bm{\eta}_{3,t} & = \bm{\xi}_{3,t}, \quad \bm{\xi}_{3,t} \sim \mathcal{N}(\bm{0}, \bm{\Psi}_{3}),
    \end{aligned}
    \label{eq:eg5-between-timepoint-model}
\end{equation}
where $\bm{\eta}_{3,t}$ has size $V_{3} = 7$. The first four element of $\bm{\eta}_{3,t}$ are related to the means of the measurements and the last three elements are related to the factor loadings in \eqref{eq:eg5-within-model} below. The covariance matrix is $\bm{\Psi}_{3} = \text{diag}(\psi_{3,1},\psi_{3,2},\psi_{3,3},0, \psi_{3,5}, \psi_{3,6},\psi_{3,7})$, where the zero in the middle, along with the zero in the first line of \eqref{eq:eg5-between-timepoint-model}, implies that the intercept of restraint does not vary between timepoints. The between-participants model \eqref{eq:BetweenIndividualModel} is
\begin{equation}
    \begin{aligned}
        \bm{y}_{2,i}    & =
        \begin{bmatrix}
            \eta_{2,i1} & \eta_{2,i2} & \eta_{2,i3} & \eta_{2,i4}
        \end{bmatrix}^{T}                                                               \\
        \bm{\eta}_{2,i} & = \bm{\alpha}_{2} + \bm{\xi}_{2,i}, \quad \bm{\xi}_{2,i} \sim \mathcal{N}(\bm{0}, \bm{\Psi}_{2}),
    \end{aligned}
    \label{eq:eg5-between-model}
\end{equation}
where $\bm{\eta}_{2,i}$ has size $V_{2}=13$. The first four elements of $\bm{\eta}_{2,i}$ are related to the intercepts of the measured items. Note that this implies that the intercept of restraint is allowed to vary between participants. The covariance matrix is $\bm{\Psi}_{2} = \text{diag}(\psi_{2,1},\dots,\psi_{2,13})$. The within-level model is
\begin{equation}
    \begin{aligned}
        \bm{y}_{1,it}    & =
        \begin{bmatrix}
            \lambda_{1,it} & 0 \\
            \lambda_{2,it} & 0 \\
            \lambda_{3,it} & 0 \\
            0              & 1
        \end{bmatrix}
        \bm{\eta}_{1,it}
        + \bm{\epsilon}_{1,it} , \quad \bm{\epsilon}_{1,it} \sim \mathcal{N}(\bm{0}, \bm{\Sigma}_{1,i})                                              \\
        \bm{\eta}_{1,it} & = \bm{\Phi}_{1,i} \bm{\eta}_{1,i,t-1} + \bm{\xi}_{1,it}, \quad \bm{\xi}_{1,it} \sim \mathcal{N}(\bm{0}, \bm{\Psi}_{1,i}),
    \end{aligned}
    \label{eq:eg5-within-model}
\end{equation}
where $\bm{\Sigma}_{1,i} = \text{diag}(\sigma_{1,i}^{2}, \sigma_{1,i}^{2}, \sigma_{1,i}^{2}, 0)$, $\bm{\Psi}_{1,i} = \text{diag}(1, \psi_{1,i}^{2})$, and the autoregression matrix $\bm{\Phi}_{1,i}$ is allowed to vary between participants. The factor loadings in \eqref{eq:eg5-within-model} are given by $\lambda_{1,it} = \eta_{3,t5} + \eta_{2,i5}$, $\lambda_{2,it} = \eta_{3,t6} + \eta_{2,i6}$, and $\lambda_{3,it} = \eta_{3,t7} + \eta_{2,i7}$. To avoid sign indeterminacy, we constrained the elements of $\bm{\alpha}_{2}$ in \eqref{eq:eg5-between-model} corresponding to the population means of the factor loadings to be non-negative.

In this model, measurement invariance would imply that $\psi_{3,5}=\psi_{3,6}=\psi_{3,7}=0$. Note that the underlying assumption in the case of non-invariance is that the factor loadings for two adjacent timepoints $t$ and $t+1$ are no more similar than the factor loadings for two distant measurements, e.g., $t=1$ and $t=T$. If there is reason to suspect that the loadings have a time trend, a more realistic model could include a covariate $t$ in the second line of \eqref{eq:eg5-between-timepoint-model}. However, since our goal in this paper is to compare algorithmic efficiency, we chose to use the approach which seems to be standard for cross-classified DSEM \citep{asparouhov_dynamic_2018,mcneish_measurement_2021}.

We simulated 100 Monte Carlo samples with $N=100$ and $T=150$. Based on the pilot runs described in Online Resource 1, we set the simulation parameters of NUTS-Kalman identically to the previous sections. For NUTS-Joint we set the target acceptance during warm-up to $0.95$ and the maximum treedepth to $10$. Again, for both NUTS based algorithms we used a warm-up of 1,000 iterations and ran 10,000 post-warm-up iterations in each of four parallel chains, yielding 40,000 post-warm-up iterations. For Metropolis-within-Gibbs we discarded the first 20,000 iterations after the 2,000 iteration long adaptation phase as burn-in, based on inspection of the trace plots from the pilot, all of which can be found in our OSF repository. We then ran 100,000 post-burn-in iterations in each of four parallel chains, again yielding 400,000 post-burn-in samples. We implemented a Kalman filter processing one element of $\bm{y}_{1,it}$ at a time \citep{koopmanFastFilteringSmoothing2000}. 

NUTS-Kalman and Gibbs had maximum $\hat{R}$ values below $1.01$ in the majority of samples, whereas NUTS-Joint had a large proportion of runs with poor convergence and $\hat{R}>1.10$ (see Figure S15 of Online Resource 4). In other words, the model considered in this section was particularly challenging for NUTS-Joint, and this algorithm should not be recommended regardless of its efficiency. Figure \ref{fig:eg5-efficiency} shows box plots of efficiency in terms of bulk-ESS and tail-ESS for the three algorithms. Using median values across simulations, NUTS-Kalman was $2.9$ times more efficient than Metropolis-within-Gibbs in terms of bulk-ESS per minute and $3.3$ times more efficient in terms of tail-ESS per minute.

\begin{figure}
    \centering
    \includegraphics[width=\linewidth]{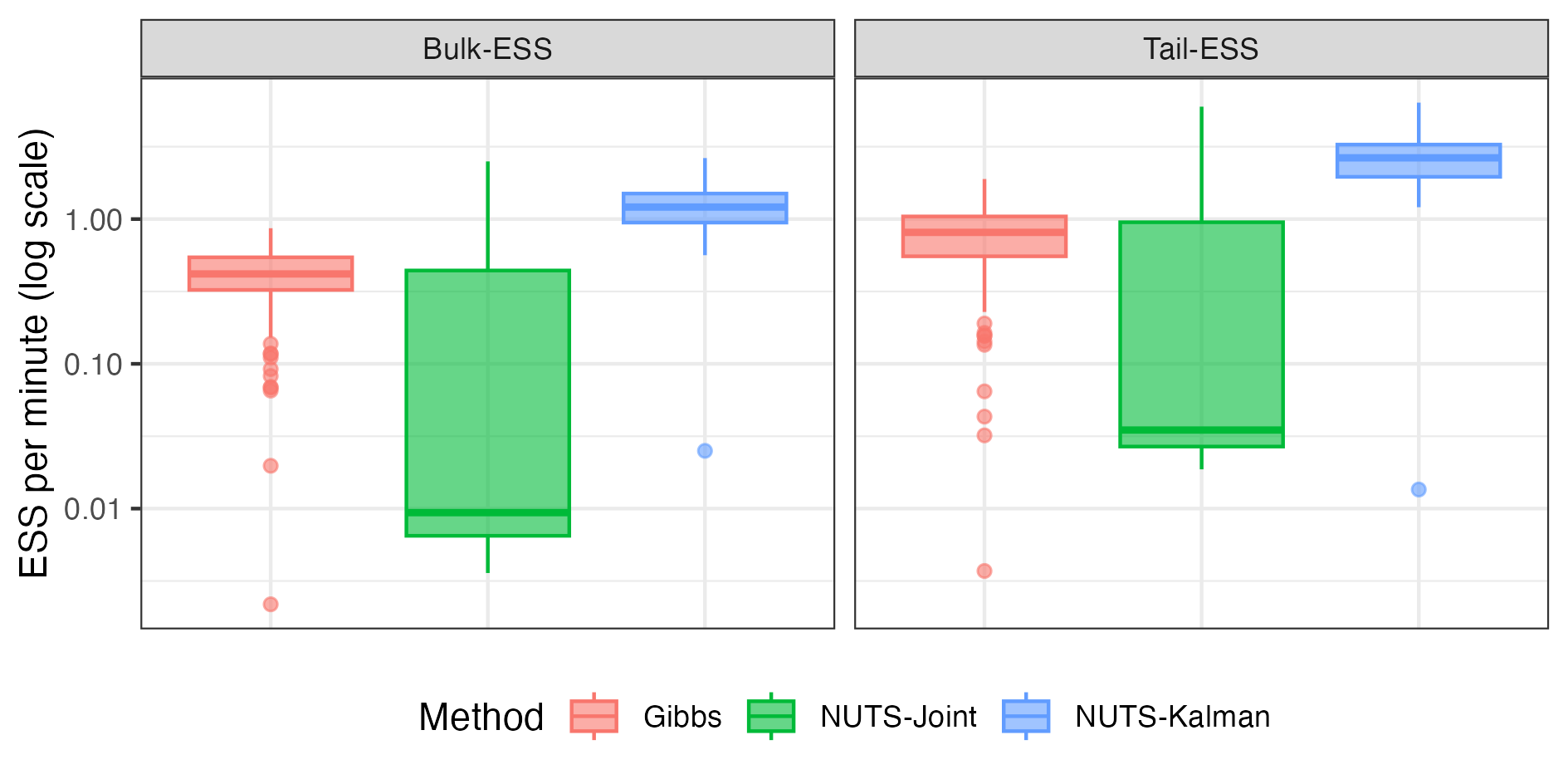}
    \caption{Efficiency measured in bulk-ESS per minute (left) and tail-ESS per minute (right) for the cross-classified VAR(1) model}
    \label{fig:eg5-efficiency}
\end{figure}

\section{Example Application}
\label{sec:Application}

The simulation experiments compared the algorithms under controlled conditions in which the fitted model was identical to the data-generating process. To examine how they behave on real data, where the model is necessarily an approximation, we analyzed the affect time series collected by \citet{rowlandMindfulnessAffectNetworkDensity2020}, openly available as dataset \textsc{0006\_rowland} of the openESM database \citep{siepeIntroducingOpenESMDatabase2025}. All analyses reported in this section used the same computing cluster as the simulations, with four parallel chains for each algorithm.

\subsection{Data and Model}
\label{sec:ApplicationModel}

In the study of \citet{rowlandMindfulnessAffectNetworkDensity2020}, 125 undergraduate students\footnote{The study recruited 137 students, of which 11 withdrew and one was removed due to noncompliance \citep{rowlandMindfulnessAffectNetworkDensity2020}.} reported their momentary affect six times per day for 40 consecutive days, prompted at random times during the day with consecutive prompts separated by 45 to 200 minutes (mean 103 minutes). At each prompt, participants rated eight affective states on a visual analogue scale from 0 (\textit{not at all}) to 100 (\textit{very much}): happy, excited, relaxed, and satisfied, reflecting positive affect, and angry, anxious, depressed, and sad, reflecting negative affect. The harmonized data form a complete grid of $T = 40 \times 6 = 240$ potential timepoints per participant, of which 28.1\% are missing and per-participant missingness ranges from 6\% to 68\%. For each participant-timepoint, either all eight responses were missing or none were missing. These data are typical of the intensive longitudinal designs DSEM is intended for, and they are messy in ways the simulated data were not: the responses are bounded with strong floor effects (25--34\% of responses on the negative affect items are exactly 0), the within-day prompt spacing is irregular, and consecutive timepoints are separated by anywhere from 45 minutes to an overnight gap of at least 14 hours. Descriptive plots of the data are shown in Online Resource 5.

Following common applied practice---see,  \citet{haslbeckMultimodalitySkewnessEmotion2023} for a discussion of skewness, \citet{dehaan-rietdijkDiscreteVsContinuousTime2017} for a discussion of discrete versus continuous modeling, and \citet{berkhoutLetSleepingDogs2025} for a discussion of the night gap---we fit a discrete-time model treating the beeps as an equally spaced series with conditionally Gaussian responses. Since all three algorithms target exactly the same posterior distribution, this misspecification does not affect the validity of the comparison, but it does produce a challenging posterior geometry.

We fitted a bivariate latent VAR(1) DSEM, using the response decomposition \eqref{eq:ResponseDecomposition} with $\bm{y}_{it} = \bm{y}_{1,it} + \bm{y}_{2,i}$ and no between-timepoint level. The within-level model \eqref{eq:WithinLevelModel} has $U = 8$ indicators, rescaled to $0$--$10$ for numerical stability, measuring $V_{1} = 2$ latent factors,
\begin{equation}
    \begin{aligned}
        \bm{y}_{1,it}    & = \bm{\Lambda} \bm{\eta}_{1,it} + \bm{\epsilon}_{1,it}, \quad \bm{\epsilon}_{1,it} \sim \mathcal{N}\left(\bm{0}, \bm{\Sigma}_{1,i}\right)                    \\
        \bm{\eta}_{1,it} & = \bm{\Phi}_{i} \bm{\eta}_{1,i,t-1} + \bm{\xi}_{1,it}, \quad \bm{\xi}_{1,it} \sim \mathcal{N}\left(\bm{0}, \bm{\Psi}_{1,i}\right),
    \end{aligned}
    \label{eq:ApplicationModel}
\end{equation}
where the first four items load on the positive affect factor and the last four on the negative affect factor in an independent cluster structure, with the loading of the first item in each cluster fixed to 1 and the remaining six loadings free, constrained positive, and invariant across participants. The residual covariance matrix is $\bm{\Sigma}_{1,i} = \text{diag}(\sigma_{1i}^{2}, \dots, \sigma_{8i}^{2})$ with $\sigma_{ui} = \exp(\gamma_{u} + \delta_{i})$, combining item-specific baseline scales $\gamma_{u}$ with a participant-specific shift $\delta_{i}$. The innovation covariance matrix
\begin{equation*}
    \bm{\Psi}_{1,i} =
    \begin{bmatrix}
        \psi_{1i}^{2} & \rho_{i} \psi_{1i} \psi_{2i} \\
        \rho_{i} \psi_{1i} \psi_{2i} & \psi_{2i}^{2}
    \end{bmatrix}
\end{equation*}
allows correlated positive and negative affect innovations. The between-participant latent vector $\bm{\eta}_{2,i}$ thus has dimension $V_{2} = 16$, comprising the eight item intercepts $\nu_{ui}$, the four elements of $\bm{\Phi}_{i}$ on the $\text{arctanh}$ scale, $\log \psi_{1i}$ and $\log \psi_{2i}$, $\text{arctanh}(\rho_{i})$, and $\delta_{i}$, each normally distributed around its population mean with its own variance, exactly as in the models of Section \ref{sec:Simulations}. Priors follow the conventions of the simulation experiments: $\mathcal{N}(0, 5^{2})$ for the population-level item intercepts, $\mathcal{N}(0, 1)$ for the population means of the transformed dynamic parameters and for $\gamma_{u}$, half-Cauchy priors with scale 2 for the random-intercept standard deviations and scale 1 for the standard deviations of the transformed autoregressive parameters, half-normal $\mathcal{N}_{+}(0, 1)$ priors for the remaining random-effect standard deviations, $\mathcal{N}(1, 2^{2})$ constrained positive for the free loadings, and $\bm{\eta}_{1,i1} \sim \mathcal{N}(\bm{0}, \bm{I}_{2})$ for the initial state.

The three algorithms were implemented exactly as in Section \ref{sec:Simulations}, with identical priors, constraints, and non-centered parameterizations, and hence target the same posterior. For NUTS-Kalman we collapsed the eight indicators onto the two factors following \citet{jungbackerLikelihoodbasedDynamicFactor2015}, here extended to item-specific residual variances, so each filter step requires only $2 \times 2$ operations. Missing beeps were handled by prediction-only steps in the Kalman filter, marginalizing them exactly. For NUTS-Joint the measurement densities at missing beeps are omitted while the latent states are still sampled on the full grid. The Metropolis-within-Gibbs implementation follows \citet[p.~365]{asparouhov_dynamic_2018} and treats the missing responses by data augmentation, sampling the $8 \times 8{,}430 = 67{,}440$ missing values as additional nodes in every iteration. All three implementations still target the same posterior for the model parameters.

\subsection{Results}
\label{sec:ApplicationResults}

For the production runs, NUTS-Kalman and NUTS-Joint were run for 1{,}000 warm-up and 2{,}000 sampling iterations in each of four chains. Based on initial pilot runs, both used a maximum tree depth of 10, but NUTS-Joint had the target acceptance during warm-up set to $0.99$ while NUTS-Kalman had it set to $0.80$. For Metropolis-within-Gibbs, we ran as many iterations as we could within the time limit of the cluster. To this end, after an adaptation phase of 2,000 iterations and discarding the first 5,000 iterations as burn-in, we saved the output after 500 iterations were completed in each of the four chains. After saving, we assessed whether an additional 500 iterations could be safely run within an 88 hour time limit; if yes, 500 more iterations were run, and if no, the job was finished. This strategy yielded 44,000 post-burn-in iterations in each chain and hence a total of 176,000 posterior samples. Total wall times were $22$ hours for NUTS-Kalman, 62 hours for NUTS-Joint, and 87 hours for Metropolis-within-Gibbs. Convergence according to $\hat{R} < 1.01$ over the 45 population-level parameters of interest was reached by NUTS-Kalman, but neither Metropolis-within-Gibbs (max $\hat{R} = 1.36$) nor NUTS-Joint (max $\hat{R} = 2.37$). Despite the troubling convergence diagnostics, the posterior means and standard deviations from the three algorithms agreed closely, consistent with all three targeting the same posterior. The correlation with posterior means of NUTS-Kalman was $1.000$ for both Metropolis-within-Gibbs and NUTS-Joint, while the correlation with posterior standard deviations was $0.991$ and $0.927$, respectively, consistent with the very poor diagnostics of NUTS-Joint. Figures S3 and S4 of Online Resource 5 show further details.

\begin{figure}
    \centering
    \includegraphics[width=\linewidth]{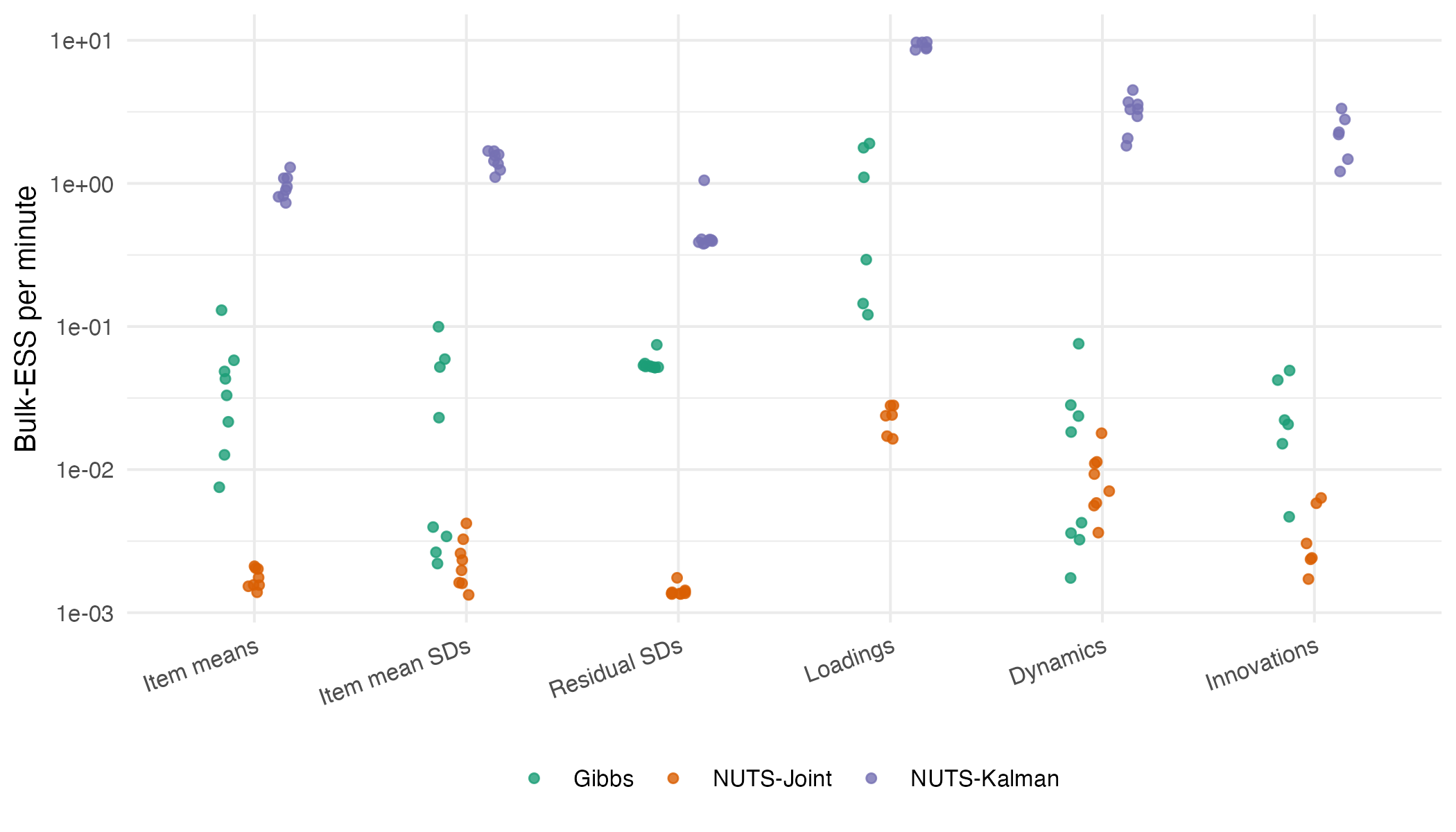}
    \caption{Sampling efficiency (bulk-ESS per minute, log scale) for the 45 population-level parameters of the application model, by parameter group and algorithm}
    \label{fig:application-efficiency}
\end{figure}

Figure \ref{fig:application-efficiency} shows the efficiency of the three algorithms for the parameters of interest. Compared to the simulation experiments, the efficiency gain of NUTS-Kalman was much more pronounced on real data, likely due to the presence of missing values. In terms of bulk-ESS per time unit, NUTS-Kalman was 217 times more efficient than Metropolis-within-Gibbs and 284 times more efficient than NUTS-Joint. Table \ref{tab:application-time} translates this into the time required to reach standard ESS targets \citep{veenman_bayesian_2024, zitzmann_going_2019, kruschke_doing_2015}: NUTS-Kalman is projected to reach ESS of 1{,}000 for all parameters of interest in 44 hours, whereas the other two algorithms are projected to need more than a year. The numbers in Table \ref{tab:application-time} should be interpreted with caution, however, since they are projections based on the wall time and the minimum bulk-ESS: the latter was 492 for NUTS-Kalman over the 8,000 post-warm-up iterations but only $9.2$ for Metropolis-within-Gibbs over 176,000 iterations and $4.9$ for NUTS-Joint over 8,000 iterations. Despite this uncertainty, the results clearly convey that among the three candidates, NUTS-Kalman is the only practically feasible algorithm for this application.

\begin{table}
    \centering
    \caption{\label{tab:application-time}Time (hours) to reach target bulk-ESS for the application model, based on the minimum ESS rate over the 45 population-level parameters of interest. Note that the numbers for Gibbs and NUTS-Joint are very uncertain projections}
    \centering
    \begin{tabular}[t]{lrrr}
        \toprule
        Method & Target=400 & Target=1,000 & Target=10,000\\
        \midrule
            Gibbs & 3814 & 9535 & 95346\\
            NUTS-Joint & 5006 & 12516 & 125158\\
            NUTS-Kalman & 18 & 44 & 440\\            
        \bottomrule
    \end{tabular}
\end{table}

Substantively, the population-level estimates from NUTS-Kalman were in line with what is known about momentary affect dynamics: moderate inertia of both factors (posterior means $0.627$ for positive and $0.800$ for negative affect), small cross-lagged effects ($0.09$ and $0.10$), and a substantial negative innovation correlation (posterior mean $-0.93$, 95\% credible interval $[-0.95, -0.90]$). We emphasize that, given the floor effects and the equal-interval treatment of the beep grid, these parameters should be interpreted as rough estimates from a misspecified model.

\begin{figure}
    \centering
    \includegraphics[width=\linewidth]{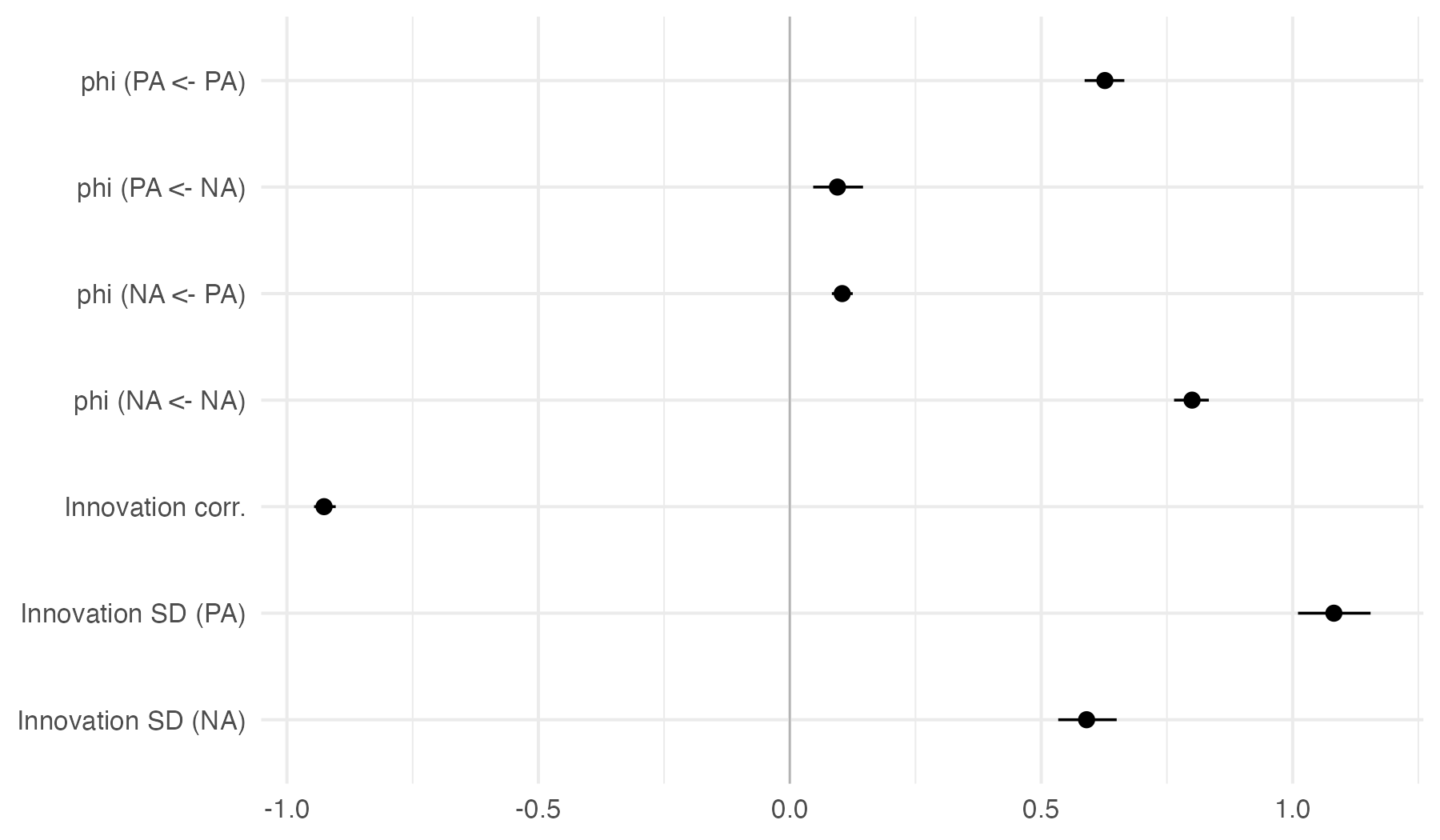}
    \caption{Posterior means and 95\% credible intervals for main parameters of interest in the application example. NA denotes negative affect and PA denotes positive affect. 'phi (PA <- NA)' is the cross-lagged effect of NA on PA, and analogously for the other rows}
    \label{fig:application-estimates}
\end{figure}

\section{Discussion}
\label{sec:Discussion}

In this paper, we have shown that the within-level model of a DSEM can be reformulated as a state space model, and proposed the NUTS-Kalman algorithm which utilizes this reformulation to replace sampling of within-level latent variables with matrix computation using the Kalman filter. With this algorithm, the number of parameters to be sampled per MCMC iteration is reduced from $\mathcal{O}(N \cdot T)$ to $\mathcal{O}(N+T)$. The simulation experiments for five different model scenarios confirmed that the theoretical benefits of NUTS-Kalman also translate to improved efficiency in practice.

The most closely related previous work is \citet{Driver2018}, who estimated a hierarchical continuous-time dynamic model \citep{oudContinuousTimeState2000} by combining a continuous-discrete Kalman filter \citep{kalman_new_1961} with Hamiltonian Monte Carlo in \textsc{Stan}. The filtering engine is thus similar to ours, but the mapping problem is not. Their state vector contains the latent processes $\bm{\eta}(t)$ only, measured contemporaneously, and the discrete-time transition and innovation matrices follow directly from the matrix exponential $\exp({\bm{A} \Delta t})$; no state augmentation is required. The DSEM framework of \citet{asparouhov_dynamic_2018}, in contrast, admits arbitrary lags $L \geq 1$ in both the measurement and the dynamic equation, regressions on lagged manifest variables, and contemporaneous structural regressions. Obtaining a Markovian representation therefore requires an augmented state stacking the lagged histories of both the latent and the manifest variables (Theorem \ref{th:Theorem1}), which is our main theoretical contribution. DSEM further adds a between-timepoint level with no counterpart in hierarchical CTSEM, so the parameters remaining after marginalization differ: we sample $\bm{\Theta}$, $\bm{\eta}_{2}$, and $\bm{\eta}_{3}$, whereas \citet{Driver2018} sample population parameters and subject-specific deviations only.

NUTS-Kalman is particularly useful when the observed indicators are relatively noisy measurements of the underlying latent variables. We can see this in Figure \ref{fig:eg3-trivariate-efficiency}, which shows the results for the trivariate VAR(1) model considered in Section \ref{sec:TrivariateAR1}. There, the advantage of NUTS-Kalman over Metropolis-within-Gibbs was $9.6$ times with a single indicator per latent variable, but only $1.6$ and $2.0$ times with two and three indicators. Similarly, in the model considered in Section \ref{sec:LatentScalarAR1} with a single latent variable measured by a single indicator the relative advantage of NUTS-Kalman was much more pronounced than in the model of Section \ref{sec:MultiIndicatorAR1} with three indicators measuring a single latent variable. A statistical interpretation of this is as follows: when the indicators precisely measure the latent variables, their conditional autocorrelation given the data is relatively low, and the Gibbs sampler loses comparatively little by updating them one at a time. In contrast, when the indicators yield imprecise information about the latent variables, conditioning on the data does not reduce autocorrelation much. This geometry is challenging for the Gibbs sampler \citep{parkImprovingGibbsSampler2022}, and is the prototypical example in which Hamiltonian Monte Carlo excels \citep{hoffman_no_2014}. It does not follow, however, that Hamiltonian Monte Carlo alone is sufficient: NUTS-Joint operates on the same geometry with the same sampler, and was the least efficient of the three algorithms in the more complex settings.

The example application in Section \ref{sec:Application} shows a case in which both our Metropolis-within-Gibbs implementation and NUTS-Joint were so inefficient that we conclude that they are not practically feasible. In terms of sampling efficiency they were 217 and 284 times less efficient than NUTS-Kalman. In addition to the challenging posterior geometry, we hypothesize that the main driving factor behind this is the presence of 28.1\% missing responses, which the Kalman filter handles very efficiently, in contrast to the other two algorithms which rely on brute-force sampling.

We have only focused on normally distributed responses, but categorical variables can also be handled with a latent response formulation \citep{albertBayesianAnalysisBinary1993}. However, the latent responses cannot be integrated out by the Kalman filter, leaving the number of parameters to be sampled back at $\mathcal{O}(N\cdot T)$. With a non-identity link function (e.g., probit, logit), the within-level model \eqref{eq:WithinLevelModel} becomes non-linear. By reformulating it as a non-linear state space model, approximate filtering methods \citep[Ch. 10]{durbin_time_2012} or Laplace approximations \citep{margossianHamiltonianMonteCarlo2020} could be used to analytically integrate over all latent variables. Such an algorithm would not target the true posterior, however. Particle filters, on the other hand, are widely used for sampling from non-linear state space models \citep{naessethElementsSequentialMonte2019}, but cannot be combined with Hamiltonian Monte Carlo due to their stochastic nature \citep{chenStochasticGradientHamiltonian2014}. Pseudo-marginal Hamiltonian Monte Carlo \citep{alenlovPseudoMarginalHamiltonianMonte2021} might be an alternative. We leave this as an interesting problem for further study.

Our framework also supports DSEMs with certain types of nonlinear terms. Theorem \ref{th:Theorem1} does not require the model to be linear, but that it be linear and Gaussian in the within-level latent variables conditionally on everything else that is sampled. Since $\bm{\Theta}$, $\bm{\eta}_{2}$, and $\bm{\eta}_{3}$ are drawn by NUTS and hence fixed when the filter runs, they may enter the system matrices arbitrarily nonlinearly. The models of Section \ref{sec:Simulations} already rely on this, as the autoregressive coefficients and the noise standard deviations are nonlinear transformations of between-level latent variables. It also means that semiparametric extensions \citep{sorensenModelingCyclesTrends2025, mccormickSemiparametricTimeVaryingParameters2026} remain within the framework. In contrast, interactions between within-level latent variables, as studied by \citet{holtmannModelingWithinlevelLatent2025} fall outside the scope of the present algorithm, and estimating them requires sampling the within-level latent variables.

Another problem of high practical relevance is to scale DSEM to larger data. For example, \citet{juddTrainingSpatialCognition2021} analyzed data from a mathematical learning app in which 17,648 children had completed on average 5,077 trials each. In these cases, approximate methods like Pathfinder \citep{zhangPathfinderParallelQuasiNewton2022} and automatic differentiation variational inference \citep{kucukelbirAutomaticDifferentiationVariational2017} will likely be necessary to handle the high-dimensional parameter space. However, the efficiency of these approximations still depends on the underlying Kalman filter evaluation. For large $N$, frameworks that utilize just-in-time (JIT) compiled vectorization (e.g., \textsc{NumPyro} \citep{phanComposableEffectsFlexible2019}) can broadcast the filter across participants simultaneously. For large $T$, parallel Kalman filter algorithms \citep{sarkkaTemporalParallelizationBayesian2021} that scale logarithmically with the number of timepoints by utilizing graphical processing units offer a promising path to making models of this magnitude computationally feasible.

Finally, the NUTS-Kalman algorithm is not limited to a single software environment; NUTS \citep{hoffman_no_2014} is now the standard in major probabilistic programming frameworks beyond \textsc{Stan}, including \textsc{PyMC} \citep{abril-plaPyMCModernComprehensive2023}, \textsc{NumPyro}, and \text{Turing.jl} \citep{fjeldeTuringjlGeneralPurposeProbabilistic2025}. This ensures that our proposed method can be implemented within the ecosystem that best suits a researcher's needs.

\section*{Data Availability}
Complete source code for reproducing all results in this study is available in our OSF
repository, \url{https://osf.io/p754m}.

\printbibliography

\appendix
\numberwithin{equation}{section}

\section{Proof of Theorem \ref{th:Theorem1}}
\label{app:Proof}

\begin{definition}[Strictly Lagged Polynomial Matrices]
    \label{def:strictly-lagged}
    For any polynomial matrix $\bm{P}(L) = \sum_{l=0}^{L} \bm{P}_{l} L^{l}$, we define the strictly lagged polynomial matrix $\bm{P}^{*}(L)$ as the polynomial obtained by removing the contemporaneous (zero-lag) coefficient $\bm{P}_{0}$:
    \begin{equation}
        \bm{P}^{*}(L) = \bm{P}(L) - \bm{P}_{0} = \sum_{l=1}^{L} \bm{P}_{l} L^{l}.
    \end{equation}
\end{definition}

\subsection{Reformulating the Within-Level Model}

Isolating the lag-$0$ components on the left-hand side, the coupled within-level model \eqref{eq:WithinLevelModel} can be rewritten as
\begin{align}
    \label{eq:y_separated}
    \left(\bm{I} - \bm{R}_{0it}\right) \bm{y}_{1,it} - \bm{\Lambda}_{1,0it} \bm{\eta}_{1,it} & = \bm{\nu}_{1,it} + \bm{K}_{1,it} \bm{X}_{1,it} + \bm{\Lambda}^{*}_{1,it}(L) \bm{\eta}_{1,it} + \bm{R}^{*}_{it}(L) \bm{y}_{1,it} + \bm{\epsilon}_{1,it} \\
    \label{eq:eta_separated}
    \left(\bm{I} - \bm{B}_{1,0it}\right) \bm{\eta}_{1,it} - \bm{Q}_{0it} \bm{y}_{1,it}       & = \bm{\alpha}_{1,it} + \bm{\Gamma}_{1,it} \bm{X}_{1,it} + \bm{B}^{*}_{1,it}(L) \bm{\eta}_{1,it} + \bm{Q}^{*}_{it}(L) \bm{y}_{1,it} + \bm{\xi}_{1,it}.
\end{align}
Let $\bm{A}_{0it} = (\bm{I} - \bm{R}_{0it})^{-1}$ and $\bm{\Xi}_{0it} = (\bm{I} - \bm{B}_{1,0it})^{-1}$, and note that these matrix inverses exist because we have assumed that $\bm{R}_{0it}$ and $\bm{B}_{1,0it}$ are strictly lower triangular. We assemble the left-hand sides of \eqref{eq:y_separated} and \eqref{eq:eta_separated} into a $2 \times 2$ block matrix structure and compute its inverse to explicitly isolate the contemporaneous states,
\begin{equation}
    \label{eq:JointSystem}
    \begin{bmatrix} \bm{\eta}_{1,it} \\ \bm{y}_{1,it} \end{bmatrix}
    =
    \begin{bmatrix} \bm{\Xi}_{0it}^{-1} & -\bm{Q}_{0it} \\ -\bm{\Lambda}_{1,0it} & \bm{A}_{0it}^{-1} \end{bmatrix}^{-1}
    \begin{bmatrix} \text{RHS}_{\eta, it} \\ \text{RHS}_{y, it} \end{bmatrix},
\end{equation}
where $\text{RHS}_{\eta, it}$ and $\text{RHS}_{y, it}$ represent the respective right-hand sides of \eqref{eq:eta_separated} and \eqref{eq:y_separated}. Block-matrix inversion via the Schur complement yields
\begin{equation*}
    \begin{bmatrix}
        \bm{\Xi}_{0it}^{-1}   & -\bm{Q}_{0it}     \\
        -\bm{\Lambda}_{1,0it} & \bm{A}_{0it}^{-1}
    \end{bmatrix}^{-1}
    = \begin{bmatrix}
        \bm{M}_{1,it}                                     & \bm{M}_{1,it} \bm{Q}_{0it} \bm{A}_{0it} \\
        \bm{N}_{1,it} \bm{\Lambda}_{1,0it} \bm{\Xi}_{0it} & \bm{N}_{1,it}
    \end{bmatrix},
\end{equation*}
where
\begin{align}
    \bm{M}_{1,it} & = \left\{ \bm{I} - \bm{\Xi}_{0it} \bm{Q}_{0it} \bm{A}_{0it} \bm{\Lambda}_{1,0it} \right\}^{-1} \bm{\Xi}_{0it} \label{eq:M1it} \\
    \bm{N}_{1,it} & = \left\{ \bm{I} - \bm{A}_{0it} \bm{\Lambda}_{1,0it} \bm{\Xi}_{0it} \bm{Q}_{0it}  \right\}^{-1} \bm{A}_{0it}.
    \label{eq:N1it}
\end{align}

\subsection{Finding Terms in State Space Representation}

The state-space transition \eqref{eq:StateSpaceModelTheorem1} maps variables at $t$ to $t+1$. By advancing the indices of our decoupled system \eqref{eq:JointSystem} to $t+1$, we obtain the exact definitions for the transition parameters.

First, distributing the block-inverse across the historical states in the right-hand side directly yields the composite polynomials. For the structural update equation ($\bm{\eta}_{1,i,t+1}$), we identify
\begin{equation}
    \begin{aligned}
        \bm{\mathcal{P}}^{\eta}_{i,t+1}(L) & = \bm{B}^{*}_{1,i,t+1}(L) + \bm{Q}_{0,i,t+1} \bm{A}_{0,i,t+1} \bm{\Lambda}^{*}_{1,i,t+1}(L) \\
        \bm{\mathcal{P}}^{y}_{i,t+1}(L)    & = \bm{Q}^{*}_{i,t+1}(L) + \bm{Q}_{0,i,t+1} \bm{A}_{0,i,t+1} \bm{R}^{*}_{i,t+1}(L).
    \end{aligned}
    \label{eq:P_poly}
\end{equation}
For the measurement update equation ($\bm{y}_{1,i,t+1}$), grouping the lagged state elements identifies
\begin{equation}
    \begin{aligned}
        \bm{\mathcal{Q}}^{\eta}_{i,t+1}(L) & = \bm{\Lambda}^{*}_{1,i,t+1}(L) + \bm{\Lambda}_{1,0,i,t+1} \bm{\Xi}_{0,i,t+1} \bm{B}^{*}_{1,i,t+1}(L) \\
        \bm{\mathcal{Q}}^{y}_{i,t+1}(L)    & = \bm{R}^{*}_{i,t+1}(L) + \bm{\Lambda}_{1,0,i,t+1} \bm{\Xi}_{0,i,t+1} \bm{Q}^{*}_{i,t+1}(L).
    \end{aligned}
    \label{eq:Q_poly}
\end{equation}
Extracting the $k$-th order coefficients generates the transition sub-blocks $\bm{T}_{it}^{(1,1)}$, $\bm{T}_{it}^{(1,2)}$, $\bm{T}_{it}^{(3,1)}$, and $\bm{T}_{it}^{(3,2)}$.

Next, distributing the block-inverse matrices across the exogenous deterministic variables yields the intercepts
\begin{align*}
    \bm{c}_{it}^{(1)} & = \bm{M}_{1,i,t+1} \left\{ \bm{\alpha}_{1,i,t+1} + \bm{\Gamma}_{1,i,t+1} \bm{X}_{1,i,t+1} + \bm{Q}_{0,i,t+1}\bm{A}_{0,i,t+1} \left( \bm{\nu}_{1,i,t+1} + \bm{K}_{1,i,t+1} \bm{X}_{1,i,t+1} \right) \right\}            \\
    \bm{c}_{it}^{(3)} & = \bm{N}_{1,i,t+1} \left\{ \bm{\nu}_{1,i,t+1} + \bm{K}_{1,i,t+1} \bm{X}_{1,i,t+1} + \bm{\Lambda}_{1,0,i,t+1}\bm{\Xi}_{0,i,t+1} \left( \bm{\alpha}_{1,i,t+1} + \bm{\Gamma}_{1,i,t+1} \bm{X}_{1,i,t+1} \right) \right\},
\end{align*}
yielding the full augmented intercept vector
\begin{equation}
    \label{eq:c_intercept}
    \bm{c}_{it} = \begin{bmatrix}
        \bm{c}_{it}^{(1)}{}^T & \bm{0}_{(L-1)V_1}{}^T & \bm{c}_{it}^{(3)}{}^T & \bm{0}_{(L-1)U}{}^T
    \end{bmatrix}^T.
\end{equation}

Finally, distributing the inverse matrices over the stochastic errors isolates the serially independent process noise components
\begin{align}
    \bm{w}_{it}^{(1)} & = \bm{M}_{1,i,t+1} \left( \bm{\xi}_{1,i,t+1} + \bm{Q}_{0,i,t+1} \bm{A}_{0,i,t+1} \bm{\epsilon}_{1,i,t+1} \right) \label{eq:w1_noise}            \\
    \bm{w}_{it}^{(3)} & = \bm{N}_{1,i,t+1} \left( \bm{\epsilon}_{1,i,t+1} + \bm{\Lambda}_{1,0,i,t+1} \bm{\Xi}_{0,i,t+1} \bm{\xi}_{1,i,t+1} \right). \label{eq:w3_noise}
\end{align}
By standard assumptions, $\bm{\xi}_{1,it} \sim \mathcal{N}(\bm{0}, \bm{\Psi}_{1,it})$ and $\bm{\epsilon}_{1,it} \sim \mathcal{N}(\bm{0}, \bm{\Sigma}_{1,it})$ are mutually independent. Taking the variances and cross-covariances of \eqref{eq:w1_noise} and \eqref{eq:w3_noise} generates the exact covariance components
\begin{align*}
    \bm{W}_{it}^{(1,1)} & = \bm{M}_{1,i,t+1} \left( \bm{\Psi}_{1,i,t+1} + \bm{Q}_{0,i,t+1} \bm{A}_{0,i,t+1} \bm{\Sigma}_{1,i,t+1} \bm{A}_{0,i,t+1}^{T} \bm{Q}_{0,i,t+1}^{T} \right) \bm{M}_{1,i,t+1}^{T}                     \\
    \bm{W}_{it}^{(3,3)} & = \bm{N}_{1,i,t+1} \left( \bm{\Sigma}_{1,i,t+1} + \bm{\Lambda}_{1,0,i,t+1} \bm{\Xi}_{0,i,t+1} \bm{\Psi}_{1,i,t+1} \bm{\Xi}_{0,i,t+1}^{T} \bm{\Lambda}_{1,0,i,t+1}^{T} \right) \bm{N}_{1,i,t+1}^{T} \\
    \bm{W}_{it}^{(1,3)} & = \bm{M}_{1,i,t+1} \left( \bm{\Psi}_{1,i,t+1} \bm{\Xi}_{0,i,t+1}^{T} \bm{\Lambda}_{1,0,i,t+1}^{T} + \bm{Q}_{0,i,t+1} \bm{A}_{0,i,t+1} \bm{\Sigma}_{1,i,t+1} \right) \bm{N}_{1,i,t+1}^{T},
\end{align*}
with $\bm{W}_{it}^{(3,1)} = \left( \bm{W}_{it}^{(1,3)} \right)^{T}$.

Because the strict state lag elements transition deterministically, they contain no process variance. This enforces a sparse, block-diagonal covariance mapping corresponding to our $L$-lag augmented state vector space
\begin{equation}
    \label{eq:w_process_noise}
    \bm{W}_{it} =
    \begin{bmatrix}
        \bm{W}_{it}^{(1,1)} & \bm{0} & \bm{W}_{it}^{(1,3)} & \bm{0} \\
        \bm{0}              & \bm{0} & \bm{0}              & \bm{0} \\
        \bm{W}_{it}^{(3,1)} & \bm{0} & \bm{W}_{it}^{(3,3)} & \bm{0} \\
        \bm{0}              & \bm{0} & \bm{0}              & \bm{0}
    \end{bmatrix}.
\end{equation}
This explicit formulation proves exact equivalence with the state space dynamics specified in Theorem \ref{th:Theorem1}. \qed

\onlineresource{1}

\section{Piloting of Simulation Experiments}

Comparing simulation experiments involving MCMC algorithms is difficult, due to the need to set the algorithm parameters in a way that is fair to each algorithm, and which does not require manual tuning for each Monte Carlo sample. To address this issue, we ran pilots for each of the simulation settings considered, with the aim of finding the parameters for each algorithm yielding satisfactory convergence rates while not being too conservative. 

For the Metropolis-within-Gibbs sampler implemented in \textsc{JAGS} \citep{plummer2003jags}, the main tuning parameters are the length of the adaptation phase, and the additional burn-in, if any. We assessed the convergence by studying trace plots of the parameters of interest. Exemplary trace plots are shown in the sections below, while complete trace plots for every Monte Carlo sample in the pilot study are available in our OSF repository.\footnote{\url{https://osf.io/p754m/overview?view_only=f99de4145ddf49c9a76f0e38b5e42223}}

For NUTS-Kalman and NUTS-Joint, both implemented in \textsc{Stan} \citep{carpenter_stan_2017}, there is potentially a larger number of parameters which can be tuned. We focused our pilots on the two most important ones: the target acceptance probability during warm-up and the maximum tree depth. For these NUTS based algorithms our definition of success was that there should be zero divergent transitions after warm-up. These plots are shown below.

\section{Results from the Pilots}

For each parameter configuration for each of the simulation experiments, we ran 10 Monte Carlo samples with exactly the same data-generating mechanism as in the final simulations. For the NUTS based algorithms, we varied the target acceptance probability during warm-up following the discrete values $0.80$, $0.95$, and $0.99$. The candidate tree depths were 10 and 15. The number of warm-up iterations was always 1,000 in each of 4 parallel chains. Since the goal in the pilot was to assess convergence, and not to obtain a posterior sample, we only ran the samplers for an additional 2,000 iterations after warm-up. For the Metropolis-within-Gibbs algorithm, we used an adaptation phase of 2,000, and ran only 10,000 post-adaptation iterations (30,000 for the model of Section 4.5) in each of 4 parallel chains, since again the goal was to study convergence, and not to obtain a posterior sample.

\subsection{Model 1: Univariate Latent AR(1), Section 4.1}

Figure \ref{r1:fig:eg1-divergent-transition} shows the mean number of divergent transitions after warm-up across the 10 Monte Carlo samples. Based on this, we set the target acceptance probability to 0.80 for NUTS-Kalman and to 0.95 for NUTS-Joint. Since the maximum tree depth seems to have a limited effect, we set it to 10, to obtain a more efficient sampler.

\begin{figure}
    \centering
    \includegraphics[width=.9\linewidth]{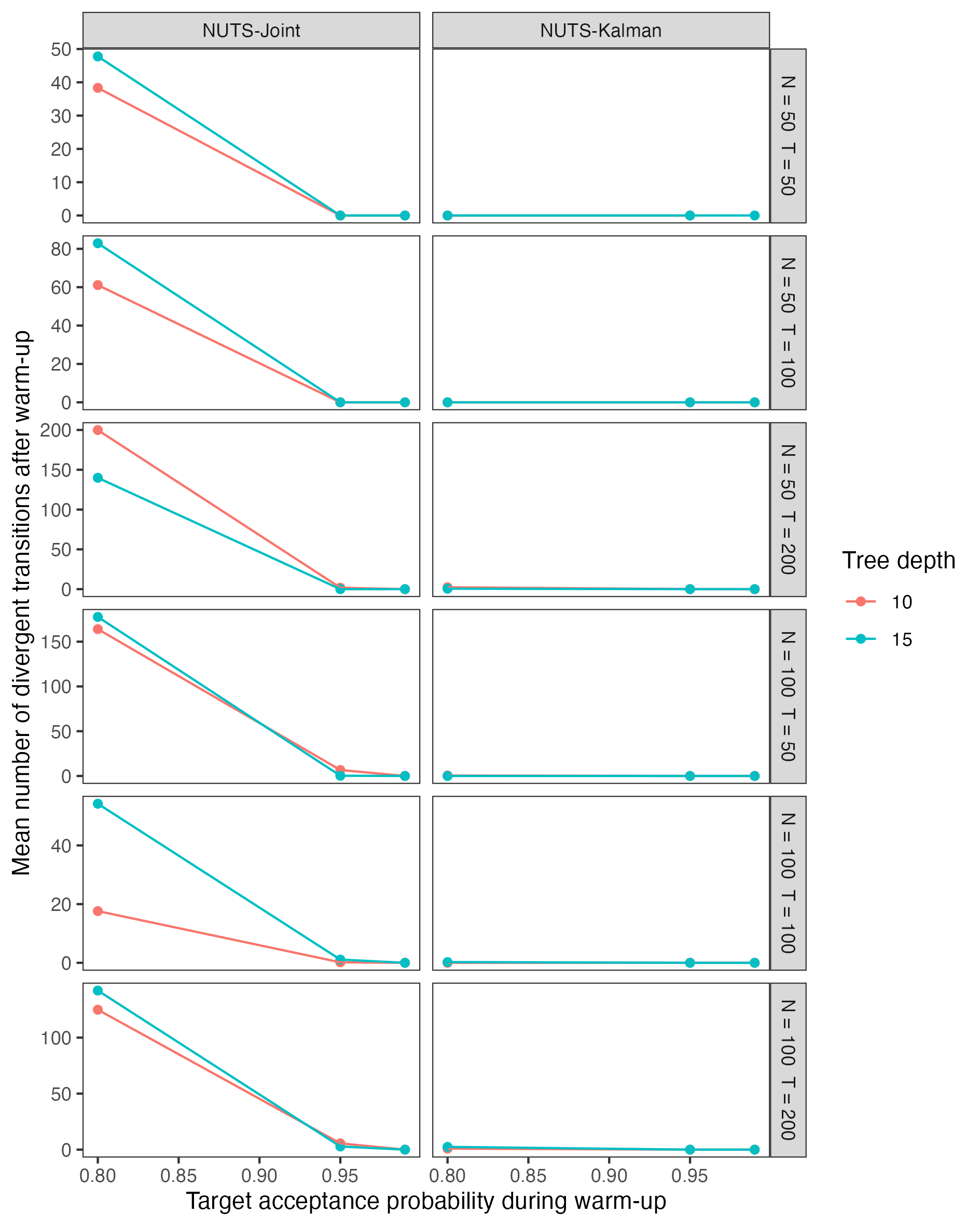}
    \caption{Average number of divergent transition for Model 1 pilot.}
    \label{r1:fig:eg1-divergent-transition}
\end{figure}

Figure \ref{r1:fig:eg1-trace-example} shows an example trace plot for the Metropolis-within-Gibbs sampler. After assessing several such trace plots, we concluded that the sampler seems to mix well after the adaptation phase of 2,000 iterations and set the burn-in to 0.

\begin{figure}
    \centering
    \includegraphics[width=\linewidth]{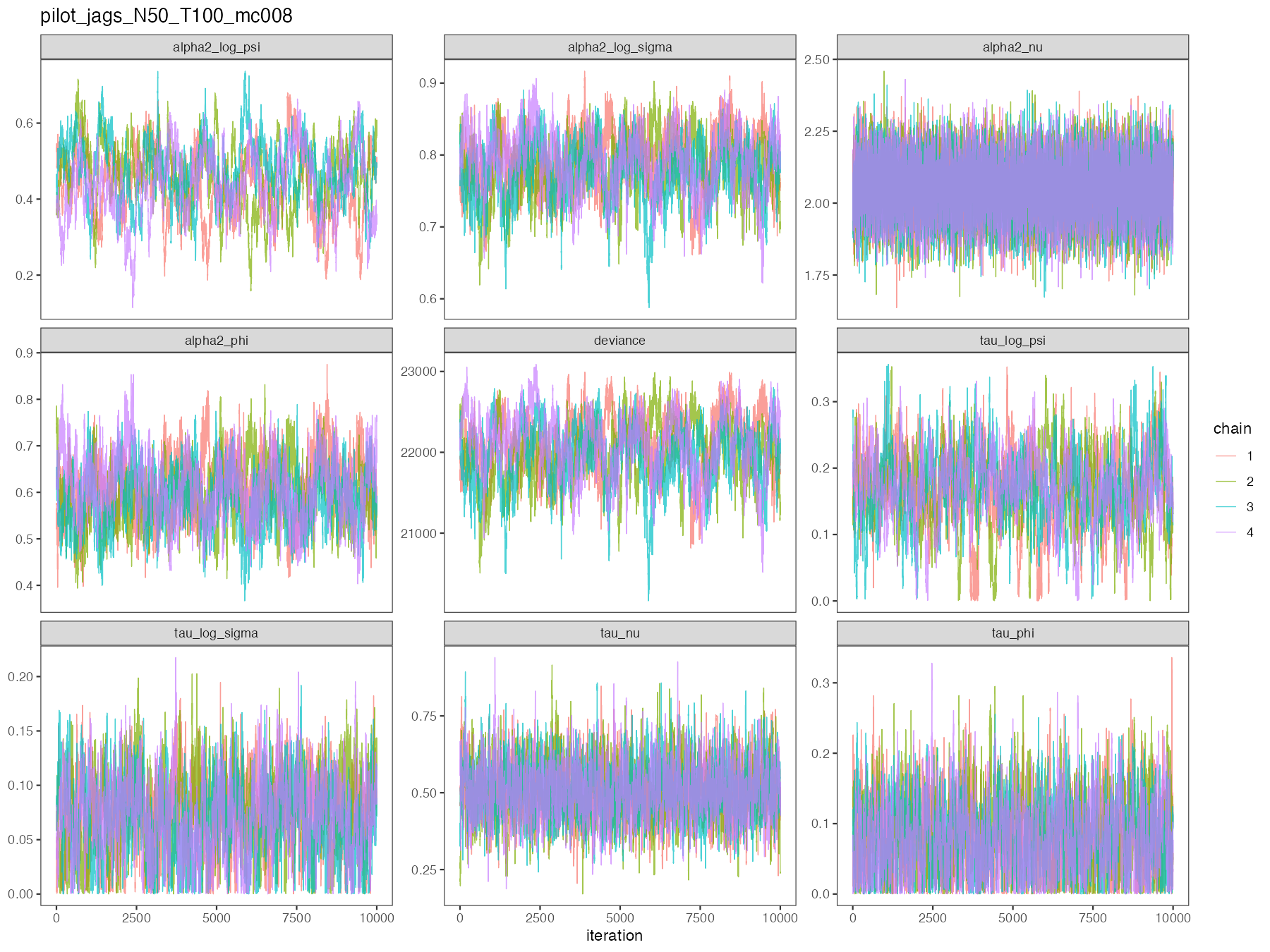}
    \caption{Post-adaptation trace plot for Monte Carlo sample 8 in the \textsc{JAGS} pilot for Model 1 with $N=50$ and $T=100$.}
    \label{r1:fig:eg1-trace-example}
\end{figure}

\subsection{Model 2: Three-Indicator AR(1), Section 4.2}

Figure \ref{r1:fig:eg2-divergent-transitions} shows the mean number of divergent transitions after warm-up across the 10 Monte Carlo samples. Based on this, we set the target acceptance probability to 0.80 for NUTS-Kalman and to 0.95 for NUTS-Joint. Since the maximum tree depth seems to have a limited effect, we set it to 10, to obtain a more efficient sampler.

\begin{figure}
    \centering
    \includegraphics[width=\linewidth]{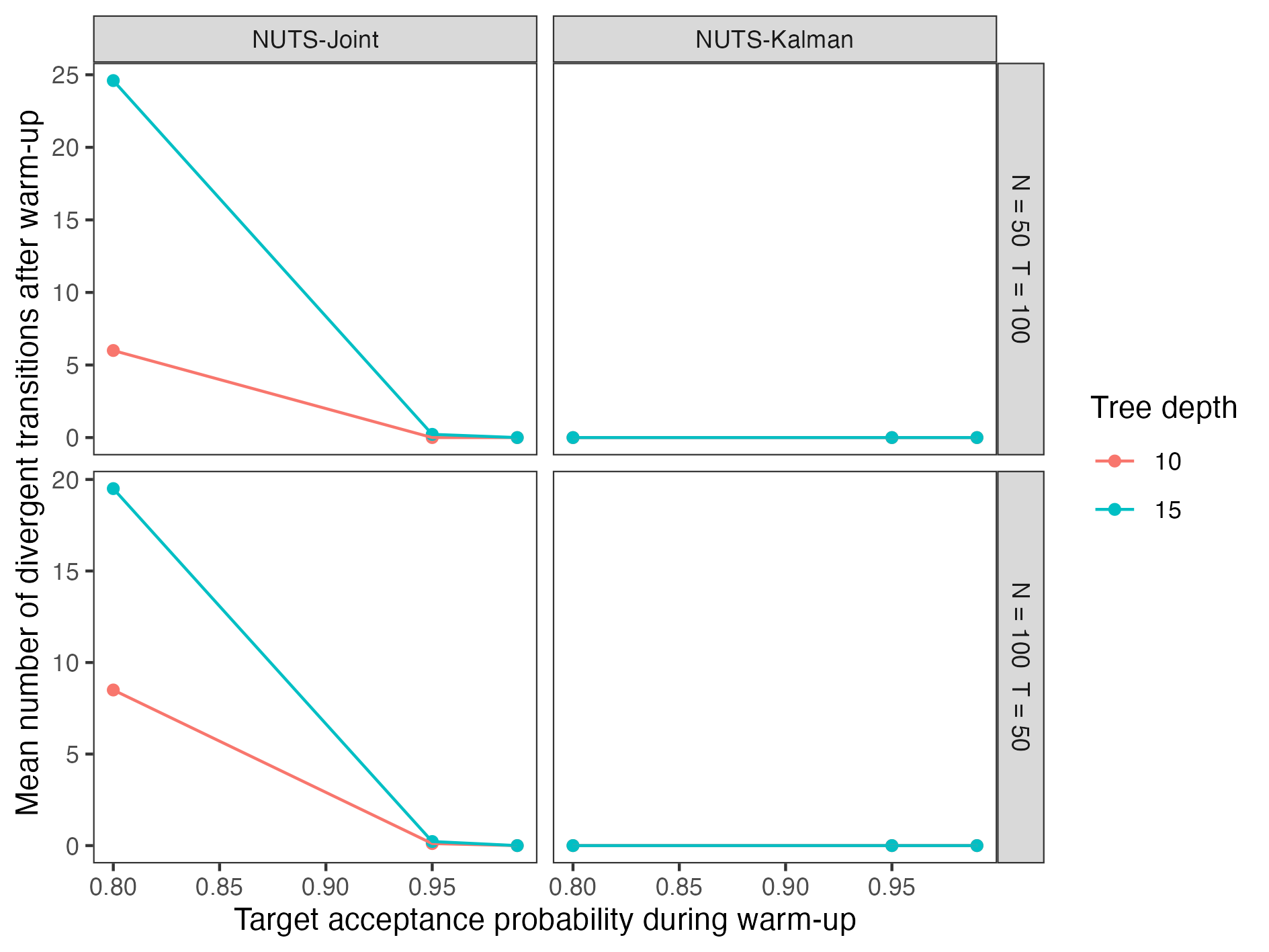}
    \caption{Average number of divergent transitions for Model 2 pilot.}
    \label{r1:fig:eg2-divergent-transitions}
\end{figure}

Figure \ref{r1:fig:eg2-trace-example} shows an example trace plot for the Metropolis-within-Gibbs sampler. After assessing several such trace plots, we concluded that the sampler seems to mix well after the adaptation phase of 2,000 iterations and set the burn-in to 0.

\begin{figure}
    \centering
    \includegraphics[width=\linewidth]{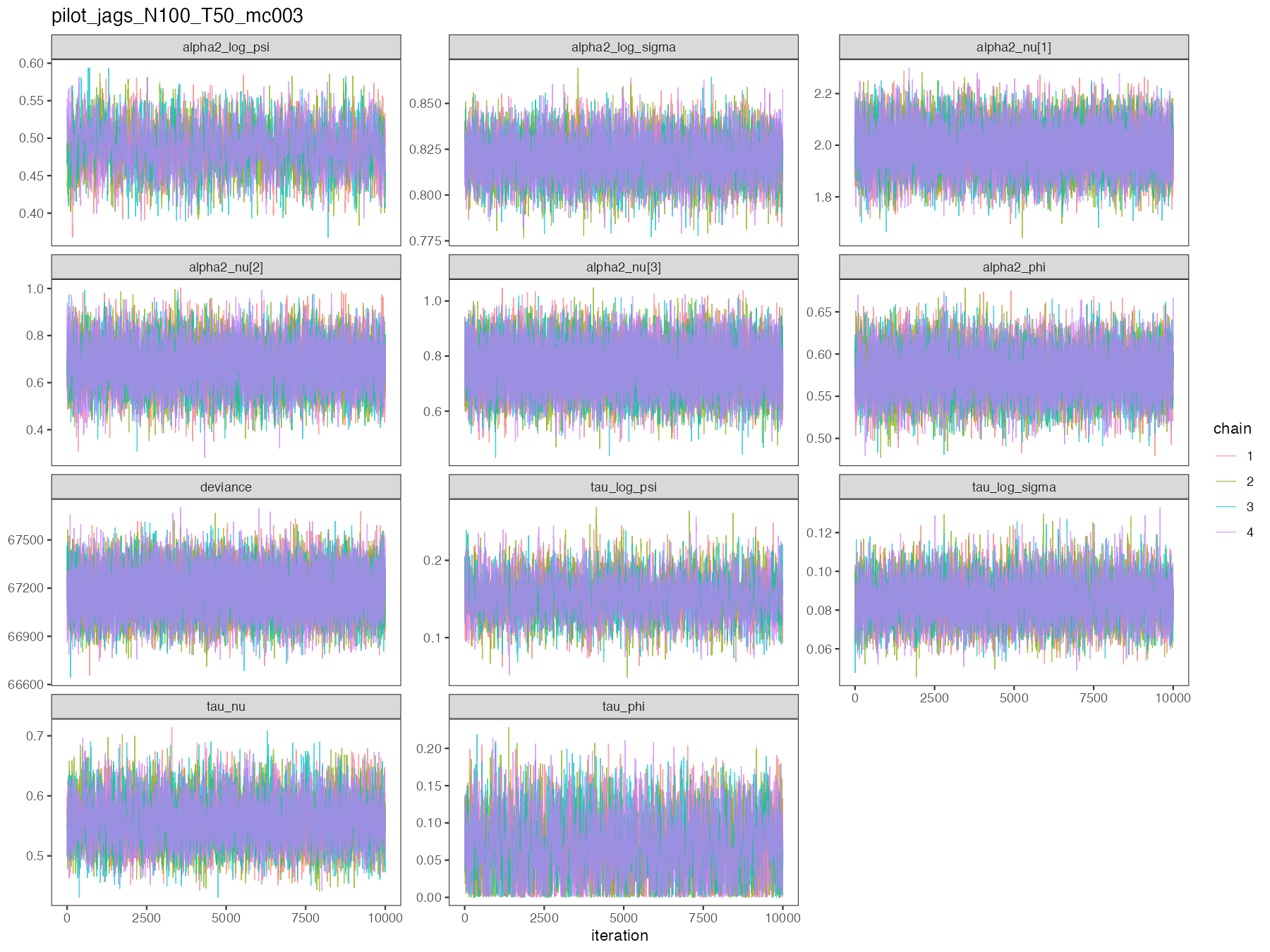}
    \caption{Post-adaptation trace plot for Monte Carlo sample 3 in the \textsc{JAGS} pilot for Model 2 with $N=100$ and $T=50$.}
    \label{r1:fig:eg2-trace-example}
\end{figure}

\subsection{Model 3: Trivariate VAR(1) with K Indicators, Section 4.3}

Figure \ref{r1:fig:eg3-divergent-transitions} shows the mean number of divergent transitions after warm-up across the 10 Monte Carlo samples. Based on this, we set the target acceptance probability to 0.80 for NUTS-Kalman for all $K$ and to 0.99 for NUTS-Joint with $K=1$ and to 0.95 for NUTS-Joint with $K \in \{2,3\}$. Since the maximum tree depth seems to have a limited effect, we set it to 10, to obtain a more efficient sampler.

\begin{figure}
    \centering
    \includegraphics[width=\linewidth]{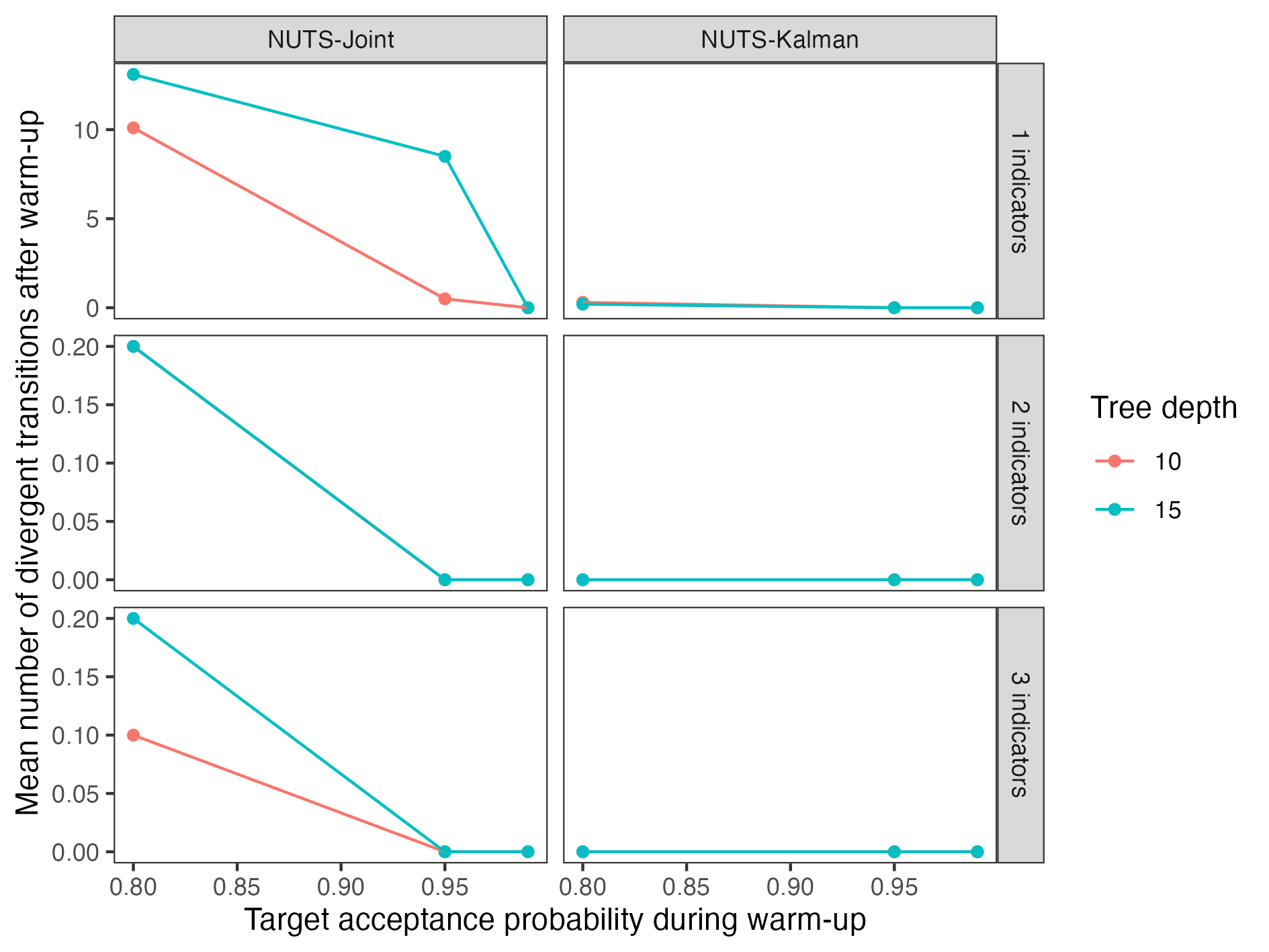}
    \caption{Average number of divergent transitions for Model 3 pilot.}
    \label{r1:fig:eg3-divergent-transitions}
\end{figure}

Figure \ref{r1:fig:eg3-trace-example} shows an example trace plot for the Metropolis-within-Gibbs sampler. After assessing several such trace plots, we concluded that the sampler seems to mix well after the adaptation phase of 2,000 iterations and set the burn-in to 0.

\begin{figure}
    \centering
    \includegraphics[width=\linewidth]{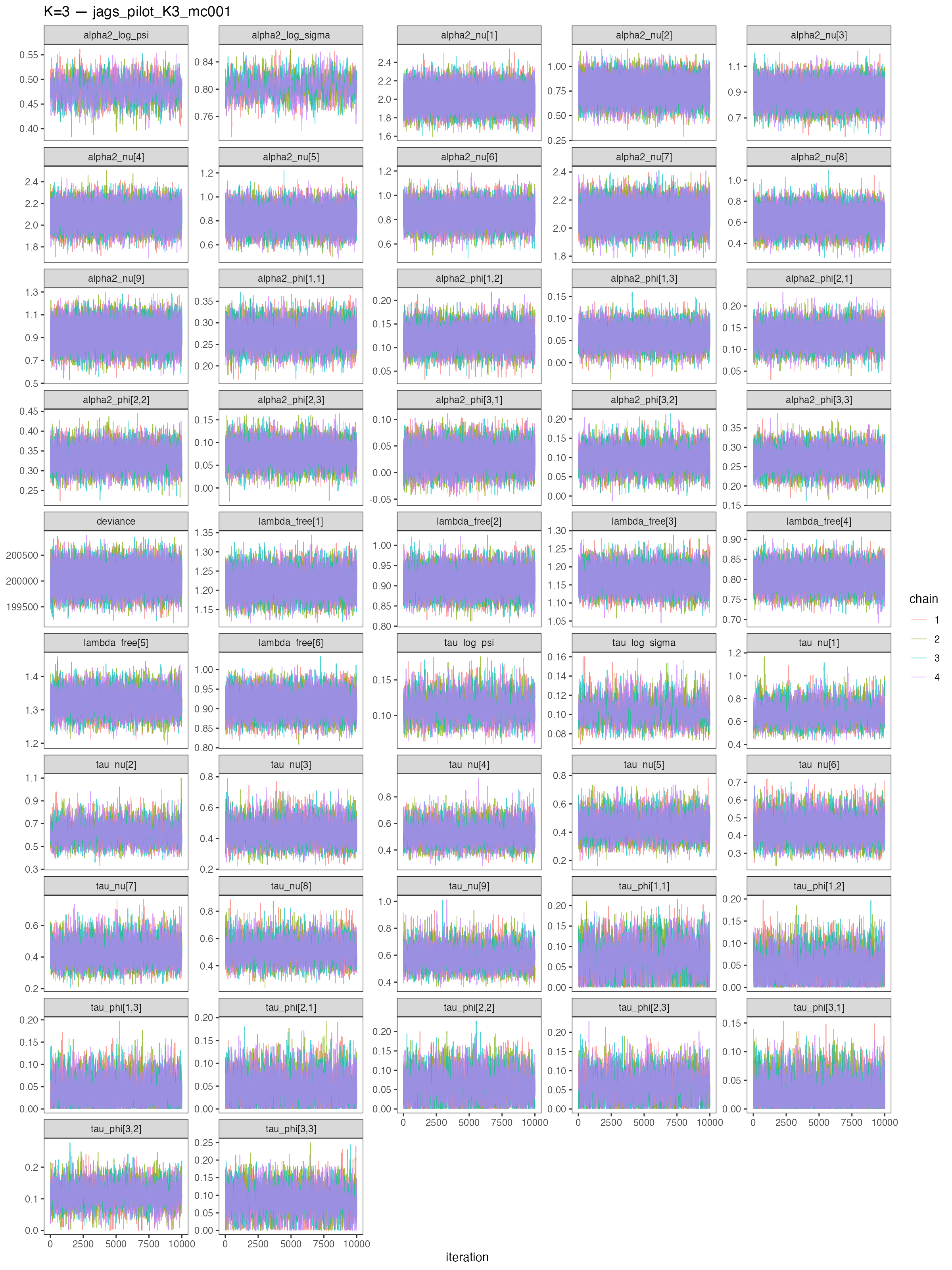}
    \caption{Post-adaptation trace plot for Monte Carlo sample 1 in the \textsc{JAGS} pilot for Model 3 with $K=3$.}
    \label{r1:fig:eg3-trace-example}
\end{figure}

\subsection{Model 4: Multilevel Latent AR(p), Section 4.4}

Figure \ref{r1:fig:eg4-divergent-transition} shows the average number of divergent transitions for various model parameters. Based on this, we set the target acceptance during warm-up to $0.80$ and the maximum tree depth to $10$ for NUTS-Kalman. For NUTS-Joint we set the target acceptance during warm-up to $0.99$; in the lag-4 case we additionally set the maximum tree depth to $15$ while it was $10$ for lag 1, 2, and 3.

\begin{figure}
    \centering
    \includegraphics[width=\linewidth]{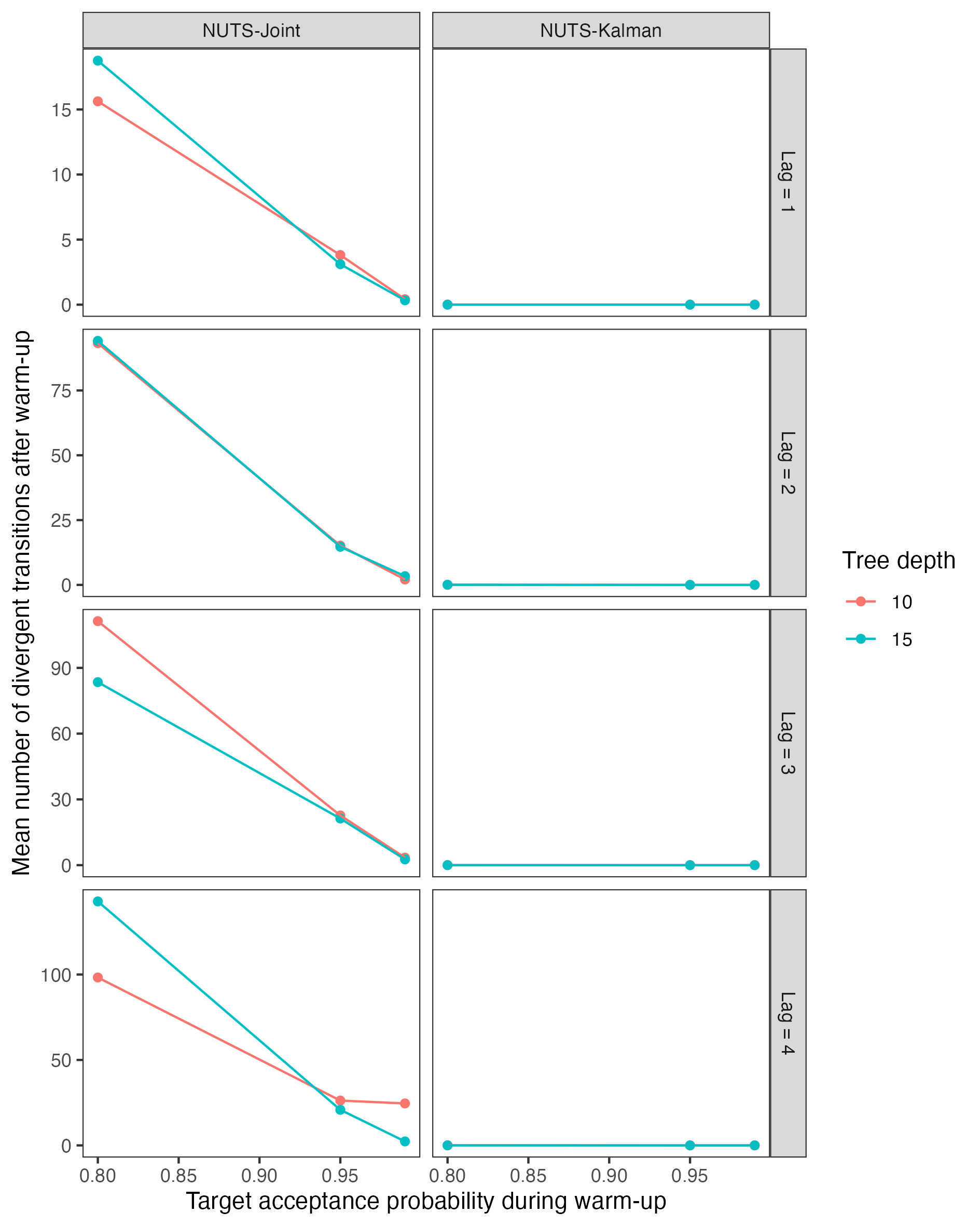}
    \caption{Average number of divergent transitions for Model 4 pilot.}
    \label{r1:fig:eg4-divergent-transition}
\end{figure}

Figure \ref{r1:fig:eg4-trace-example} shows an example trace plot for the Metropolis-within-Gibbs sampler. After assessing several such trace plots, we concluded that the sampler seems to mix well after the adaptation phase of 2,000 iterations and set the burn-in to 0.

\begin{figure}
    \centering
    \includegraphics[width=\linewidth]{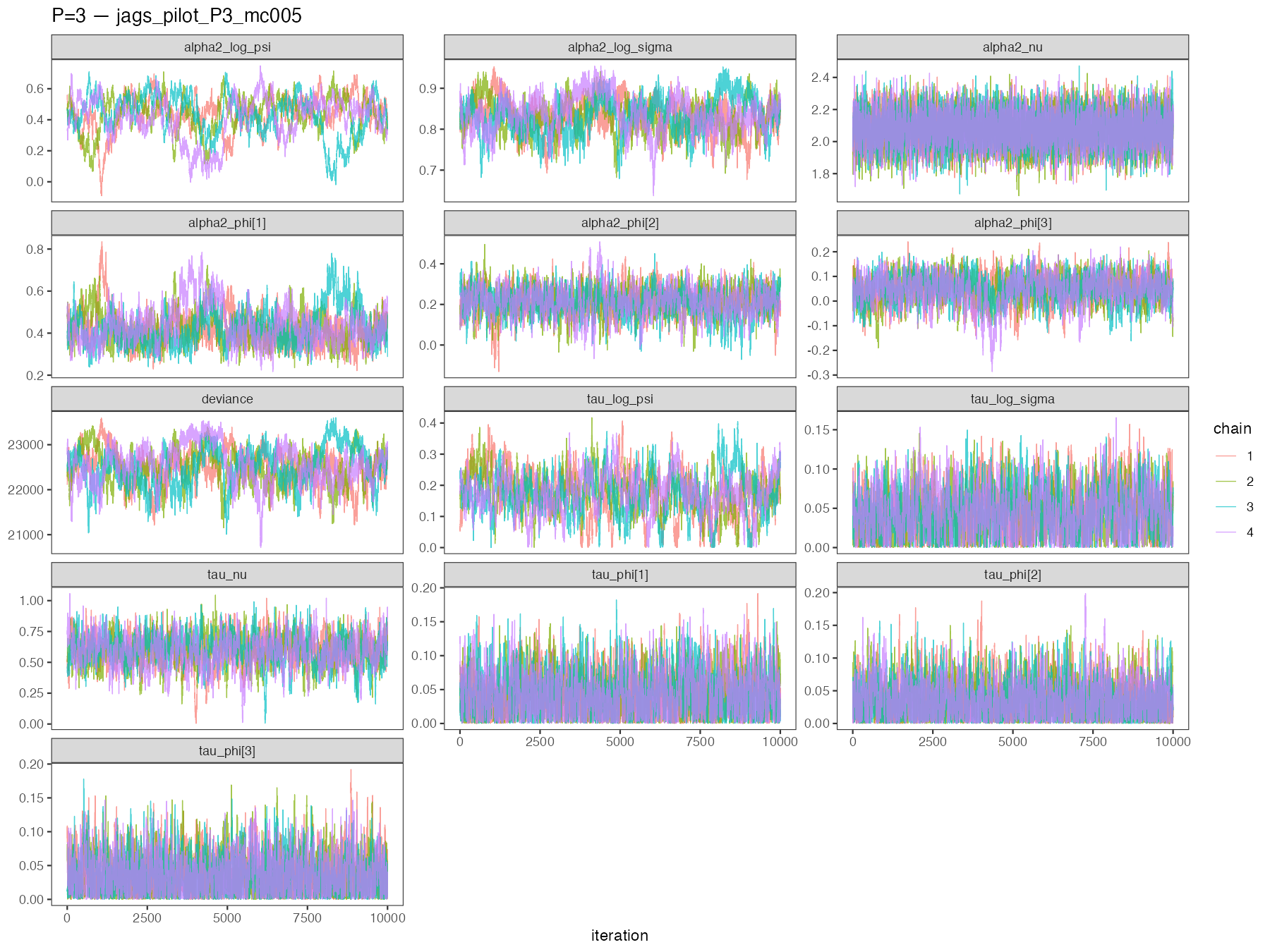}
    \caption{Post-adaptation trace plot for Monte Carlo sample 5 in the \textsc{JAGS} pilot for Model 4 with lag 3.}
    \label{r1:fig:eg4-trace-example}
\end{figure}

\subsection{Model 5: Cross-Classified VAR(1) with Time-Varying Loadings, Section 4.5}

Figure \ref{r1:fig:eg5-divergent-transitions} shows the average number of divergent transitions in the pilot simulations. Based on this, we set the target acceptance probability during warm-up to $0.80$ for NUTS-Kalman and to $0.95$ for NUTS-Joint. For both algorithms we set the maximum tree depth to $10$.

\begin{figure}
    \centering
    \includegraphics[width=\linewidth]{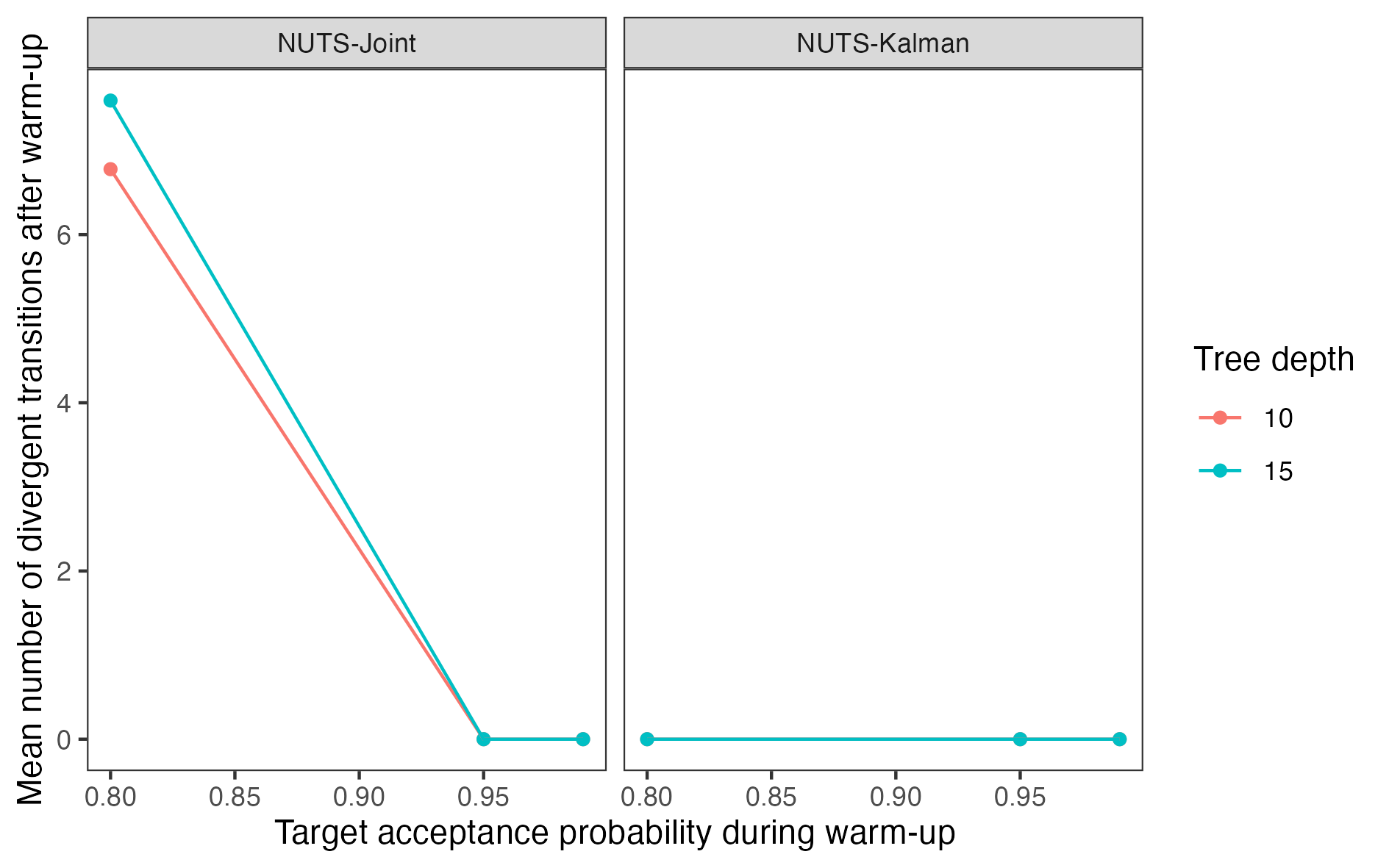}
    \caption{Average number of divergent transitions for Model 5 pilot.}
    \label{r1:fig:eg5-divergent-transitions}
\end{figure}

Figure \ref{r1:fig:eg5-trace-example} shows an example trace plot for the Metropolis-within-Gibbs sampler. Note that the chains for parameters \textsc{mean\_alpha\_phi[2,2]} and \textsc{tau2[11]} are not stuck and do indeed show good convergence, but that this is hard to see in the plots due to the extreme initial value. After assessing several such trace plots, we concluded that the sampler needed an additional burn-in phase after adaptation, which we set to 20,000 iterations.

\begin{figure}
    \centering
    \includegraphics[width=\linewidth]{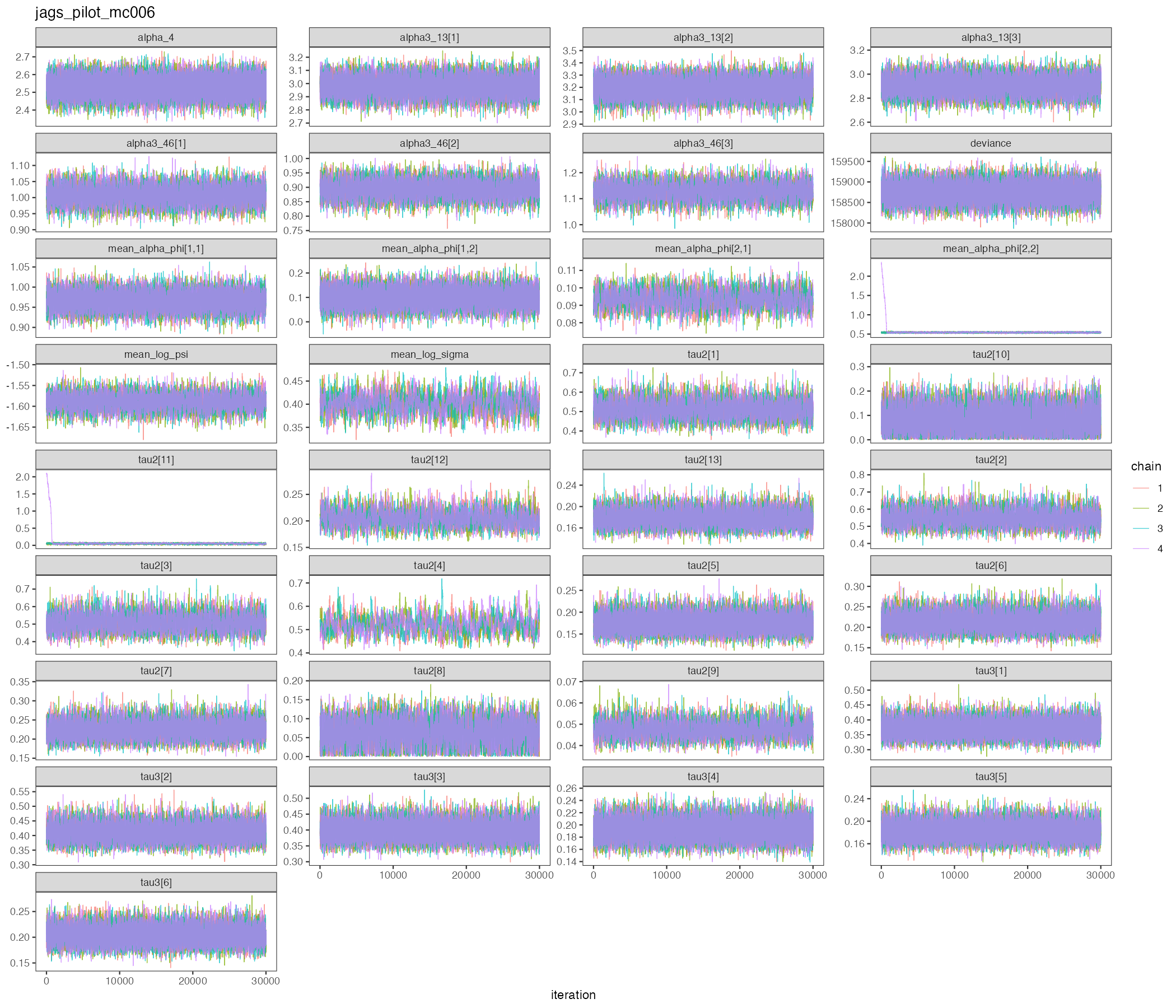}
    \caption{Post-adaptation trace plot for Monte Carlo sample 6 in the \textsc{JAGS} pilot for Model 5.}
    \label{r1:fig:eg5-trace-example}
\end{figure}

\onlineresource{2}

Throughout the manuscript the proposed NUTS-Kalman algorithm is compared against a
baseline sampler that we describe as ``Metropolis-within-Gibbs'' and which is similar to the algorithm described in \citet{asparouhov_dynamic_2018}. Because we do not have an \textsc{Mplus} license, that baseline is implemented in \textsc{JAGS (4.3.2)} \citep{plummer2003jags}. Here we document exactly which sampling method \textsc{JAGS} applies to which parameter, and how that assignment corresponds to the block structure of \citet{asparouhov_dynamic_2018}'s algorithm, which is implemented in \textsc{Mplus} \citep{mplus}.

\citet{asparouhov_dynamic_2018}'s algorithm partitions all unknowns into thirteen blocks $B1$--$B13$ and updates each from its full conditional
\citep[p.~366]{asparouhov_dynamic_2018}. Twelve of these blocks have explicit conditional
distributions and are drawn by exact Gibbs steps. Only block $B13$, containing the random variance parameters,\footnote{We use the term \emph{random} to describe individual-specific or timepoint-specific effects in this supplementary document for consistency with \citet{asparouhov_dynamic_2018}, even though these terms are ambiguous in a Bayesian context.} has a conditional that is ``not explicit,'' and there \citet[Appendix~B]{asparouhov_dynamic_2018} resort to a random-walk Metropolis--Hastings step with an adaptively tuned proposal variance. Our \textsc{JAGS} implementation follows the
same principle: every parameter whose full conditional is conjugate-normal is drawn by an exact Gibbs step, and every parameter whose full conditional is not of standard form is updated by Metropolis-Hastings using slice sampling \citep{neal_slice_2003}. Importantly, the two implementations classify the $\mathcal{O}(N \cdot T)$ within-level latent states and the random variance parameters of block $B13$ identically, and differ only in which general-purpose method is applied to the non-conjugate minority: random-walk Metropolis--Hastings in \textsc{Mplus} and slice sampling in \textsc{JAGS}. Our use of bounding transforms on the autoregressive random effects, and of weakly informative rather than conjugate priors on the population standard deviations, moves some parameters out of the conjugate class that \textsc{Mplus} would draw by Gibbs. These differences arise from our model and prior specification rather than from \textsc{JAGS}.

\section{The \textsc{Mplus} Algorithm}
\label{r2:sec:mplus}

\citet[p.~366]{asparouhov_dynamic_2018} describe their estimator as follows:

\begin{quote}
``The estimation is based on the MCMC algorithm via the Gibbs sampler. All model parameters, latent variables, random effects, between-Level~2 components, between-Level~3 components, and missing data are arranged into blocks. Each of these blocks is updated (\dots) from the conditional distribution of that block, conditional on all other remaining blocks and the data.''
\end{quote}

The blocks they refer to are
\begin{center}
\begin{tabular}{ll}
$B1$ & $\bm{Y}_{2,i}$ (between-individual components) \\
$B2$ & All random slopes $\bm{s}_{2,i}$ \\
$B3$ & $\bm{Y}_{3,t}$ (between-timepoint components) \\
$B4$ & All random slopes $\bm{s}_{3,t}$ \\
$B5$ & Other latent variables $\bm{\eta}_{2,i}$ and $\bm{\eta}_{3,t}$ \\
$B6$ & Latent variables $\bm{\eta}_{1,it}$, including initial conditions \\
$B7$ & Missing data variables $\bm{Y}_{it}$ \\
$B8$ & Initial conditions $\bm{Y}_{1,it}$ and $\bm{X}_{1,it}$ for $t \leq 0$ \\
$B9$ & Threshold parameters for categorical variables \\
$B10$ & Underlying variables for categorical variables \\
$B11$ & Nonrandom intercept, slope and loading parameters $\bm{\theta}_1$ \\
$B12$ & Nonrandom variance, covariance and correlation parameters $\bm{\theta}_2$ \\
$B13$ & Random variance parameters \\
\end{tabular}
\end{center}
Blocks $B9$ and $B10$ are irrelevant here, since all our simulations use continuous
indicators, and $B7$--$B8$ are irrelevant since we simulate complete data with $t \geq 1$.\footnote{Block B7 is, however, relevant in the application example in Section 5.}

Two features of this algorithm matter for the present comparison:

\textbf{The latent states are drawn by exact Gibbs, one time point at a time.} For block $B6$, \citet[Appendix~B]{asparouhov_dynamic_2018} give the conditional
\begin{equation}
  \label{r2:eq:asp-b6}
  [\bm{\eta} \mid \cdot] \sim \mathcal{N}(\bm{D}\bm{d}, \bm{D})      
\end{equation}
where
\begin{equation*}
    \bm{D} = \left( \bm{\Lambda}^{T}\bm{\Theta}^{-1}\bm{\Lambda} + \bm{\Psi}_{0}^{-1} \right)^{-1}
\end{equation*}
and
\begin{equation*}
    \bm{d} = \bm{\Lambda}^{T}\bm{\Theta}^{-1}(\bm{y} - \bm{\nu} - \bm{K}\bm{x})
  + \bm{\Psi}_{0}^{-1}\bm{B}_{0}^{-1}(\bm{\alpha} + \bm{\Gamma}\bm{x}),
\end{equation*}
and note that ``because $\bm{\eta}_{1,it}$ are not independent across time they cannot be
generated simultaneously in an efficient manner, as that will require computing the large
joint conditional distribution of $\bm{\eta}_{1,it}$ for all $t$. Therefore block $B6$ is
essentially split into separate blocks, one for each individual $i$ and time $t$,'' each
updated conditionally on the neighbouring states
$\bm{\eta}_{1,i,t-L}, \dots, \bm{\eta}_{1,i,t+L}$. This single-site conditional update of the
$\mathcal{O}(N \cdot T)$ latent states is precisely what the present
paper replaces with Kalman marginalization.

\textbf{Only the random variance parameters need Metropolis--Hastings.} For block $B13$, \citet[Appendix~B]{asparouhov_dynamic_2018} write:

\begin{quote}
    ``The conditional distribution of the random effects from Equation~11 is not explicit and we use the Metropolis--Hastings algorithm to generate values from that distribution. Suppose that
    $Y_{1,it_1}$ has a random residual variance $\sigma_i = \mathrm{Exp}(s_{2,i})$, where $s_{2,i}$
    is a normally distributed random effect.''
\end{quote}

A proposal $s_1 \sim \mathcal{N}(s_0, V)$ is drawn and accepted with probability $\min(1, R)$, where
\begin{equation}
  \label{r2:eq:asp-b13}
  R = \frac{P(s_{2,i} = s_1 \mid \cdot)}{P(s_{2,i} = s_0 \mid \cdot)}
      \prod_{t} \frac{P(Y_{1,it_1} \mid \sigma_i = \mathrm{Exp}(s_1))}
                     {P(Y_{1,it_1} \mid \sigma_i = \mathrm{Exp}(s_0))} .
\end{equation}
The proposal variance $V$ ``is chosen to be a small value such as 0.1 and is adjusted during
a burn-in stage of the estimation to obtain optimal mixing; that is, optimal acceptance rate
in the Metropolis--Hastings algorithm. The optimal acceptance rate is considered to be between
.25 and .50.'' They add that ``random variances for latent factors are estimated similarly.''

The trigger for the Metropolis--Hastings step is therefore precisely the exponential link
between a normally distributed random effect and a variance parameter. Every one of our
\textsc{JAGS} models contains exactly this construction, written as
\begin{lstlisting}
prec_sigma[i] <- exp(-2 * log_sigma[i]).
\end{lstlisting}

\section{How \textsc{JAGS} Chooses Its Samplers}
\label{r2:sec:jags-samplers}

\textsc{JAGS} assigns samplers automatically. As the user manual explains \citep[\S3.3.3]{plummer2017jags}:

\begin{quote}
``JAGS holds a list of sampler factory objects, which inspect the graph, recognize sets of parameters that can be updated with specific methods, and generate sampler objects for them. The list of sampler factories is traversed in order, starting with sampling methods that are efficient, but limited to certain specific model structures and ending with the most generic, possibly inefficient, methods.''
\end{quote} 

Two factories are relevant to our models:

\begin{itemize}
\item \texttt{bugs::Conjugate}, which ``creates samplers for nodes with conjugate
distributions, i.e.\ where the full conditional distribution of the node belongs to the same
family as the prior'' \citep[\S9.3]{plummer2017jags}. For a node with a normal prior
whose stochastic children are normal with means linear in the node and precisions not
depending on it, this produces an \emph{exact draw from the normal full conditional}, i.e.\ a
Gibbs step.

\item \texttt{base::RealSlicer}, the univariate continuous slice sampler. Slice sampling
\citep{neal_slice_2003} is described in the manual as ``the `work horse' sampling method
in \textsc{JAGS}'' \citep[\S8.2.2]{plummer2017jags}. It is the fallback for continuous scalar
nodes for which no conjugate structure is detected.
\end{itemize}

Because the conjugate factory is queried first, the sampler assignment is fully determined by
whether a node's full conditional is normal.

Table~\ref{r2:tab:correspondence} states the general mapping. It holds for all five models considered in the simulations.

\begin{table}[htbp]
\centering
\caption{Correspondence between the \textsc{Mplus} blocks of \citet{asparouhov_dynamic_2018}
and the samplers assigned by \textsc{JAGS} in our implementation. In the upper panel the two
implementations place the parameter in the same class, and either both draw it exactly or both
resort to a general-purpose method. In the lower panel the classification differs, and in both
cases the cause is our model and prior specification rather than \textsc{JAGS}: \textsc{Mplus}
applies no bounding transform to random slopes, and uses conjugate priors for the variance
parameters of block $B12$. The same specification is used for the \textsc{Stan}
implementations of NUTS-Kalman and NUTS-Joint, so it affects all three algorithms compared in
the manuscript equally.}
\label{r2:tab:correspondence}
\small
\begin{tabular}{p{3.85cm}p{2.95cm}p{2.9cm}p{2.9cm}}
\toprule
Quantity & \textsc{Mplus} block & \textsc{Mplus} update & \textsc{JAGS} update \\
\midrule
\multicolumn{4}{l}{\emph{Same classification in both implementations}} \\
\addlinespace
Within-level latent states $\bm{\eta}_{1,it}$
  & $B6$
  & Exact Gibbs, Eq.~\eqref{r2:eq:asp-b6}
  & Exact Gibbs \newline (\texttt{bugs::Conjugate}) \\
\addlinespace
Between-individual components and random intercepts
  & $B1$, $B5$
  & Exact Gibbs
  & Exact Gibbs \newline (\texttt{bugs::Conjugate}) \\
\addlinespace
Between-timepoint components (loadings, intercepts)
  & $B3$, $B5$
  & Exact Gibbs
  & Exact Gibbs \newline (\texttt{bugs::Conjugate}) \\
\addlinespace
Nonrandom intercepts and loadings $\bm{\theta}_1$
  & $B11$
  & Exact Gibbs
  & Exact Gibbs \newline (\texttt{bugs::Conjugate}) \\
\addlinespace
Random autoregressive coefficients entering \emph{linearly} \newline (Model 4 only)
  & $B2$
  & Exact Gibbs
  & Exact Gibbs \newline (\texttt{bugs::Conjugate}) \\
\addlinespace
Random variance parameters \newline ($\sigma_i = \exp(\cdot)$, $\psi_i = \exp(\cdot)$)
  & $B13$
  & \textbf{Random-walk MH}, \newline Eq.~\eqref{r2:eq:asp-b13}
  & \textbf{Slice sampling} \newline (\texttt{base::RealSlicer}) \\
\midrule
\multicolumn{4}{l}{\emph{Classification differs; see Section~\ref{r2:sec:differences}}} \\
\addlinespace
Random autoregressive coefficients \emph{bounded} by a $\tanh$ or inverse-logit transform to
enforce stationarity \newline (Models 1, 2, 3 and 5)
  & $B2$
  & Exact Gibbs; \newline \textsc{Mplus} applies no bounding transform to random slopes
  & \textbf{Slice sampling} \newline (\texttt{base::RealSlicer}) \\
\addlinespace
Population standard deviations $\tau$
  & $B12$
  & Exact Gibbs under conjugate \newline (inverse-gamma) priors
  & \textbf{Slice sampling}, except where a half-normal prior scales a linearly entering
    component (Model 5) \\
\bottomrule
\end{tabular}
\end{table}

\section{Where the Two Algorithms Genuinely Differ}
\label{r2:sec:differences}

For completeness we list the differences we are aware of. None of them concerns the
$\mathcal{O}(N \cdot T)$ latent states.

\begin{enumerate}
\item \textbf{Non-conjugate updates: slice sampling versus random-walk Metropolis--Hastings.}
Both are general-purpose methods for non-standard univariate conditionals, and both are Metropolis--Hastings algorithms in the broad sense: slice sampling \citep{neal_slice_2003} augments the target with an
auxiliary height variable and alternates conditional draws, and can be viewed as an
auto-tuning alternative to a random-walk proposal. Slice sampling has the practical advantage
of requiring no proposal-variance tuning, whereas the \textsc{Mplus} scheme adapts $V$ during burn-in
to an acceptance rate between $0.25$ and $0.50$ and then freezes it. We are not aware of a
general result establishing which is more efficient for these particular conditionals.

\item \textbf{Blocking of the latent state vector.} When $V_1 > 1$ latent processes are present, \textsc{Mplus} draws $\bm{\eta}_{1,it}$ jointly from the multivariate conditional whereas \textsc{JAGS} draws its components one at a time. This affects only Model 3 (Section 4.3), since $V_1 = 1$ in Models 1, 2, 4 and 5. Blocked updates mix at least as well as
single-site updates, so on this point \textsc{JAGS} is at a disadvantage
relative to \textsc{Mplus}.

\item \textbf{Bounding transforms on the autoregressive random effects.} In Models 1, 2, 3 and
5 the individual-specific autoregressive coefficients are passed through a $\tanh$ or
inverse-logit transform, so that the implied coefficient lies in $(-1,1)$. This is a deliberate
feature of our specification, but it removes these parameters from the conjugate class: their
full conditionals are no longer normal, and \textsc{JAGS} slice samples them. In the block
structure of \citet{asparouhov_dynamic_2018} they belong to $B2$, ``all random slopes,'' which
is an explicit Gibbs block, since \textsc{Mplus} places no such constraint on random slopes.
Model 4, in which the autoregressive coefficients enter linearly, is the exception, and it is
correspondingly the model in which the slice-sampled set is smallest, both absolutely and as a
proportion.

\item \textbf{Parameterisation and priors of the population standard deviations.} Our models
use a non-centred parameterisation, $\theta_i = \alpha_2 + \tau z_i$ with
$z_i \sim \mathcal{N}(0,1)$, together with weakly informative priors on $\tau$.
\textsc{Mplus} uses a centred parameterisation with conjugate inverse-gamma or inverse-Wishart
priors for the variance parameters of block $B12$, for which an explicit Gibbs draw is always
available. Under our specification, whether the conditional of a given $\tau$ is conjugate
depends on two things: whether its prior is normal, and whether the component it scales enters
the model linearly. A half-normal prior, \texttt{dnorm(0, 1) T(0, )}, on a $\tau$ that scales a
linearly entering component yields a truncated normal conditional and a conjugate sampler;
this is the case for the nineteen standard deviations of Model 5, thirteen of which are drawn
by exact Gibbs. A half-Cauchy prior, \texttt{dt(0, s, 1) T(0, )}, is not conjugate whatever the
component does, and a $\tau$ scaling a component that enters through an exponential or $\tanh$
link has a non-normal conditional whatever its prior. The number of population standard
deviations that are consequently slice sampled is four in Model 1, four in Model 2, twenty in
Model 3, seven in Model 4 and six in Model 5. We chose the non-centred form with weakly
informative priors because it is the specification used for the \textsc{Stan} implementations,
and using the same priors and the same parameterisation across all three algorithms is
necessary for the comparison of effective sample sizes to be meaningful.

\item \textbf{Implementation and language.} \textsc{Mplus} is compiled special-purpose software;
\textsc{JAGS} is a general-purpose interpreter that traverses a directed acyclic graph. As
\citet{li_fitting_2022} report, this can produce wall-time differences at fixed iteration count. This is a matter of constant factors rather than algorithmic family.
\end{enumerate}

\section{Empirical Verification}
\label{r2:sec:verification}

The sampler assignments can be read directly off a compiled model with \texttt{rjags::list.samplers}, which ``returns a named list with an entry for each Sampler used by the model \dots The names of the list elements indicate the sampling methods that are used to update each block'' \citep{plummerRjagsBayesianGraphical2025}. The folder for each of the five simulation experiments in our OSF repository\footnote{\url{https://osf.io/p754m/overview?view_only=f99de4145ddf49c9a76f0e38b5e42223}} contains an \textsc{R} script named \verb!list_samplers.R! which lists the samplers. The outputs for all five experiments are listed in Table \ref{r2:tab:verification}.

\begin{table}[htbp]
\centering
\caption{Output of \texttt{list.samplers} for each simulation model, at the configuration
stated in the first column. Only two sampling methods occur anywhere in the five models, and
in every model \texttt{bugs::ConjugateNormal} is the method assigned to all \texttt{eta}
nodes.}
\label{r2:tab:verification}
\small
\begin{tabular}{lrrrr}
\toprule
Model & Conjugate & Slice sampling & Other & \% conjugate \\
\midrule
4.1 ($N=100$, $T=200$) & 20101 & 307 & 0 & 98.5 \\
4.2 ($N=100$, $T=50$) & 5305 & 307 & 0 & 94.5\\
4.3 ($N=50$, $T=100$, $K=3$) & 15465 & 581 & 0 & 96.4 \\
4.4 ($N=100$, $T=50$, $p=4$) & 5505 & 209 & 0 & 96.3\\
4.5 ($N=100$, $T=150$) & 16620 & 612 & 0 & 96.4\\
\bottomrule
\end{tabular}
\end{table}

\section{Conclusion}

The baseline sampler used throughout the manuscript is a Metropolis-within-Gibbs algorithm in
the same sense as the algorithm of \citet{asparouhov_dynamic_2018}: parameters with explicit
conditional distributions are updated by exact Gibbs draws, and the remainder by a
general-purpose Metropolis-type method. The two implementations agree on which parameters fall
into which category. They differ in which general-purpose method is used for the second
category, slice sampling rather than random-walk Metropolis--Hastings, and in the blocking of
the latent state vector when more than one latent process is present.

Most importantly, the $\mathcal{O}(N \cdot T)$ within-level latent variables, which together
with the other conjugate blocks constitute $94.5\%$--$98.5\%$ of all sampled quantities and
which are the object of the present paper, are
updated by exact single-site Gibbs draws from Equation~\eqref{r2:eq:asp-b6} in both
implementations. The inefficiency that NUTS-Kalman removes is a property of that shared
update, not of \textsc{JAGS}.

\onlineresource{3}

This supplementary document contains further details about the models used in the simulation experiments, and a description of how computation time was measured in the simulation experiments and in the application example (Section \ref{r3:sec:timing}).

\section{Multilevel Scalar Latent AR(1) Model}
\label{r3:sec:eg1}

This is the model of Section 4.1 of the manuscript. Table \ref{r3:tab:eg1-truth} lists the population parameters used to generate the data.

\begin{table}[H]
    \centering
    \caption{True population parameters for the multilevel scalar latent AR(1) model.}
    \label{r3:tab:eg1-truth}
    \begin{tabular}{lll}
        \toprule
        Parameter             & Description                                              & True value \\
        \midrule
        $\nu$                 & Mean of the participant-specific intercept               & $2.0$ \\
        $\tau_{\nu}$          & Standard deviation of the participant-specific intercept & $0.5$ \\
        $\phi$                & Mean of the unconstrained AR(1) effect (arctanh scale)   & $0.6$ \\
        $\tau_{\phi}$         & Standard deviation of the unconstrained AR(1) effect     & $0.1$ \\
        $\log\sigma$          & Mean of the log measurement noise standard deviation     & $0.8$ \\
        $\tau_{\log\sigma}$   & Standard deviation of the log measurement noise          & $0.1$ \\
        $\log\psi$            & Mean of the log process noise standard deviation         & $0.5$ \\
        $\tau_{\log\psi}$     & Standard deviation of the log process noise              & $0.1$ \\
        \bottomrule
    \end{tabular}
\end{table}

\subsection{Priors}
\label{r3:sec:eg1-priors}

For a generic between-participant latent variable
$\eta_{2,i\theta}$ with population mean $\theta$ we write $\eta_{2,i\theta} = \theta + \tau_{\theta}z_{\theta,i}$ with $z_{\theta,i} \sim \mathcal{N}(0,1)$ independently across participants. The population-level priors are
\begin{equation}
    \label{r3:eq:eg1-priors}
    \begin{aligned}
        \nu                                                   & \sim \mathcal{N}(0, 5^{2})     \\
        \phi, \; \log\sigma, \; \log\psi                      & \sim \mathcal{N}(0, 1^{2})     \\
        \tau_{\nu}                                            & \sim \text{Half-Cauchy}(0, 2)  \\
        \tau_{\phi}, \; \tau_{\log\sigma}, \; \tau_{\log\psi} & \sim \text{Half-Cauchy}(0, 1).
    \end{aligned}
\end{equation}
The same priors were used for all three algorithms, so that the three samplers target
precisely the same posterior distribution. In \textsc{JAGS} the half-Cauchy priors are specified as truncated $t$ distributions with one degree of freedom, and the hyperbolic tangent is written as $\tanh(x) = 2\,\mathrm{ilogit}(2x) - 1$, since \textsc{JAGS} provides no \texttt{tanh} function.

\subsection{Kalman Filter Implementation}
\label{r3:sec:eg1-kalman}

Our Kalman filter implementation exploits the scalar nature of this particular model, so that all matrix operations in the general recursion reduce to scalar arithmetic and no matrix inversion is required. The initial condition is $a_{i,1|0} = 0$ and $P_{i,1|0} = 1$, matching the $\mathcal{N}(0,1)$ distribution of the initial latent state. For $t = 1, \dots, T$ the recursion is
\begin{equation}
    \label{r3:eq:eg1-kalman}
    \begin{aligned}
        v_{it}       & = y_{1,it} - (\nu_{2,i} + a_{i,t|t-1})                        \\
        F_{it}       & = P_{i,t|t-1} + \sigma_{2,i}^{2}                              \\
        a_{i,t|t}    & = a_{i,t|t-1} + P_{i,t|t-1}F_{it}^{-1}v_{it}                   \\
        P_{i,t|t}    & = P_{i,t|t-1} - P_{i,t|t-1}^{2}F_{it}^{-1}                     \\
        a_{i,t+1|t}  & = \phi_{2,i}a_{i,t|t}                                          \\
        P_{i,t+1|t}  & = \phi_{2,i}^{2}P_{i,t|t} + \psi_{2,i}^{2},
    \end{aligned}
\end{equation}
and the log likelihood is accumulated as
\begin{equation}
    \label{r3:eq:eg1-loglik}
    \mathcal{L} = -\frac{NT}{2}\log(2\pi)
                  - \frac{1}{2}\sum_{i=1}^{N}\sum_{t=1}^{T}
                    \left\{ \log F_{it} + \frac{v_{it}^{2}}{F_{it}} \right\}.
\end{equation}
This quantity is added to the target density of the NUTS-Kalman implementation once per participant, so that the $N \cdot T$ within-level latent variables $\eta_{1,it}$ never enter the sampler.

\section{Multilevel One-Factor Three-Indicator AR(1) Model}
\label{r3:sec:eg2}

This is the model of Section 4.2 of the manuscript. Table \ref{r3:tab:eg2-truth} lists the population parameters used to generate the data.

\begin{table}[H]
    \centering
    \caption{True population parameters for the multilevel one-factor three-indicator AR(1)
    model.}
    \label{r3:tab:eg2-truth}
    \begin{tabular}{lll}
        \toprule
        Parameter             & Description                                               & True value \\
        \midrule
        $\nu_{1}$             & Mean of the participant-specific intercept, indicator 1   & $2.0$ \\
        $\nu_{2}$             & Mean of the participant-specific intercept, indicator 2   & $0.6$ \\
        $\nu_{3}$             & Mean of the participant-specific intercept, indicator 3   & $0.8$ \\
        $\tau_{\nu}$          & Standard deviation of the participant-specific intercepts & $0.5$ \\
        $\phi$                & Mean of the unconstrained AR(1) effect (arctanh scale)    & $0.6$ \\
        $\tau_{\phi}$         & Standard deviation of the unconstrained AR(1) effect      & $0.1$ \\
        $\log\sigma$          & Mean of the log measurement noise standard deviation      & $0.8$ \\
        $\tau_{\log\sigma}$   & Standard deviation of the log measurement noise           & $0.1$ \\
        $\log\psi$            & Mean of the log process noise standard deviation          & $0.5$ \\
        $\tau_{\log\psi}$     & Standard deviation of the log process noise               & $0.1$ \\
        $\lambda_{1}$         & Factor loading for indicator 1 (fixed marker)             & $1.0$ \\
        $\lambda_{2}$         & Factor loading for indicator 2                            & $1.2$ \\
        $\lambda_{3}$         & Factor loading for indicator 3                            & $0.9$ \\
        \bottomrule
    \end{tabular}
\end{table}

\subsection{Priors}
\label{r3:sec:eg2-priors}

The population-level priors are
\begin{equation}
    \label{r3:eq:eg2-priors}
    \begin{aligned}
        \nu_{1}                               & \sim \mathcal{N}(0, 5^{2})     \\
        \nu_{2}, \; \nu_{3}                   & \sim \mathcal{N}(0, 1^{2})     \\
        \phi, \; \log\sigma, \; \log\psi      & \sim \mathcal{N}(0, 1^{2})     \\
        \tau_{\nu}                            & \sim \text{Half-Cauchy}(0, 2)  \\
        \tau_{\phi}                           & \sim \text{Half-Cauchy}(0, 1)  \\
        \tau_{\log\sigma}, \; \tau_{\log\psi} & \sim \text{Half-Normal}(0, 1)  \\
        \lambda_{2}, \; \lambda_{3}           & \sim \text{Half-Normal}(1, 2).
    \end{aligned}
\end{equation}
Two differences from Section \ref{r3:sec:eg1-priors} are worth noting. First, the priors on the
standard deviations of the two log noise parameters are half-normal rather than half-Cauchy.
Second, the free loadings are given half-normal priors centered at one and truncated below at
zero, which rules out the sign indeterminacy that would otherwise arise from the fact that
$(\bm{\lambda}, \eta_{1,it})$ and $(-\bm{\lambda}, -\eta_{1,it})$ imply the same distribution
for the data. As in Section \ref{r3:sec:eg1}, all three algorithms use these same priors.

\subsection{Kalman Filter Implementation}
\label{r3:sec:eg2-kalman}

The state remains scalar, but each observation is now a three-dimensional vector, so the prediction error variance $\bm{F}_{it} = P_{i,t|t-1}\bm{\lambda}\bm{\lambda}^{T} + \sigma_{2,i}^{2}\bm{I}_{3}$ is a $3 \times 3$ matrix. Rather than forming and inverting it at every timepoint, our
implementation exploits the fact that $\bm{F}_{it}$ is a rank-one update of a scaled identity matrix. Writing $c_{i} = \bm{\lambda}^{T}\bm{\lambda}/\sigma_{2,i}^{2}$ and
$\bm{v}_{it} = \bm{y}_{1,it} - \bm{\nu}_{2,i} - \bm{\lambda}a_{i,t|t-1}$, the matrix
determinant lemma and the Woodbury identity give
\begin{equation}
    \label{r3:eq:eg2-woodbury}
    \log\left|\bm{F}_{it}\right| = 6\log\sigma_{2,i} + \log\left(1 + P_{i,t|t-1}c_{i}\right),
    \qquad
    \bm{v}_{it}^{T}\bm{F}_{it}^{-1}\bm{v}_{it} = r_{it} - \frac{P_{i,t|t-1}}{1 + P_{i,t|t-1}c_{i}}q_{it}^{2},
\end{equation}
where $q_{it} = \bm{\lambda}^{T}\bm{v}_{it}/\sigma_{2,i}^{2}$ and
$r_{it} = \bm{v}_{it}^{T}\bm{v}_{it}/\sigma_{2,i}^{2}$ are scalars. The filter therefore runs entirely in scalar arithmetic. With initial condition $a_{i,1|0} = 0$ and $P_{i,1|0} = 1$, the recursion for $t = 1, \dots, T$ is
\begin{equation}
    \label{r3:eq:eg2-kalman}
    \begin{aligned}
        a_{i,t|t}   & = a_{i,t|t-1} + \frac{P_{i,t|t-1}}{1 + P_{i,t|t-1}c_{i}}q_{it} \\
        P_{i,t|t}   & = \frac{P_{i,t|t-1}}{1 + P_{i,t|t-1}c_{i}}                     \\
        a_{i,t+1|t} & = \phi_{2,i}a_{i,t|t}                                          \\
        P_{i,t+1|t} & = \phi_{2,i}^{2}P_{i,t|t} + \psi_{2,i}^{2},
    \end{aligned}
\end{equation}
and the log likelihood is
\begin{equation}
    \label{r3:eq:eg2-loglik}
    \mathcal{L} = -\frac{3NT}{2}\log(2\pi)
                  - \frac{1}{2}\sum_{i=1}^{N}\sum_{t=1}^{T}
                    \left\{ \log\left|\bm{F}_{it}\right|
                            + \bm{v}_{it}^{T}\bm{F}_{it}^{-1}\bm{v}_{it} \right\},
\end{equation}
with the two terms evaluated as in \eqref{r3:eq:eg2-woodbury}.

\section{Multilevel Trivariate VAR(1) Model}
\label{r3:sec:eg3}

This is the model of Section 4.3 of the manuscript. Table \ref{r3:tab:eg3-truth} lists the population parameters.

\begin{table}[H]
    \centering
    \caption{True population parameters for the multilevel trivariate VAR(1) model. The
    intercept means depend on the position of the indicator within its block, and the free
    loadings are only present when $K > 1$.}
    \label{r3:tab:eg3-truth}
    \begin{tabular}{lll}
        \toprule
        Parameter                         & Description                                               & True value \\
        \midrule
        $\nu_{(j-1)K+1}$                  & Mean intercept, first indicator of each block             & $2.0$ \\
        $\nu_{(j-1)K+2}$                  & Mean intercept, second indicator of each block            & $0.6$ \\
        $\nu_{(j-1)K+3}$                  & Mean intercept, third indicator of each block             & $0.8$ \\
        $\tau_{\nu[u]}$                   & Standard deviation of each participant-specific intercept & $0.5$ \\
        $\phi_{rr}$                       & Mean unconstrained auto-effect (arctanh scale)            & $0.30$ \\
        $\phi_{rc}$, $|r-c|=1$            & Mean unconstrained adjacent cross-effect                  & $0.10$ \\
        $\phi_{rc}$, $|r-c|=2$            & Mean unconstrained distant cross-effect                   & $0.05$ \\
        $\tau_{\phi[r,c]}$                & Standard deviation of each unconstrained VAR(1) effect    & $0.05$ \\
        $\log\sigma$                      & Mean of the log measurement noise standard deviation      & $0.8$ \\
        $\tau_{\log\sigma}$               & Standard deviation of the log measurement noise           & $0.1$ \\
        $\log\psi$                        & Mean of the log process noise standard deviation          & $0.5$ \\
        $\tau_{\log\psi}$                 & Standard deviation of the log process noise               & $0.1$ \\
        $\lambda_{1}, \dots, \lambda_{6}$ & Free factor loadings, in order of appearance              & $1.2$, $0.9$, $1.1$, \\
                                          &                                                           & $0.8$, $1.3$, $0.85$ \\
        \bottomrule
    \end{tabular}
\end{table}

\subsection{Priors}
\label{r3:sec:eg3-priors}

The non-centered parameterization of Section \ref{r3:sec:eg1-priors} was used again. The
population-level priors are
\begin{equation}
    \label{r3:eq:eg3-priors}
    \begin{aligned}
        \nu_{(j-1)K+1}                        & \sim \mathcal{N}(0, 5^{2})     \\
        \nu_{(j-1)K+k}, \; k > 1              & \sim \mathcal{N}(0, 1^{2})     \\
        \phi_{rc}, \; \log\sigma, \; \log\psi & \sim \mathcal{N}(0, 1^{2})     \\
        \tau_{\nu[u]}                         & \sim \text{Half-Cauchy}(0, 2)  \\
        \tau_{\phi[r,c]}                      & \sim \text{Half-Cauchy}(0, 1)  \\
        \tau_{\log\sigma}, \; \tau_{\log\psi} & \sim \text{Half-Normal}(0, 1)  \\
        \lambda_{1}, \dots, \lambda_{3(K-1)}  & \sim \text{Half-Normal}(1, 2),
    \end{aligned}
\end{equation}
following the same pattern as Section \ref{r3:sec:eg2-priors}: the marker indicator of each block
receives the diffuse $\mathcal{N}(0, 5^{2})$ prior appropriate for an unstandardized intercept,
the remaining intercepts receive $\mathcal{N}(0, 1^{2})$, the noise scale parameters receive
half-normal priors, and the free loadings are truncated below at zero to remove the sign
indeterminacy. Both algorithms use these same priors.

\subsection{Kalman Filter Implementation}
\label{r3:sec:eg3-kalman}

The state is now three-dimensional and the observation is $U$-dimensional, so a direct
implementation would require inverting a $U \times U$ prediction error variance matrix at
every timepoint, which for $K = 3$ is $9 \times 9$. Our implementation avoids this by first
reducing each observation to a three-dimensional sufficient statistic.

Because $\bm{\Lambda}$ has the block structure described above,
$\bm{\Lambda}^{T}\bm{\Lambda} = \bm{C} = \text{diag}(c_{1}, c_{2}, c_{3})$ is diagonal, with
$c_{j} = 1 + \sum_{k=1}^{K-1}\lambda_{(j-1)(K-1)+k}^{2}$. Writing
$\bm{d}_{it} = \bm{y}_{1,it} - \bm{\nu}_{2,i}$ and
$\bm{q}_{it} = \bm{\Lambda}^{T}\bm{d}_{it}$, the generalized least squares projection of the
observation onto the column space of $\bm{\Lambda}$ is
$\tilde{\bm{y}}_{it} = \bm{C}^{-1}\bm{q}_{it}$, and
\begin{equation}
    \label{r3:eq:eg3-decomp}
    \left\| \bm{d}_{it} - \bm{\Lambda}\bm{\eta}_{1,it} \right\|^{2}
    = o_{it}
      + \left(\tilde{\bm{y}}_{it} - \bm{\eta}_{1,it}\right)^{T}
        \bm{C}\left(\tilde{\bm{y}}_{it} - \bm{\eta}_{1,it}\right),
    \qquad
    o_{it} = \bm{d}_{it}^{T}\bm{d}_{it} - \bm{q}_{it}^{T}\bm{C}^{-1}\bm{q}_{it}.
\end{equation}
The residual sum of squares $o_{it}$ is the part of the observation orthogonal to the column
space of $\bm{\Lambda}$. It does not involve $\bm{\eta}_{1,it}$, and therefore contributes only
an additive constant to the log likelihood. The remaining term shows that
$\tilde{\bm{y}}_{it}$ is a sufficient statistic for $\bm{\eta}_{1,it}$, distributed as
$\mathcal{N}(\bm{\eta}_{1,it}, \sigma_{2,i}^{2}\bm{C}^{-1})$. The filter therefore runs on the
three-dimensional series $\tilde{\bm{y}}_{it}$ with diagonal measurement error variance
$\bm{R}_{i} = \sigma_{2,i}^{2}\bm{C}^{-1}$, irrespective of how many indicators there are.

With initial condition $\bm{a}_{i,1|0} = \bm{0}$ and $\bm{P}_{i,1|0} = \bm{I}_{3}$, the
recursion for $t = 1, \dots, T$ is
\begin{equation}
    \label{r3:eq:eg3-kalman}
    \begin{aligned}
        \bm{W}_{it}         & = \bm{P}_{i,t|t-1} + \bm{R}_{i}                                                       \\
        \bm{a}_{i,t|t}      & = \bm{a}_{i,t|t-1}
                                + \left(\bm{I}_{3} - \bm{R}_{i}\bm{W}_{it}^{-1}\right)
                                  \left(\tilde{\bm{y}}_{it} - \bm{a}_{i,t|t-1}\right)                               \\
        \bm{P}_{i,t|t}      & = \bm{R}_{i} - \bm{R}_{i}\bm{W}_{it}^{-1}\bm{R}_{i}                                   \\
        \bm{a}_{i,t+1|t}    & = \bm{\Phi}_{2,i}\bm{a}_{i,t|t}                                                       \\
        \bm{P}_{i,t+1|t}    & = \bm{\Phi}_{2,i}\bm{P}_{i,t|t}\bm{\Phi}_{2,i}^{T} + \psi_{2,i}^{2}\bm{I}_{3},
    \end{aligned}
\end{equation}
where $\bm{W}_{it}$ is inverted through its Cholesky factor, which also supplies
$\log|\bm{W}_{it}|$. The log likelihood is
\begin{equation}
    \label{r3:eq:eg3-loglik}
    \begin{aligned}
        \mathcal{L} = {} & -\frac{UNT}{2}\log(2\pi)
                           - (U-3)NT\log\sigma_{2,i}
                           - \frac{NT}{2}\sum_{j=1}^{3}\log c_{j}                                       \\
                         & - \frac{1}{2}\sum_{i=1}^{N}\sum_{t=1}^{T}
                             \left\{ \log\left|\bm{W}_{it}\right|
                                     + \frac{o_{it}}{\sigma_{2,i}^{2}}
                                     + \left(\tilde{\bm{y}}_{it} - \bm{a}_{i,t|t-1}\right)^{T}
                                       \bm{W}_{it}^{-1}
                                       \left(\tilde{\bm{y}}_{it} - \bm{a}_{i,t|t-1}\right) \right\},
    \end{aligned}
\end{equation}
where the terms involving $\sigma_{2,i}$ are of course evaluated per participant. The largest
matrix inverted is thus $3 \times 3$ regardless of $K$. We verified numerically that this
reduction returns the same log likelihood as a Kalman filter operating directly on the
$U$-dimensional observations, for $K = 1, 2$ and $3$.

\subsection{Initial Values}
\label{r3:sec:eg3-settings}

Unlike Sections \ref{r3:sec:eg1} and \ref{r3:sec:eg2}, this model requires initial values for both \textsc{Stan} implementations. Under the default $\mathcal{U}(-2,2)$ initialization on the unconstrained scale, $\tau_{\phi[r,c]}$ can start as large as $e^{2}$, which places the individual autoregression matrices far outside the stationary region and causes the latent processes to diverge during warm-up. We therefore started all chains at $\phi_{rc} = 0$ for all $r$ and $c$ with $\tau_{\phi} = 0.05\bm{J}_{3}$, at $\nu_{u} = 0$ for all $u$ with $\tau_{\nu} = 0.1\bm{1}_{U}$, at
$\log\sigma = \log\psi = 0$ with $\tau_{\log\sigma} = \tau_{\log\psi} = 0.1$, at $\lambda_{m} = 1$ for all free loadings, and at zero for all standardized deviations $z$. These values correspond to a stationary system with
unit measurement and process noise scale. The Metropolis-within-Gibbs sampler was started by sampling from the priors.

\section{Multilevel Latent AR(p) Model}
\label{r3:sec:eg4}

This is the model of Section 4.4 of the manuscript. Table \ref{r3:tab:eg4-truth}
lists the population parameters. Only the first $p$ of the four autoregressive means and standard deviations are used in any given condition.

\begin{table}[H]
    \centering
    \caption{True population parameters for the multilevel latent AR($p$) model.}
    \label{r3:tab:eg4-truth}
    \begin{tabular}{lll}
        \toprule
        Parameter             & Description                                              & True value \\
        \midrule
        $\nu$                 & Mean of the participant-specific intercept               & $2.0$ \\
        $\tau_{\nu}$          & Standard deviation of the participant-specific intercept & $0.5$ \\
        $\phi_{1}$            & Mean of the lag-1 autoregressive effect                  & $0.40$ \\
        $\phi_{2}$            & Mean of the lag-2 autoregressive effect                  & $0.20$ \\
        $\phi_{3}$            & Mean of the lag-3 autoregressive effect                  & $0.10$ \\
        $\phi_{4}$            & Mean of the lag-4 autoregressive effect                  & $0.05$ \\
        $\tau_{\phi_{l}}$      & Standard deviation of each autoregressive effect         & $0.05$ \\
        $\log\sigma$          & Mean of the log measurement noise standard deviation     & $0.8$ \\
        $\tau_{\log\sigma}$   & Standard deviation of the log measurement noise          & $0.1$ \\
        $\log\psi$            & Mean of the log process noise standard deviation         & $0.5$ \\
        $\tau_{\log\psi}$     & Standard deviation of the log process noise              & $0.1$ \\
        \bottomrule
    \end{tabular}
\end{table}

\subsection{Priors}
\label{r3:sec:eg4-priors}

The non-centered parameterization of Section \ref{r3:sec:eg1-priors} was used again, with
population-level priors
\begin{equation}
    \label{r3:eq:eg4-priors}
    \begin{aligned}
        \nu                                                      & \sim \mathcal{N}(0, 5^{2})     \\
        \phi_{l}, \; \log\sigma, \; \log\psi                     & \sim \mathcal{N}(0, 1^{2})     \\
        \tau_{\nu}                                               & \sim \text{Half-Cauchy}(0, 2)  \\
        \tau_{\phi[l]}, \; \tau_{\log\sigma}, \; \tau_{\log\psi} & \sim \text{Half-Cauchy}(0, 1).
    \end{aligned}
\end{equation}
These are exactly the priors of Section \ref{r3:sec:eg1-priors}, extended to the $p$ lags. In
particular, and in contrast to Sections \ref{r3:sec:eg2-priors} and \ref{r3:sec:eg3-priors}, the
standard deviations of the two log noise parameters carry half-Cauchy rather than half-normal
priors. All three algorithms use these same priors.

\subsection{Kalman Filter Implementation}
\label{r3:sec:eg4-kalman}

The AR($p$) process is Markovian only in the stacked state
$\tilde{\bm{\eta}}_{1,it} = (\eta_{1,it}, \eta_{1,i,t-1}, \dots, \eta_{1,i,t-p+1})^{T}$, which
is the augmentation of Theorem 1 specialized to a single latent variable, no lagged loadings
and no regressions on lagged manifest variables. The transition matrix is the companion form
\begin{equation}
    \label{r3:eq:eg4-companion}
    \bm{T}_{i} =
    \begin{pmatrix}
        \phi_{2,1i} & \phi_{2,2i} & \cdots & \phi_{2,p-1,i} & \phi_{2,pi} \\
        1           & 0           & \cdots & 0              & 0           \\
        0           & 1           & \cdots & 0              & 0           \\
        \vdots      &             & \ddots &                & \vdots      \\
        0           & 0           & \cdots & 1              & 0
    \end{pmatrix},
    \qquad
    \bm{W}_{i} = \psi_{2,i}^{2}\bm{e}_{1}\bm{e}_{1}^{T},
    \qquad
    \bm{Z} = \bm{e}_{1}^{T},
\end{equation}
where $\bm{e}_{1}$ is the first standard basis vector, so that only the leading component of
the state receives process noise and only it is measured. The initial condition is
$\bm{a}_{i,1|0} = \bm{0}$ and $\bm{P}_{i,1|0} = \bm{I}_{p}$.

Because the first $p$ latent states are drawn independently from $\mathcal{N}(0,1)$ rather
than from the autoregressive recursion, the transition in \eqref{r3:eq:eg4-companion} does not
apply until the recursion has started. For $t < p$ the filter therefore shifts the state
vector down by one position and inserts a fresh $\mathcal{N}(0,1)$ component at the top,
\begin{equation}
    \label{r3:eq:eg4-startup}
    \bm{a}_{i,t+1|t} = \bm{S}\bm{a}_{i,t|t},
    \qquad
    \bm{P}_{i,t+1|t} = \bm{S}\bm{P}_{i,t|t}\bm{S}^{T} + \bm{e}_{1}\bm{e}_{1}^{T},
    \qquad t < p,
\end{equation}
where $\bm{S}$ is the down-shift matrix with ones on the first subdiagonal and zeros
elsewhere, and switches to $\bm{T}_{i}$ and $\bm{W}_{i}$ for $t \geq p$. Since the observation
is scalar, the update step requires no matrix inversion:
\begin{equation}
    \label{r3:eq:eg4-kalman}
    \begin{aligned}
        v_{it}                & = y_{1,it} - \left(\nu_{2,i} + \bm{e}_{1}^{T}\bm{a}_{i,t|t-1}\right) \\
        F_{it}                & = \bm{e}_{1}^{T}\bm{P}_{i,t|t-1}\bm{e}_{1} + \sigma_{2,i}^{2}        \\
        \bm{a}_{i,t|t}        & = \bm{a}_{i,t|t-1} + \bm{P}_{i,t|t-1}\bm{e}_{1}F_{it}^{-1}v_{it}     \\
        \bm{P}_{i,t|t}        & = \bm{P}_{i,t|t-1}
                                  - \bm{P}_{i,t|t-1}\bm{e}_{1}F_{it}^{-1}\bm{e}_{1}^{T}\bm{P}_{i,t|t-1},
    \end{aligned}
\end{equation}
and the log likelihood has the same form as \eqref{r3:eq:eg1-loglik}. For $p \leq 4$ the
implementation uses hand-unrolled scalar code rather than matrix operations, which avoids the
overhead of allocating and multiplying small matrices at every one of the $N \cdot T$
timepoints; a general matrix implementation is retained for larger $p$.

\subsection{Initial Values}
\label{r3:sec:eg4-settings}

Both \textsc{Stan} implementations were supplied with initial values. Under the default $\mathcal{U}(-2,2)$ initialization on the unconstrained scale, $\tau_{\phi[l]}$ can start as large as $e^{2}$, which places the implied autoregressive coefficients far outside the stationary region and produces an explosive process over $T$ timepoints. All chains were therefore started at $\phi_{l} = 0$ for all $l$ with $\tau_{\phi} = 0.05\bm{1}_{p}$, at $\nu = 0$ with $\tau_{\nu} = 0.1$, at $\log\sigma = \log\psi = 0$ with
$\tau_{\log\sigma} = \tau_{\log\psi} = 0.1$, and at zero for all standardized deviations $z$. These values correspond to a white noise process with unit measurement and process noise scale, which is well inside the stationary region for every lag considered. The Metropolis-within-Gibbs sampler was started by sampling from the priors.

\section{Cross-Classified VAR(1) Model}
\label{r3:sec:eg5}

This is the model of Section 4.5 of the manuscript. Table \ref{r3:tab:eg5-truth} lists the population parameters.

\begin{table}[H]
    \centering
    \caption{True population parameters for the cross-classified VAR(1) model. The means of the
    autoregressive effects are given on the $\tanh^{-1}$ scale, on which they are sampled; the
    corresponding autoregression matrix is
    $\bm{\Phi} = \left(\begin{smallmatrix}0.75 & 0.10 \\ 0.10 & 0.50\end{smallmatrix}\right)$.}
    \label{r3:tab:eg5-truth}
    \begin{tabular}{lll}
        \toprule
        Parameter                        & Description                                                     & True value \\
        \midrule
        $\nu_{1},\nu_{2},\nu_{3}$        & Population means of the intercepts of items 1--3                 & $3.0,\;3.2,\;2.8$ \\
        $\nu_{4}$                        & Population mean of the intercept of the fourth item             & $2.5$ \\
        $\lambda_{1},\lambda_{2},\lambda_{3}$ & Population means of the factor loadings                     & $1.0,\;0.9,\;1.1$ \\
        $\phi_{11}$                      & Mean autoregressive effect, first on first                      & $0.973$ \\
        $\phi_{21}$                      & Mean cross-lagged effect, first on second                       & $0.100$ \\
        $\phi_{12}$                      & Mean cross-lagged effect, second on first                       & $0.100$ \\
        $\phi_{22}$                      & Mean autoregressive effect, second on second                    & $0.549$ \\
        $\log\sigma$                     & Mean log measurement noise standard deviation                   & $0.405$ \\
        $\log\psi$                       & Mean log process noise standard deviation                       & $-1.609$ \\
        \midrule
        $\tau_{3,1},\tau_{3,2},\tau_{3,3}$ & Between-timepoint SDs of the intercepts of items 1--3          & $0.4$ \\
        $\tau_{3,5},\tau_{3,6},\tau_{3,7}$ & Between-timepoint SDs of the factor loadings                   & $0.2$ \\
        $\tau_{2,1},\dots,\tau_{2,4}$    & Between-participant SDs of the four intercepts                  & $0.5$ \\
        $\tau_{2,5},\tau_{2,6},\tau_{2,7}$ & Between-participant SDs of the three loadings                  & $0.2$ \\
        $\tau_{2,8},\dots,\tau_{2,11}$   & Between-participant SDs of the four autoregressive effects       & $0.05$ \\
        $\tau_{2,12},\tau_{2,13}$        & Between-participant SDs of the two log noise parameters          & $0.2$ \\
        \bottomrule
    \end{tabular}
\end{table}

\subsection{Priors}
\label{r3:sec:eg5-priors}

The non-centered parameterization was used for both the participant-specific and the
timepoint-specific latent variables, so that $\eta_{2,iv}$ equals the corresponding element of
$\bm{\alpha}_{2}$ plus $\tau_{2,v}z_{2,iv}$, while the between-timepoint variables are centered
at zero, $\eta_{3,tv} = \tau_{3,v}z_{3,tv}$, with $z_{2,iv}$ and $z_{3,tv}$ standard normal. In
the notation of Section 4.5 of the manuscript,
$\bm{\Psi}_{2} = \text{diag}(\tau_{2,1}^{2},\dots,\tau_{2,13}^{2})$ and
$\bm{\Psi}_{3} = \text{diag}(\tau_{3,1}^{2},\tau_{3,2}^{2},\tau_{3,3}^{2},0,\tau_{3,5}^{2},\tau_{3,6}^{2},\tau_{3,7}^{2})$.
The population-level priors are
\begin{equation}
    \label{r3:eq:eg5-priors}
    \begin{aligned}
        \nu_{1},\dots,\nu_{4}                 & \sim \mathcal{N}(0, 5^{2})                                 \\
        \lambda_{1},\lambda_{2},\lambda_{3}   & \sim \mathcal{N}(1, 2^{2})\;\text{truncated to } (0,\infty) \\
        \phi_{rc},\; \log\sigma,\; \log\psi   & \sim \mathcal{N}(0, 1^{2})                                 \\
        \tau_{2,v},\; \tau_{3,v}              & \sim \text{Half-Normal}(0, 1).
    \end{aligned}
\end{equation}
Unlike Sections \ref{r3:sec:eg1-priors} and \ref{r3:sec:eg4-priors}, all standard deviations carry
half-normal rather than half-Cauchy priors. The population means of the factor loadings are
constrained to be non-negative, which removes the sign indeterminacy between the loadings and
the latent process. All three algorithms use these same priors.

\subsection{Kalman Filter Implementation}
\label{r3:sec:eg5-kalman}

The state is the bivariate $\bm{\eta}_{1,it}$ itself, so no augmentation is required. The transition and process noise matrices are $\bm{\Phi}_{i}$ and
$\operatorname{diag}(1, \psi_{1,i}^{2})$, and the initial condition is $\bm{a}_{i,1|0} = \bm{0}$ with $\bm{P}_{i,1|0} = \operatorname{diag}(1, \psi_{1,i}^{2})$.

Rather than updating on the four-dimensional observation at once, the implementation processes the four items sequentially within each timepoint. This is the univariate treatment of multivariate observations \citep[Ch.~6.4]{durbin_time_2012}: because $\bm{\Sigma}_{1,i}$ is diagonal, conditioning on the items one at a time gives the same filtered distribution as a single joint update, and every update is scalar, so no matrix inversion is needed anywhere in the filter. For the first three items the update is the usual one with $F_{itm} = \lambda_{itm}^{2}P_{11} + \sigma_{1,i}^{2}$.

The fourth item requires care. Since it is observed without error, its prediction variance is $F_{it4} = P_{22}$ with no measurement noise added. Conditioning on it therefore determines the second state component exactly: after the update, $P_{12}$ and $P_{22}$ are identically zero and the filtered covariance has rank one. The
prediction step then simplifies to
\begin{equation}
    \label{r3:eq:eg5-predict}
    \begin{aligned}
        P_{11,t+1|t} & = \Phi_{i,11}^{2}P_{11,t|t} + 1, \qquad
        P_{12,t+1|t}   = \Phi_{i,11}\Phi_{i,21}P_{11,t|t},                         \\
        P_{22,t+1|t} & = \Phi_{i,21}^{2}P_{11,t|t} + \psi_{1,i}^{2},
    \end{aligned}
\end{equation}
since the terms involving $P_{12,t|t}$ and $P_{22,t|t}$ vanish. A degenerate prediction variance is a genuine feature of the model rather than a numerical accident, but it is numerically delicate, and $F_{it4}$ is floored at $10^{-12}$ to avoid dividing by zero at the first timepoint of a participant whose $\psi_{1,i}$ is very small.

The same structure is exploited by the Metropolis-within-Gibbs implementation, which samples only $\eta_{1,it1}$ and recovers the second component deterministically as
$\eta_{1,it2} = y_{1,it4} - \eta_{2,i4}$. It therefore samples $N \cdot T$ latent values rather than $2N \cdot T$, and the comparison in Section 4.5 is against this more favorable baseline. NUTS-Joint likewise samples one latent series per participant. 

\subsection{Initial Values}
\label{r3:sec:eg5-settings}

Both \textsc{Stan} implementations were supplied with initial values, drawn afresh for each chain so that the chains remain dispersed. The values use no information that would be unavailable in a real analysis: intercepts start near the mean of their prior, loadings near one, all standard deviations at $\mathcal{U}(0.05, 0.15)$, the autoregressive means near zero on the $\tanh^{-1}$ scale, which corresponds to a process with no dynamics and is stationary by construction, the log noise means near zero, and all standardized deviations near zero. The Metropolis-within-Gibbs sampler was started by sampling from the priors.

\section{Measurement of Computation Time}
\label{r3:sec:timing}

All computation times reported in the manuscript and in Online Resource 4 are elapsed wall-clock times, measured in \textsc{R} by recording \texttt{Sys.time()} immediately before and immediately after a single function call which performs all phases of MCMC sampling. The same clock and the same convention are used for all three algorithms, in all five simulation experiments and in the application example, so that a reported time means the same thing regardless of the underlying engine. Times were recorded in seconds and are reported in minutes in Section 4 of the main text and in Online Resource 4, and in hours in Section 5 of the main text. The scripts implementing what is described below are available in our OSF repository, in the file \texttt{simulate.R} within each simulation experiment's folder and in \texttt{code/application/fit.R} for the application example.

\subsection{Simulation Experiments}
\label{r3:sec:timing-simulations}

Each combination of algorithm, simulation setting, and Monte Carlo sample was run as a separate job on the computing cluster described in Section 4 of the main text, with four CPU cores per job and \textsc{R} version 4.4.2. In every run, the four chains were executed in parallel, one chain per core (\texttt{parallel\_chains = 4} in \textsc{cmdstanr}; \texttt{parallel = TRUE} in \textsc{jagsUI}), so the measured wall time is in effect the time taken by the slowest chain.

For NUTS-Kalman and NUTS-Joint, the timed call is \texttt{\$sample()}, which runs the 1{,}000 warm-up iterations followed by the 10{,}000 sampling iterations in each chain. Warm-up is therefore included in the measured time. Compilation of the \textsc{Stan} programs is not: the models were compiled once, before any simulation job was submitted, and the simulation script loads the compiled program through \texttt{cmdstan\_model(exe\_file = ...)}, which cannot trigger compilation. Since a compiled \textsc{Stan} model is a reusable binary, its compilation is a one-time cost per model rather than a cost per analysis, and we accordingly treat it as we treat the writing of the model code itself.

For Metropolis-within-Gibbs, the timed call is \texttt{jagsUI::jags()}, which constructs the model graph, runs the 2{,}000-iteration adaptation phase, and draws the 100{,}000 retained iterations in each chain; in the model of Section 4.5 the additional burn-in of 20{,}000 iterations per chain also falls inside the timed call. \textsc{JAGS} has no separate compilation artifact---the graph is rebuilt in every run---so graph construction is part of its measured time; in these simulations it takes on the order of seconds.

Everything else falls outside the timed call: simulating the data, computing effective sample sizes and $\hat{R}$ with the \textsc{posterior} package, and writing results to disk. Because the warm-up and adaptation phases are short relative to the retained sampling (1{,}000 versus 10{,}000 iterations per chain for the NUTS-based algorithms, and 2{,}000 versus 100{,}000 for Metropolis-within-Gibbs), the measured times are dominated by sampling from the (approximate) posterior, as discussed in Section 4 of the main text.

Sampling efficiency was computed within each Monte Carlo sample as the bulk-ESS (or tail-ESS) divided by the total wall time in minutes, separately for each parameter of interest, and the minimum of these ratios over the parameters of interest is the efficiency of that run. Figures and tables report medians (and means) of these per-sample quantities across the Monte Carlo samples. The wall-time tables of Online Resource 4 aggregate the times themselves in the same way, which is why dividing the ESS tables by the time tables reproduces the plotted efficiencies only up to the difference between a ratio of medians and a median of ratios. The times to reach a target ESS reported in Tables 1 and 2 of the main text are, correspondingly, medians across Monte Carlo samples of the per-sample time implied by each run's minimum ESS rate.

\subsection{Application Example}
\label{r3:sec:timing-application}

The application example of Section 5 of the main text, whose model is specified in Section 5.1, follows the same convention. The two \textsc{Stan} models were compiled in a separate job before fitting, and the timed call is again \texttt{\$sample()}, covering the 1{,}000 warm-up and 2{,}000 sampling iterations of each of the four parallel chains. All three algorithms were supplied with the same data-informed initial values, and the two \textsc{Stan} runs additionally used a data-informed diagonal inverse metric with initial step size $0.01$, which shortens the first warm-up phase but does not alter the posterior; both are provided by \texttt{code/application/helpers.R} in the OSF repository.

For Metropolis-within-Gibbs, no single sampling call exists, because the production run is segmented in order to respect the wall-time limits of the cluster: each of the four chains runs in its own forked \textsc{rjags} process, which constructs the model graph, runs the 2{,}000 adaptation iterations, and then samples in segments of 500 iterations, saving all draws to disk after every segment. After each segment, the chain compares the duration of that segment with the remainder of its 88-hour budget and stops when another segment might not complete in time, using one and a half segment durations as the safety margin; graph construction and adaptation count against the budget. The reported wall time is measured around the complete parallel run of all four chains. The chains were then truncated to their common number of retained iterations, and the first 5{,}000 iterations were discarded as burn-in for all summaries. Effective sample sizes are thus computed from post-burn-in draws while the reported time covers the complete run---the same asymmetry that applies to the two other algorithms, whose warm-up and adaptation phases are also inside the measured time but excluded from the ESS computation. The times to reach a target ESS in Table 3 of the main text are projections obtained by dividing the target by each algorithm's minimum ESS per minute over the parameters of interest.

\onlineresource{4}

This document contains supplementary simulation results. 

In addition to the results presented here, our OSF repository\footnote{\url{https://osf.io/p754m/overview?view_only=f99de4145ddf49c9a76f0e38b5e42223}} contains trace plots for all parameters of interest from three instances of each simulation setting. The instances selected were the one with the lowest sampling efficiency, the sampling efficiency closest to the median, and the highest sampling efficiency, as measured by bulk-ESS per minute. The trace plots can be found in the following folders:

\begin{itemize}
    \item \verb!code/eg1-univariate-latent-ar1/trace_plot_simulations/trace_plots!
    \item \verb!code/eg2-three-indicator-ar1/trace_plot_simulations/trace_plots!
    \item \verb!code/eg3-trivariate-var1/trace_plot_simulations/trace_plots!
    \item \verb!code/eg4-latent-arp/trace_plot_simulations/trace_plots!
    \item \verb!code/eg5-cross-classified/trace_plot_simulations/trace_plots!
\end{itemize}

In general the trace plots show that all algorithms converge well, both when the sampling efficiency is high and low. When comparing trace plots across the three runs, note that the posterior distribution differs between the runs. That is, the traces for a given parameter cannot be expected to be centered around the same mean. 

\section{Multilevel Scalar Latent AR(1) Model}
\label{r4:sec:eg1}

Figure \ref{r4:fig:eg1-univariate-latent-ar1-rhat-success} shows the proportion of successful runs, with different $\hat{R}$ thresholds. Figure \ref{r4:fig:eg1-univariate-latent-ar1-tail-efficiency} shows the sampling efficiency compared in terms of tail-ESS per time unit. Figures \ref{r4:fig:eg1-univariate-latent-ar1-bulk-ess-all-pars} and \ref{r4:fig:eg1-univariate-latent-ar1-tail-ess-all-pars} show the bulk-ESS and tail-ESS per time unit for all parameters separately. 

Finally, Figure \ref{r4:fig:eg1-agreement} compares the three algorithms directly. Since all three target the same posterior for a given dataset, the comparison is made within each Monte Carlo sample rather than between distributions of estimates across samples: each algorithm is compared with the average of the three, and the difference is expressed in units of the posterior standard deviation so that it is comparable across parameters. Exact agreement is not to be expected, since each algorithm estimates the posterior with Monte Carlo error; what the figure should show, and does, is that the deviations are centred on zero and that the posterior standard deviations agree. The width of the boxes reflects how precisely each algorithm estimates the posterior, which is the subject of the efficiency comparisons above, rather than any disagreement about what the posterior is.

Table \ref{r4:tab:eg1-univariate-latent-ar1-ess} shows effective sample sizes and Table \ref{r4:tab:eg1-univariate-latent-ar1-time-total} shows wall times.

\begin{figure}
    \centering
    \includegraphics[width=\linewidth]{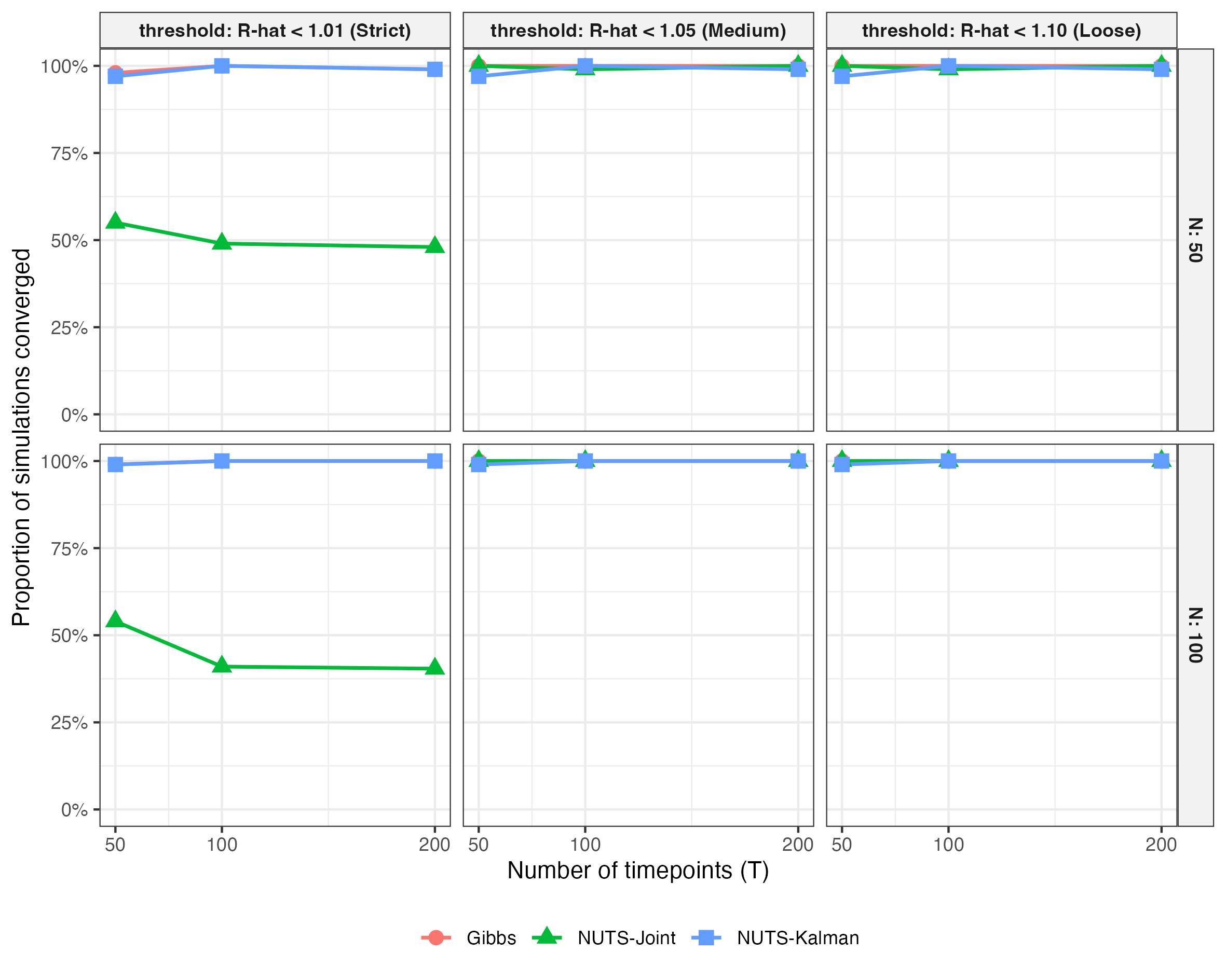}
    \caption{Proportion of simulations for which $\hat{R}$ was below the given thresholds.}
    \label{r4:fig:eg1-univariate-latent-ar1-rhat-success}
\end{figure}

\begin{figure}
    \centering
    \includegraphics[width=\linewidth]{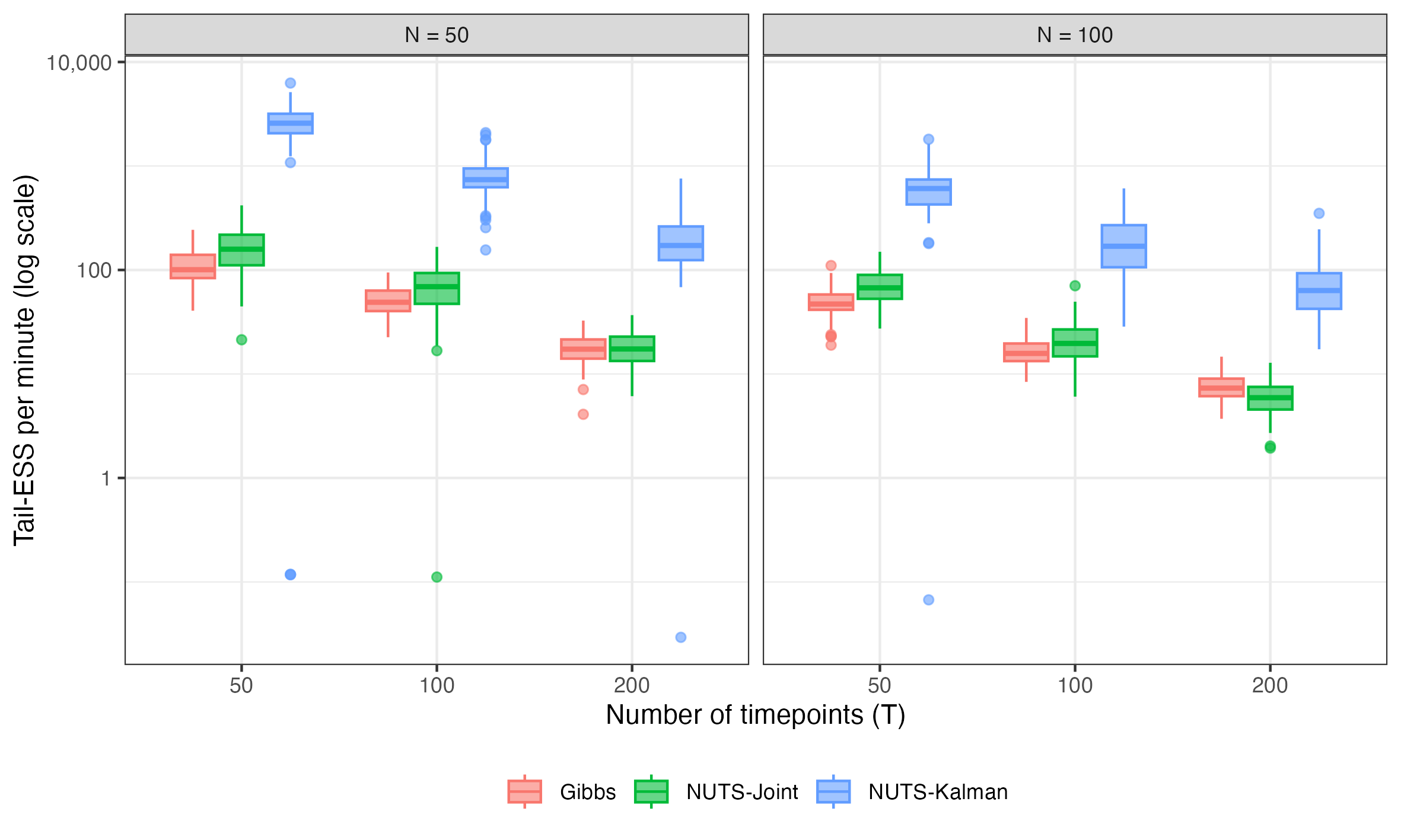}
    \caption{Box plots of tail efficiency for the three algorithms compared for the multilevel scalar latent AR(1) model.}
    \label{r4:fig:eg1-univariate-latent-ar1-tail-efficiency}
\end{figure}

\begin{figure}
    \centering
    \includegraphics[width=\linewidth]{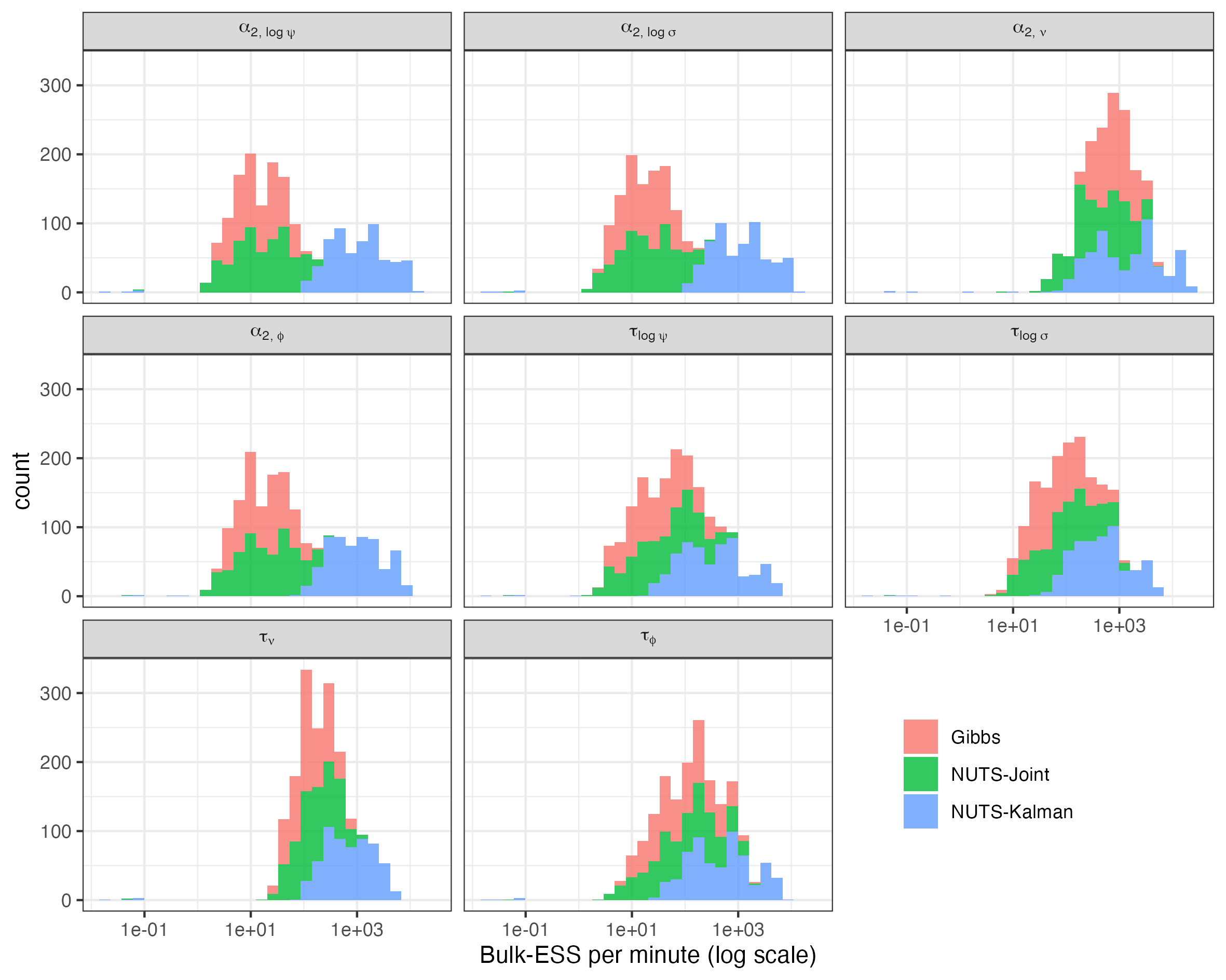}
    \caption{Sampling efficiency in terms of bulk-ESS for the three algorithms compared for the multilevel scalar latent AR(1) model.}
    \label{r4:fig:eg1-univariate-latent-ar1-bulk-ess-all-pars}
\end{figure}

\begin{figure}
    \centering
    \includegraphics[width=\linewidth]{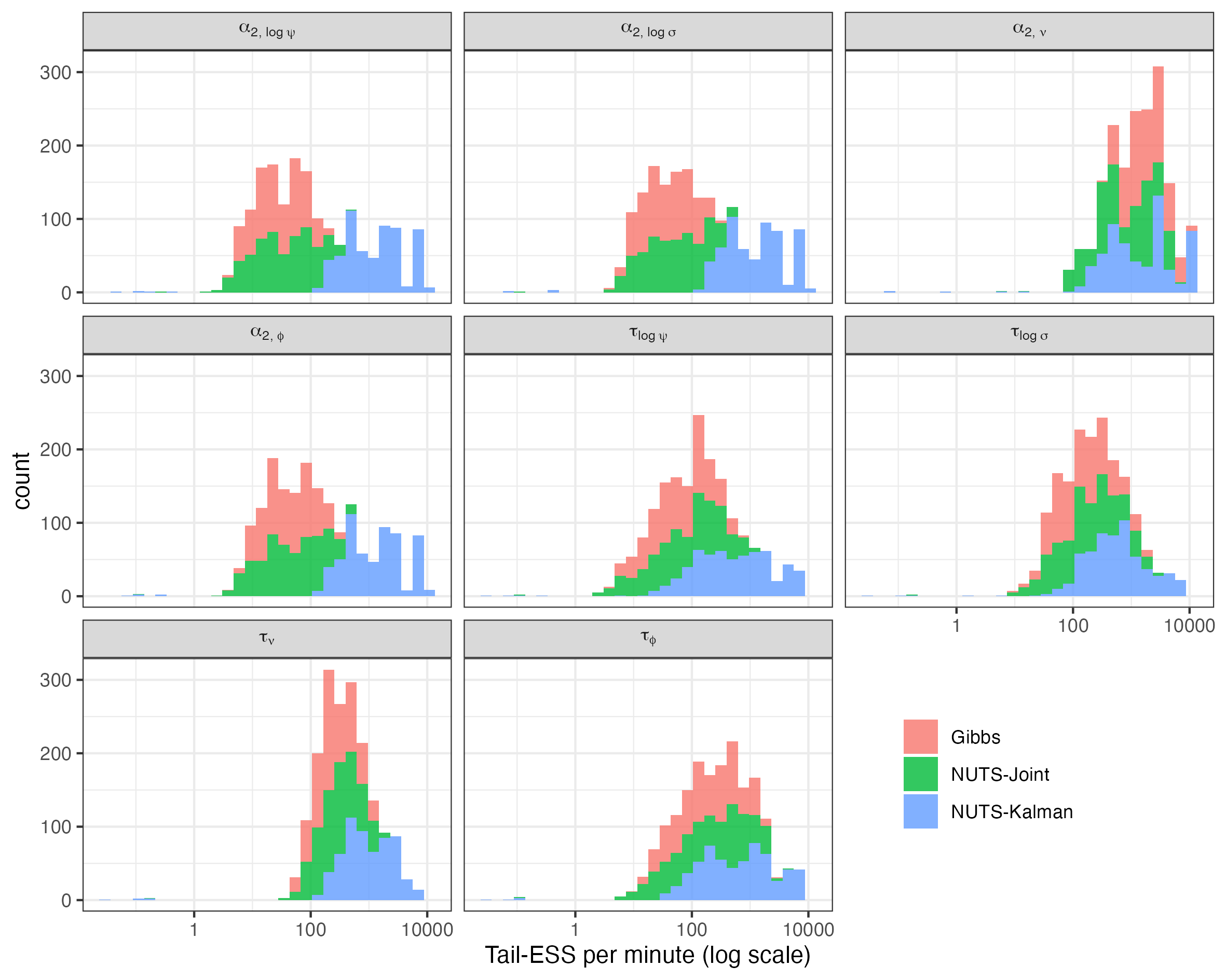}
    \caption{Sampling efficiency in terms of tail-ESS for the three algorithms compared for the multilevel scalar latent AR(1) model.}
    \label{r4:fig:eg1-univariate-latent-ar1-tail-ess-all-pars}
\end{figure}

\begin{landscape}
\begin{figure}[H]
    \centering
    \includegraphics[width=\linewidth, height=0.88\textheight, keepaspectratio]{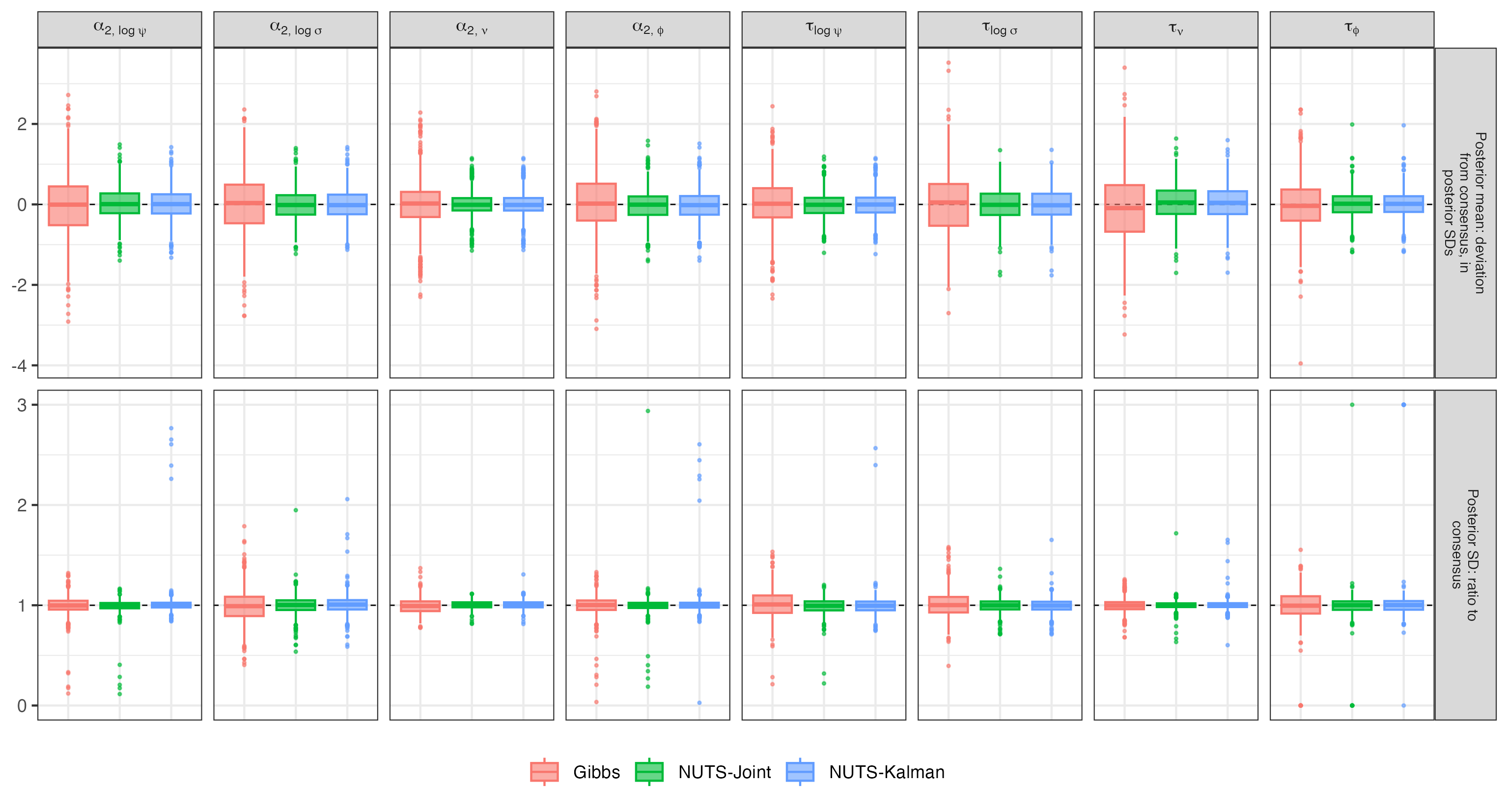}
    \caption{Agreement between the three algorithms for the multilevel scalar latent AR(1) model. Each algorithm is compared with the average of the three within each Monte Carlo sample, so the comparison is paired. The upper row shows the deviation of the posterior mean from that average, in units of the posterior standard deviation; the lower row shows the ratio of the posterior standard deviation to the same average. Boxes centred on the dashed lines indicate that the algorithms target the same posterior, and their width reflects Monte Carlo error rather than disagreement.}
    \label{r4:fig:eg1-agreement}
\end{figure}
\end{landscape}

\begin{table}
\centering
\caption{\label{r4:tab:eg1-univariate-latent-ar1-ess}Effective sample size for the multilevel scalar latent AR(1) model of Section 4.1. Within each Monte Carlo sample the minimum across the parameters of interest is taken, and the mean and median of those minima across Monte Carlo samples are reported.}
\centering
\begin{tabular}[t]{rrlrrrr}
\toprule
\multicolumn{3}{c}{ } & \multicolumn{2}{c}{Bulk-ESS} & \multicolumn{2}{c}{Tail-ESS} \\
\cmidrule(l{3pt}r{3pt}){4-5} \cmidrule(l{3pt}r{3pt}){6-7}
N & T & Method & Mean & Median & Mean & Median\\
\midrule
50 & 50 & Gibbs & 1395 & 1335 & 2691 & 2508\\
50 & 50 & NUTS-Joint & 438 & 407 & 872 & 840\\
50 & 50 & NUTS-Kalman & 8132 & 8277 & 8900 & 8643\\
50 & 100 & Gibbs & 1362 & 1342 & 2886 & 2846\\
50 & 100 & NUTS-Joint & 381 & 380 & 812 & 813\\
\addlinespace
50 & 100 & NUTS-Kalman & 7260 & 7042 & 8334 & 7446\\
50 & 200 & Gibbs & 1317 & 1321 & 2759 & 2789\\
50 & 200 & NUTS-Joint & 365 & 355 & 787 & 765\\
50 & 200 & NUTS-Kalman & 6222 & 6072 & 8339 & 7537\\
100 & 50 & Gibbs & 1454 & 1442 & 3028 & 3078\\
\addlinespace
100 & 50 & NUTS-Joint & 407 & 396 & 841 & 821\\
100 & 50 & NUTS-Kalman & 6435 & 6368 & 6922 & 6401\\
100 & 100 & Gibbs & 1353 & 1326 & 2920 & 2878\\
100 & 100 & NUTS-Joint & 369 & 376 & 798 & 787\\
100 & 100 & NUTS-Kalman & 5705 & 5309 & 7679 & 6530\\
\addlinespace
100 & 200 & Gibbs & 1333 & 1327 & 2871 & 2845\\
100 & 200 & NUTS-Joint & 362 & 368 & 774 & 784\\
100 & 200 & NUTS-Kalman & 4678 & 4443 & 6311 & 5592\\
\bottomrule
\end{tabular}
\end{table}

\begin{table}
\centering
\caption{\label{r4:tab:eg1-univariate-latent-ar1-time-total}Wall clock time in minutes for the multilevel scalar latent AR(1) model of Section 4.1, over the same Monte Carlo samples as Table \ref{r4:tab:eg1-univariate-latent-ar1-ess}. Dividing the effective sample sizes in that table by the times here recovers the efficiencies plotted in Section 4.1 up to the difference between the ratio of medians and the median of ratios, which is what the figures show.}
\centering
\begin{tabular}[t]{rrlrr}
\toprule
\multicolumn{3}{c}{ } & \multicolumn{2}{c}{Time (minutes)} \\
\cmidrule(l{3pt}r{3pt}){4-5}
N & T & Method & Mean & Median\\
\midrule
50 & 50 & Gibbs & 24.9 & 24.3\\
50 & 50 & NUTS-Joint & 5.4 & 4.6\\
50 & 50 & NUTS-Kalman & 6.1 & 3.3\\
50 & 100 & Gibbs & 57.3 & 55.8\\
50 & 100 & NUTS-Joint & 13.1 & 10.2\\
\addlinespace
50 & 100 & NUTS-Kalman & 10.7 & 10.4\\
50 & 200 & Gibbs & 161.8 & 157.7\\
50 & 200 & NUTS-Joint & 45.8 & 46.0\\
50 & 200 & NUTS-Kalman & 47.2 & 43.0\\
100 & 50 & Gibbs & 60.9 & 60.3\\
\addlinespace
100 & 50 & NUTS-Joint & 12.1 & 10.5\\
100 & 50 & NUTS-Kalman & 13.2 & 11.4\\
100 & 100 & Gibbs & 180.4 & 169.1\\
100 & 100 & NUTS-Joint & 40.6 & 41.7\\
100 & 100 & NUTS-Kalman & 43.3 & 40.9\\
\addlinespace
100 & 200 & Gibbs & 386.7 & 395.2\\
100 & 200 & NUTS-Joint & 137.8 & 132.7\\
100 & 200 & NUTS-Kalman & 96.2 & 94.9\\
\bottomrule
\end{tabular}
\end{table}

\clearpage

\section{Multilevel One-Factor Three-Indicator AR(1) Model}
\label{r4:sec:eg2}

Figure \ref{r4:fig:eg2-three-indicator-ar1-rhat-success} shows the proportion of successful runs, with different $\hat{R}$ thresholds. Figures \ref{r4:fig:eg2-three-indicator-ar1-bulk-ess-all-pars} and \ref{r4:fig:eg2-three-indicator-ar1-tail-ess-all-pars} show the bulk-ESS and tail-ESS per time unit for all parameters separately. 

Figure \ref{r4:fig:eg2-agreement} compares the three algorithms in the same way as in Section \ref{r4:sec:eg1}, and is read the same way.

Table \ref{r4:tab:eg2-three-indicator-ar1-ess} shows effective sample sizes and Table \ref{r4:tab:eg2-three-indicator-ar1-time-total} shows wall times.

\begin{figure}
    \centering
    \includegraphics[width=\linewidth]{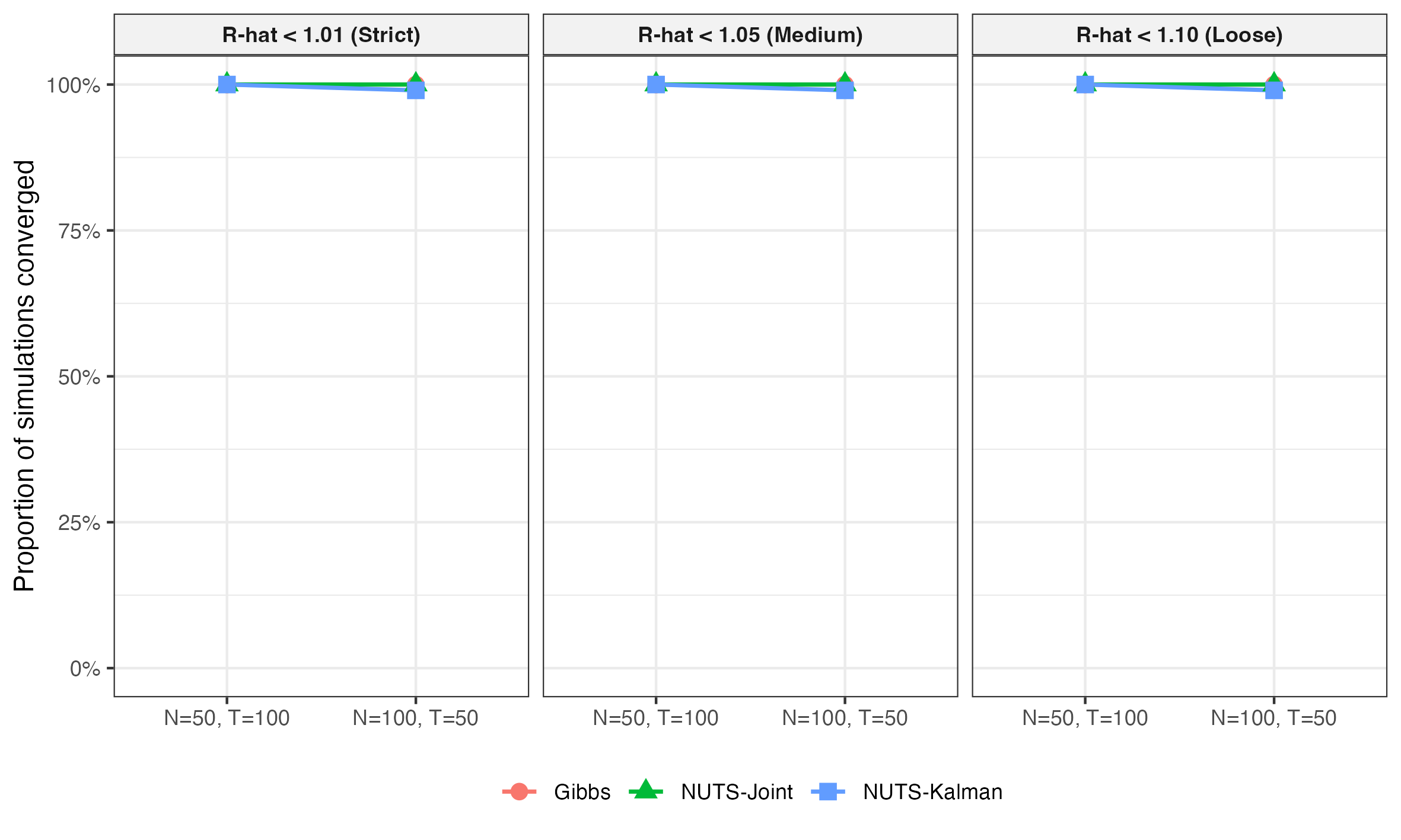}
    \caption{Proportion of simulations for which $\hat{R}$ was below the given thresholds.}
    \label{r4:fig:eg2-three-indicator-ar1-rhat-success}
\end{figure}

\begin{figure}
    \centering
    \includegraphics[width=\linewidth]{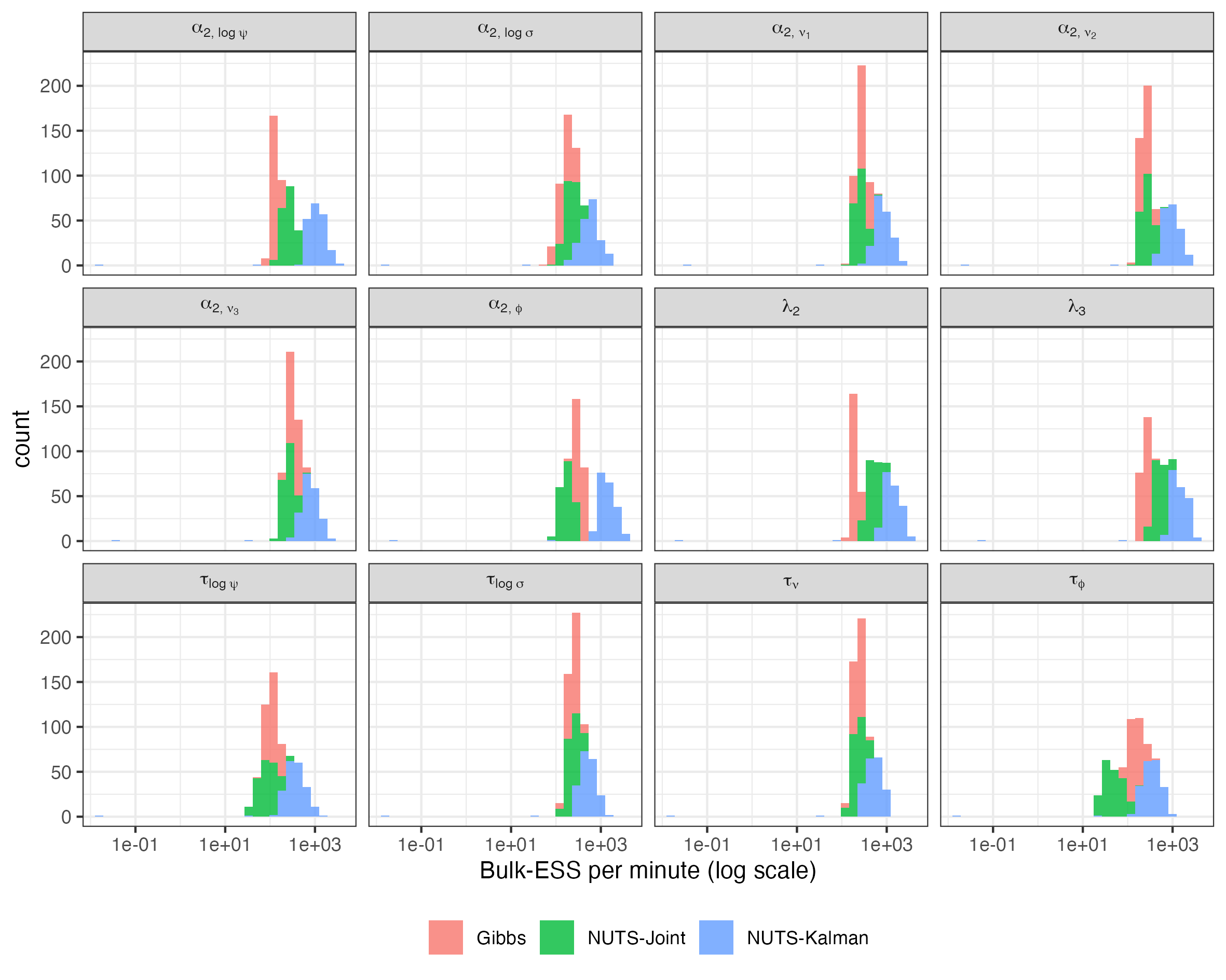}
    \caption{Sampling efficiency in terms of bulk-ESS for the three algorithms compared for the multilevel one-factor three-indicator AR(1) model.}
    \label{r4:fig:eg2-three-indicator-ar1-bulk-ess-all-pars}
\end{figure}

\begin{figure}
    \centering
    \includegraphics[width=\linewidth]{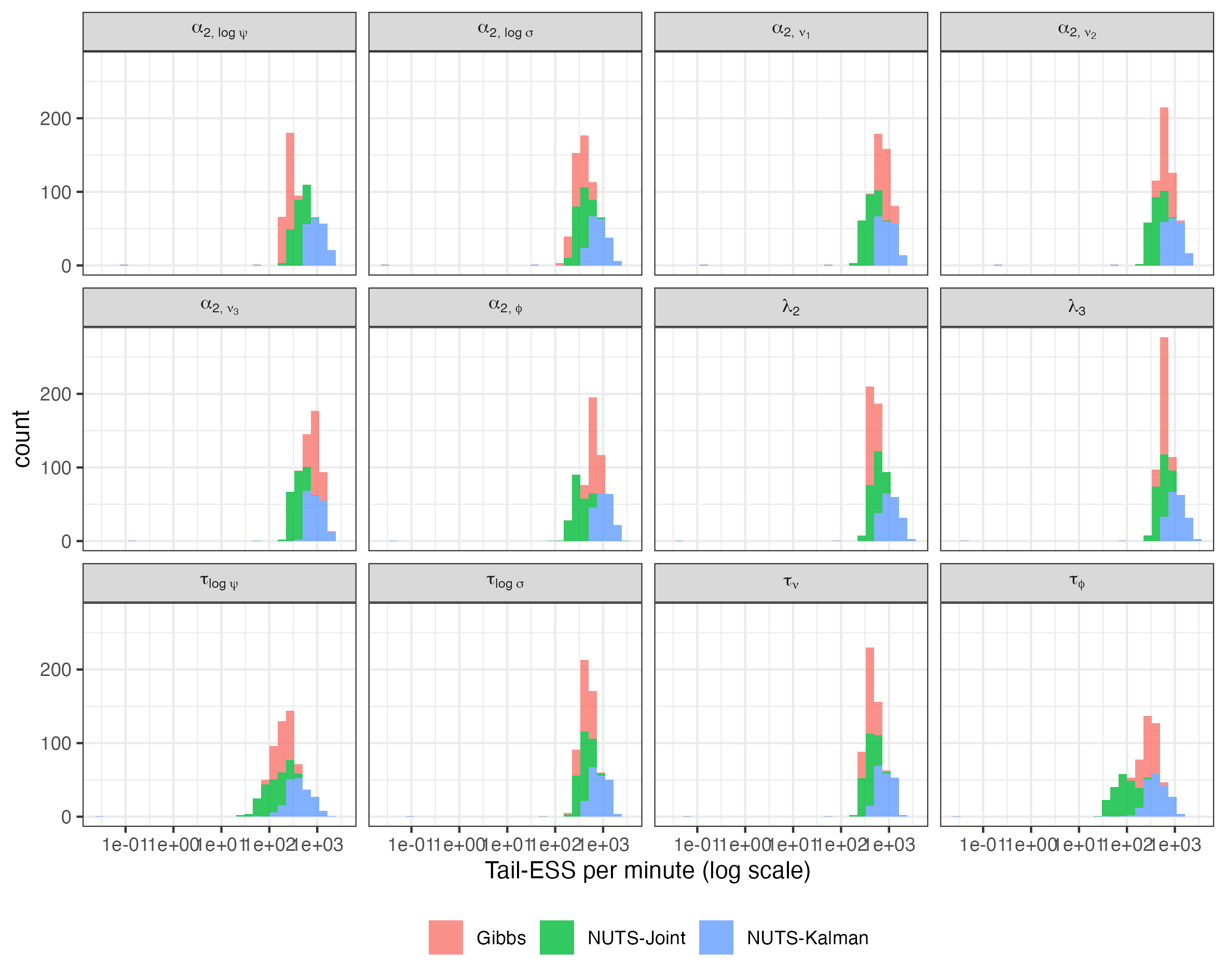}
    \caption{Sampling efficiency in terms of tail-ESS for the three algorithms compared for the multilevel one-factor three-indicator AR(1) model.}
    \label{r4:fig:eg2-three-indicator-ar1-tail-ess-all-pars}
\end{figure}

\begin{landscape}
\begin{figure}[H]
    \centering
    \includegraphics[width=\linewidth, height=0.88\textheight, keepaspectratio]{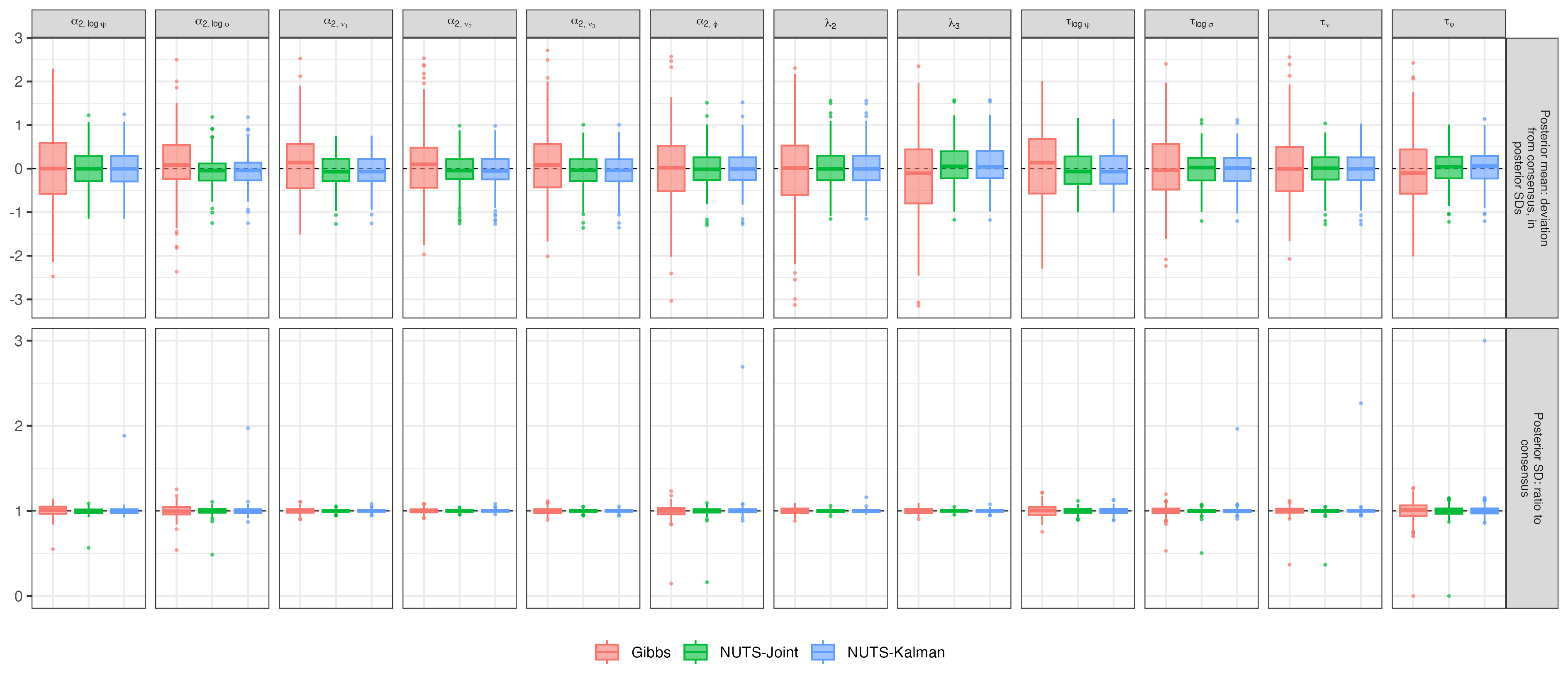}
    \caption{Agreement between the three algorithms for the multilevel one-factor three-indicator AR(1) model. Each algorithm is compared with the average of the three within each Monte Carlo sample, so the comparison is paired. The upper row shows the deviation of the posterior mean from that average, in units of the posterior standard deviation; the lower row shows the ratio of the posterior standard deviation to the same average. Boxes centred on the dashed lines indicate that the algorithms target the same posterior, and their width reflects Monte Carlo error rather than disagreement.}
    \label{r4:fig:eg2-agreement}
\end{figure}
\end{landscape}

\begin{table}
\centering
\caption{\label{r4:tab:eg2-three-indicator-ar1-ess}Effective sample size for the multilevel one-factor three-indicator AR(1) model of Section 4.2. Within each Monte Carlo sample the minimum across the parameters of interest is taken, and the mean and median of those minima across Monte Carlo samples are reported.}
\centering
\begin{tabular}[t]{rrlrrrr}
\toprule
\multicolumn{3}{c}{ } & \multicolumn{2}{c}{Bulk-ESS} & \multicolumn{2}{c}{Tail-ESS} \\
\cmidrule(l{3pt}r{3pt}){4-5} \cmidrule(l{3pt}r{3pt}){6-7}
N & T & Method & Mean & Median & Mean & Median\\
\midrule
100 & 50 & Gibbs & 15012 & 14812 & 27130 & 25610\\
100 & 50 & NUTS-Joint & 2359 & 2254 & 4338 & 4312\\
100 & 50 & NUTS-Kalman & 8792 & 8959 & 10305 & 9573\\
50 & 100 & Gibbs & 16009 & 15899 & 29596 & 29722\\
50 & 100 & NUTS-Joint & 3253 & 3090 & 5337 & 4944\\
\addlinespace
50 & 100 & NUTS-Kalman & 9284 & 9303 & 10839 & 10264\\
\bottomrule
\end{tabular}
\end{table}

\begin{table}
\centering
\caption{\label{r4:tab:eg2-three-indicator-ar1-time-total}Wall clock time in minutes for the multilevel one-factor three-indicator AR(1) model of Section 4.2, over the same Monte Carlo samples as Table \ref{r4:tab:eg2-three-indicator-ar1-ess}. Dividing the effective sample sizes in that table by the times here recovers the efficiencies plotted in Section 4.2 up to the difference between the ratio of medians and the median of ratios, which is what the figures show.}
\centering
\begin{tabular}[t]{rrlrr}
\toprule
\multicolumn{3}{c}{ } & \multicolumn{2}{c}{Time (minutes)} \\
\cmidrule(l{3pt}r{3pt}){4-5}
N & T & Method & Mean & Median\\
\midrule
100 & 50 & Gibbs & 159.8 & 168.2\\
100 & 50 & NUTS-Joint & 65.5 & 65.5\\
100 & 50 & NUTS-Kalman & 37.4 & 32.4\\
50 & 100 & Gibbs & 153.3 & 158.7\\
50 & 100 & NUTS-Joint & 55.8 & 54.1\\
\addlinespace
50 & 100 & NUTS-Kalman & 32.4 & 25.9\\
\bottomrule
\end{tabular}
\end{table}

\clearpage

\section{Multilevel Trivariate VAR(1) Model}
\label{r4:sec:eg3}

Figure \ref{r4:fig:eg3-trivariate-rhat-success} shows the proportion of successful runs, with different $\hat{R}$ thresholds. Figure \ref{r4:fig:eg3-agreement} compares the three algorithms in the same way as in Section \ref{r4:sec:eg1}. There are 28 parameters of interest here, so they are grouped into families, and the three values of $K$ are pooled.

Table \ref{r4:tab:eg3-trivariate-var1-ess} shows effective sample sizes and Table \ref{r4:tab:eg3-trivariate-var1-time-total} shows wall times.

\begin{figure}
    \centering
    \includegraphics[width=\linewidth]{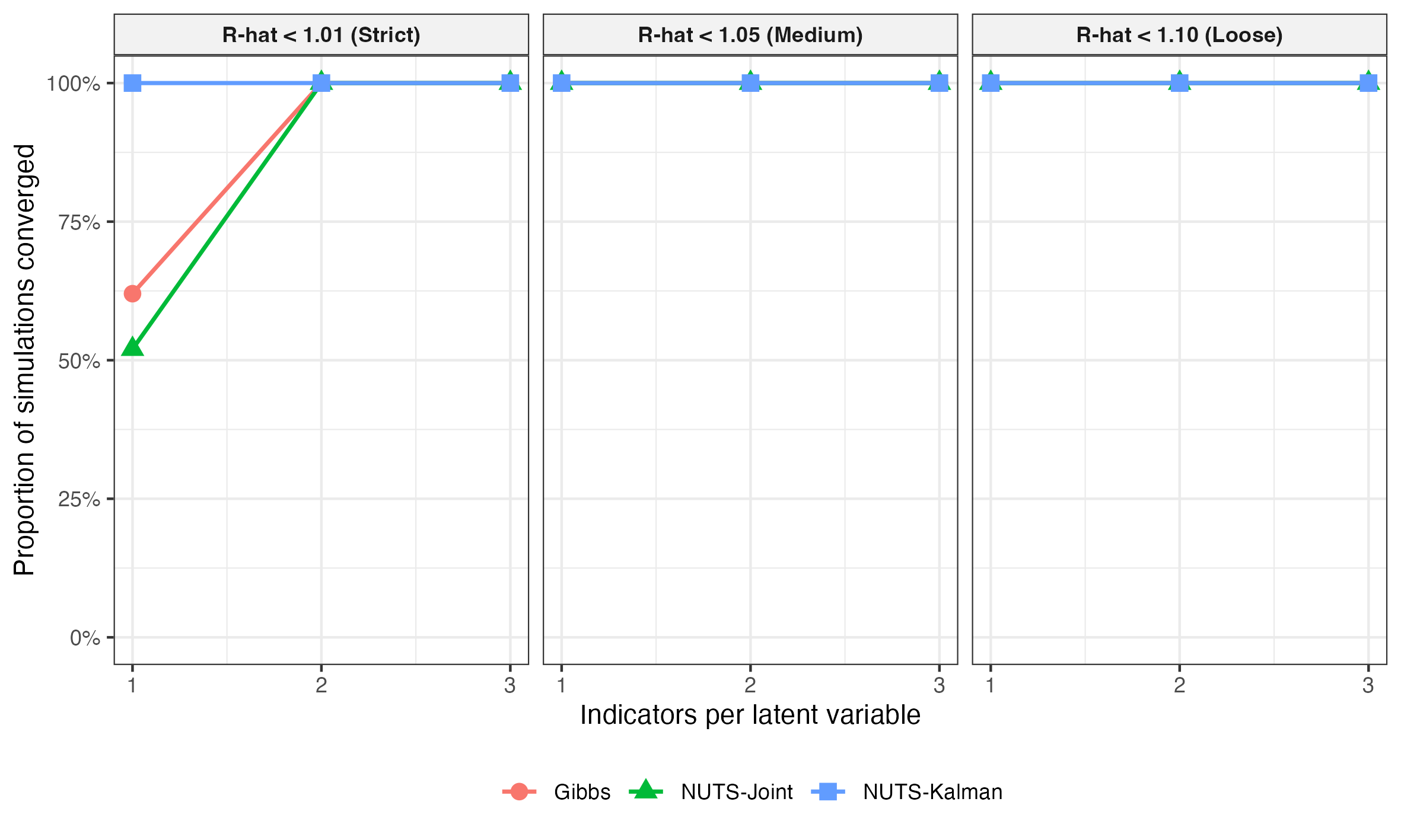}
    \caption{Proportion of simulations for which $\hat{R}$ was below the given thresholds.}
    \label{r4:fig:eg3-trivariate-rhat-success}
\end{figure}

\begin{landscape}
\begin{figure}[H]
    \centering
    \includegraphics[width=\linewidth, height=0.88\textheight, keepaspectratio]{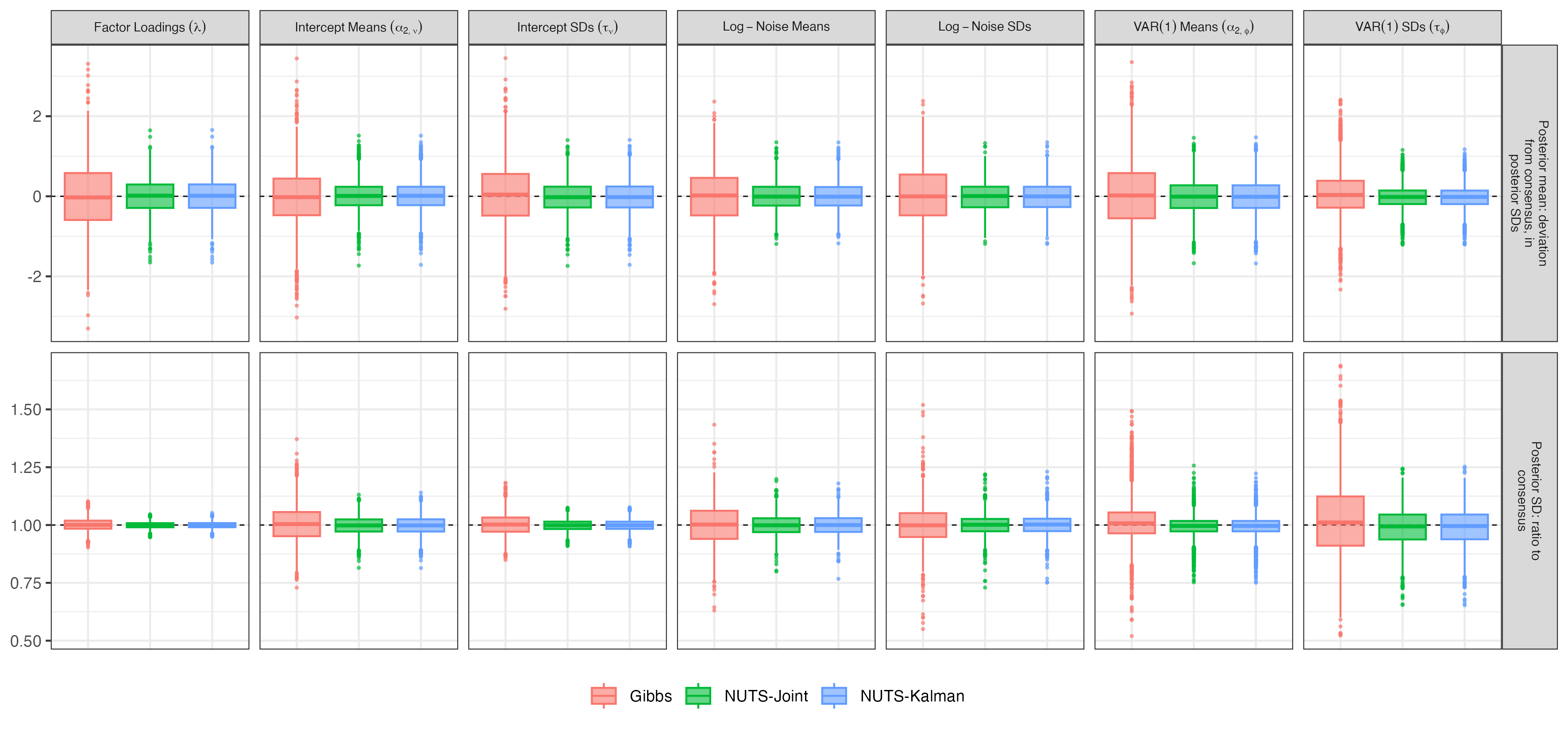}
    \caption{Agreement between the three algorithms for the multilevel trivariate VAR(1) model. Each algorithm is compared with the average of the three within each Monte Carlo sample, so the comparison is paired. The upper row shows the deviation of the posterior mean from that average, in units of the posterior standard deviation; the lower row shows the ratio of the posterior standard deviation to the same average. Boxes centred on the dashed lines indicate that the algorithms target the same posterior, and their width reflects Monte Carlo error rather than disagreement. Parameters are grouped into families, and the three values of $K$ are pooled.}
    \label{r4:fig:eg3-agreement}
\end{figure}
\end{landscape}

\begin{table}
\centering
\caption{\label{r4:tab:eg3-trivariate-var1-ess}Effective sample size for the multilevel trivariate VAR(1) model of Section 4.3, with $N=50$ and $T=100$. Within each Monte Carlo sample the minimum across the parameters of interest is taken, and the mean and median of those minima across Monte Carlo samples are reported.}
\centering
\begin{tabular}[t]{rlrrrr}
\toprule
\multicolumn{2}{c}{ } & \multicolumn{2}{c}{Bulk-ESS} & \multicolumn{2}{c}{Tail-ESS} \\
\cmidrule(l{3pt}r{3pt}){3-4} \cmidrule(l{3pt}r{3pt}){5-6}
Indicators & Method & Mean & Median & Mean & Median\\
\midrule
1 & Gibbs & 490 & 474 & 961 & 956\\
1 & NUTS-Joint & 426 & 412 & 782 & 748\\
1 & NUTS-Kalman & 4493 & 4187 & 6340 & 5935\\
2 & Gibbs & 9033 & 9094 & 18553 & 18747\\
2 & NUTS-Joint & 2583 & 2525 & 4574 & 4274\\
\addlinespace
2 & NUTS-Kalman & 8876 & 9007 & 10227 & 9758\\
3 & Gibbs & 6563 & 6393 & 13880 & 13622\\
3 & NUTS-Joint & 2683 & 2611 & 4403 & 4017\\
3 & NUTS-Kalman & 6018 & 5912 & 8768 & 8627\\
\bottomrule
\end{tabular}
\end{table}

\begin{table}
\centering
\caption{\label{r4:tab:eg3-trivariate-var1-time-total}Wall clock time in minutes for the multilevel trivariate VAR(1) model of Section 4.3, with $N=50$ and $T=100$, over the same Monte Carlo samples as Table \ref{r4:tab:eg3-trivariate-var1-ess}. Dividing the effective sample sizes in that table by the times here recovers the efficiencies plotted in Section 4.3 up to the difference between the ratio of medians and the median of ratios, which is what the figures show.}
\centering
\begin{tabular}[t]{rlrr}
\toprule
\multicolumn{2}{c}{ } & \multicolumn{2}{c}{Time (minutes)} \\
\cmidrule(l{3pt}r{3pt}){3-4}
Indicators & Method & Mean & Median\\
\midrule
1 & Gibbs & 565.3 & 563.9\\
1 & NUTS-Joint & 669.3 & 658.7\\
1 & NUTS-Kalman & 535.6 & 487.8\\
2 & Gibbs & 773.8 & 767.0\\
2 & NUTS-Joint & 431.6 & 438.2\\
\addlinespace
2 & NUTS-Kalman & 459.7 & 460.2\\
3 & Gibbs & 1015.8 & 1010.4\\
3 & NUTS-Joint & 553.7 & 558.8\\
3 & NUTS-Kalman & 490.4 & 486.3\\
\bottomrule
\end{tabular}
\end{table}

\clearpage

\section{Multilevel Latent AR(p) Model}
\label{r4:sec:eg4}

Figure \ref{r4:fig:eg4-latent-arp-rhat-success} shows the proportion of successful runs, with different $\hat{R}$ thresholds. Figure \ref{r4:fig:eg4-latent-arp-tail-efficiency} shows a comparison of the algorithms in terms of tail efficiency. Figure \ref{r4:fig:eg4-agreement} compares the three algorithms in the same way as in Section \ref{r4:sec:eg1}, with the parameters grouped into families and the four lags pooled.

Table \ref{r4:tab:eg4-latent-arp-ess} shows effective sample sizes and Table \ref{r4:tab:eg4-latent-arp-time-total} shows wall times.

\begin{figure}
    \centering
    \includegraphics[width=\linewidth]{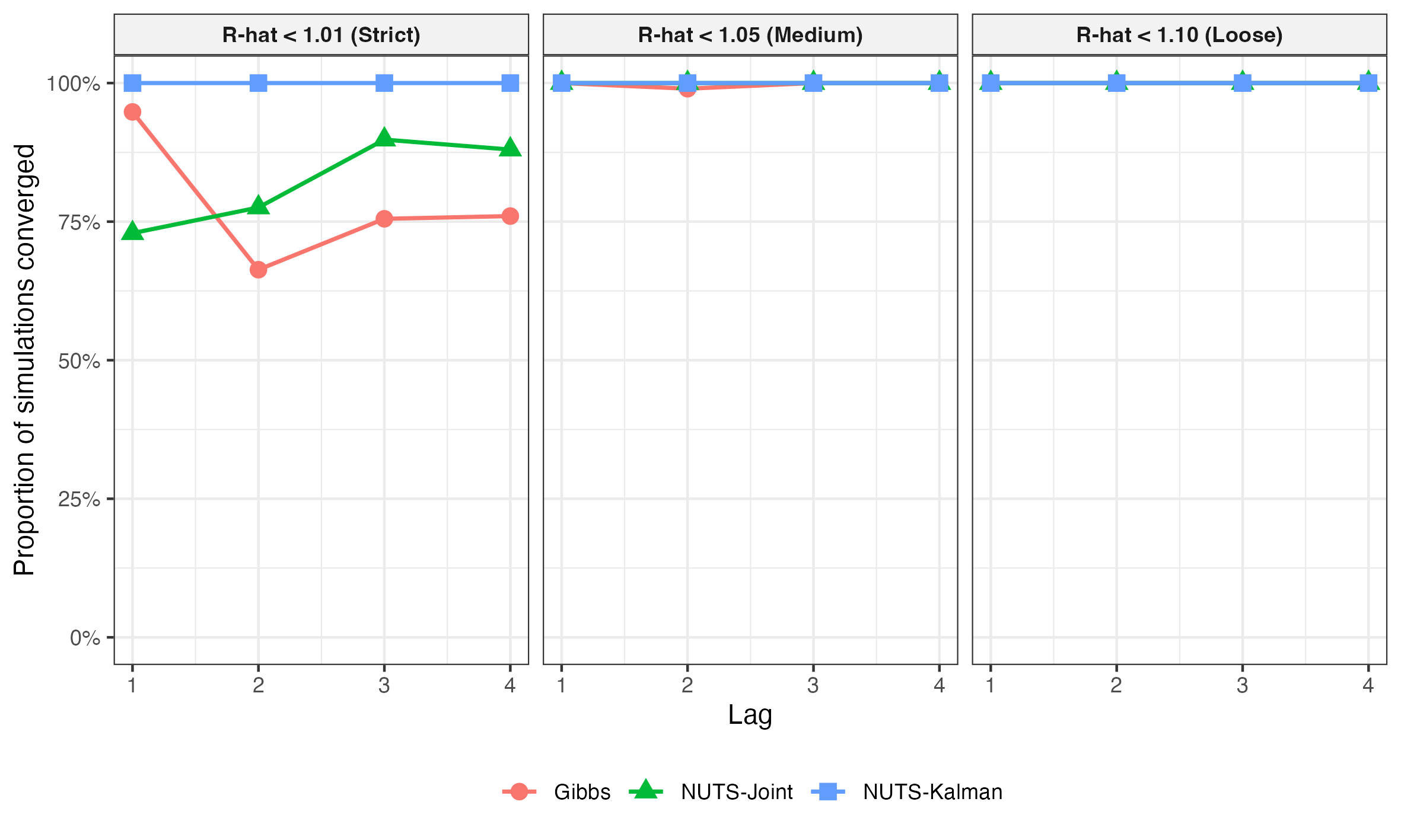}
    \caption{Proportion of simulations for which $\hat{R}$ was below the given thresholds.}
    \label{r4:fig:eg4-latent-arp-rhat-success}
\end{figure}

\begin{figure}
    \centering
    \includegraphics[width=\linewidth]{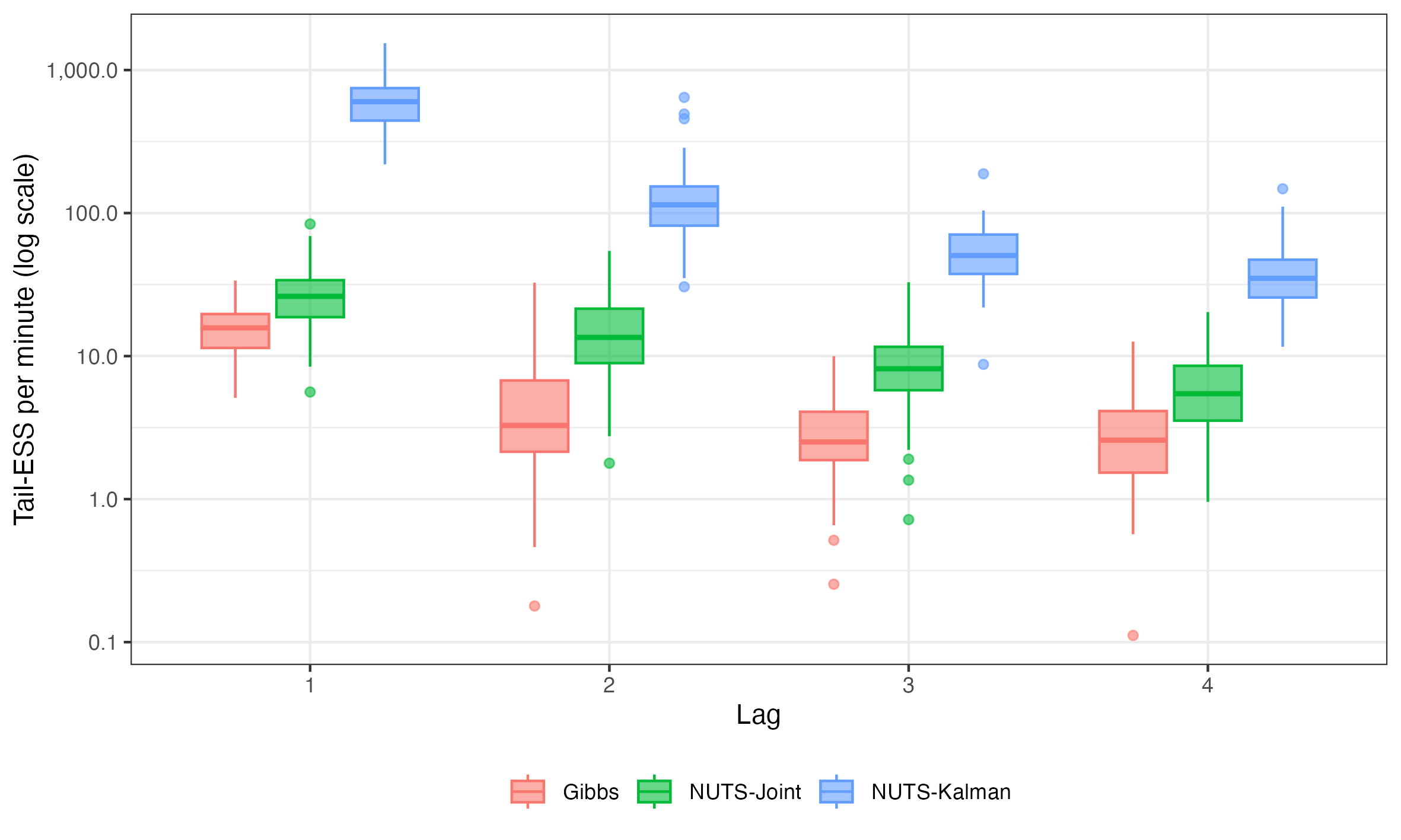}
    \caption{Box plot of tail efficiency for the multilevel latent AR($p$) model.}
    \label{r4:fig:eg4-latent-arp-tail-efficiency}
\end{figure}

\begin{landscape}
\begin{figure}[H]
    \centering
    \includegraphics[width=\linewidth, height=0.88\textheight, keepaspectratio]{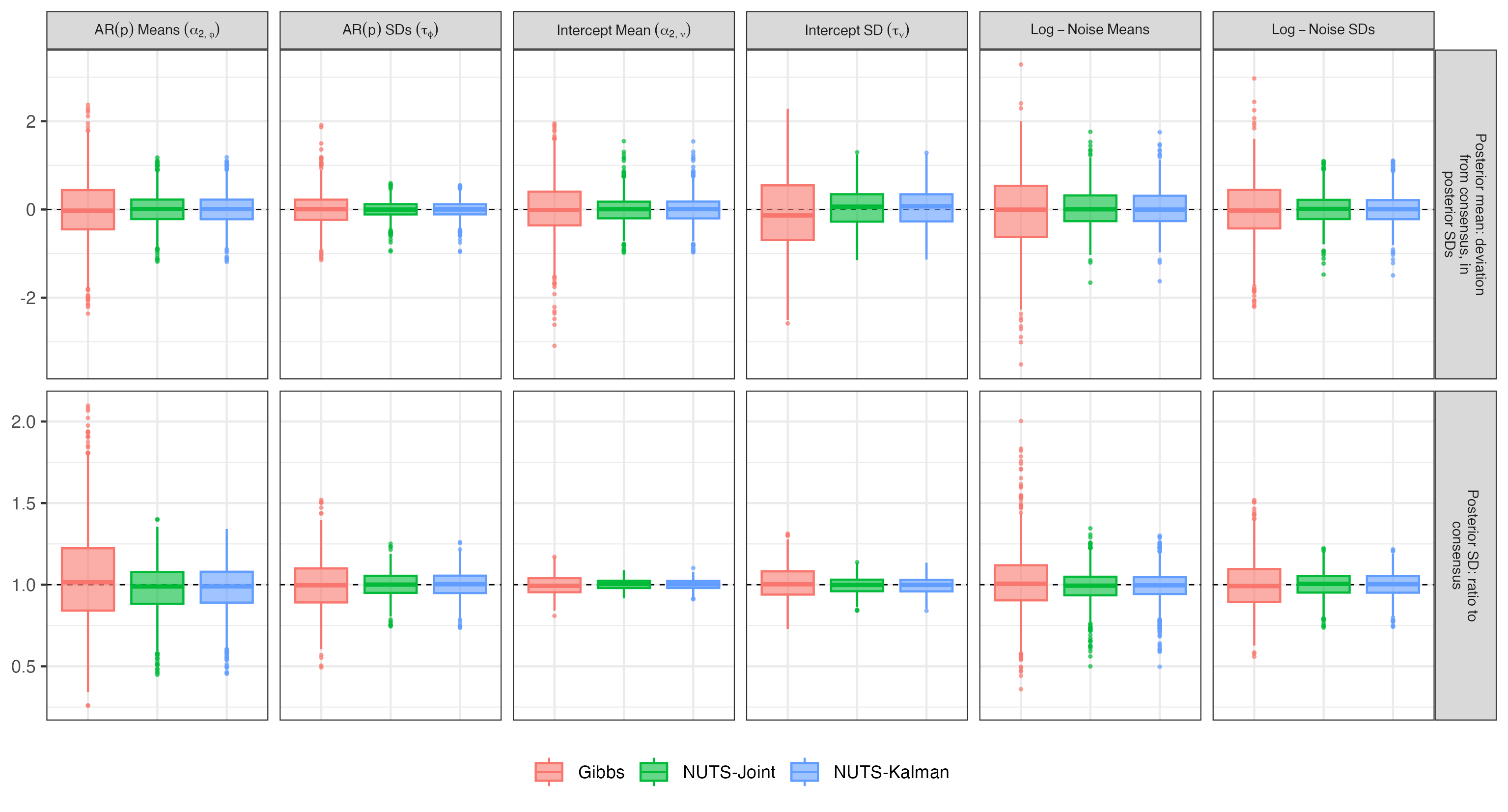}
    \caption{Agreement between the three algorithms for the multilevel latent AR($p$) model. Each algorithm is compared with the average of the three within each Monte Carlo sample, so the comparison is paired. The upper row shows the deviation of the posterior mean from that average, in units of the posterior standard deviation; the lower row shows the ratio of the posterior standard deviation to the same average. Boxes centred on the dashed lines indicate that the algorithms target the same posterior, and their width reflects Monte Carlo error rather than disagreement. Parameters are grouped into families, and the four lags are pooled.}
    \label{r4:fig:eg4-agreement}
\end{figure}
\end{landscape}

\begin{table}
\centering
\caption{\label{r4:tab:eg4-latent-arp-ess}Effective sample size for the multilevel latent AR(p) model of Section 4.4, with $N=100$ and $T=50$. Within each Monte Carlo sample the minimum across the parameters of interest is taken, and the mean and median of those minima across Monte Carlo samples are reported.}
\centering
\begin{tabular}[t]{rlrrrr}
\toprule
\multicolumn{2}{c}{ } & \multicolumn{2}{c}{Bulk-ESS} & \multicolumn{2}{c}{Tail-ESS} \\
\cmidrule(l{3pt}r{3pt}){3-4} \cmidrule(l{3pt}r{3pt}){5-6}
Lag & Method & Mean & Median & Mean & Median\\
\midrule
1 & Gibbs & 794 & 810 & 1489 & 1433\\
1 & NUTS-Joint & 527 & 476 & 919 & 880\\
1 & NUTS-Kalman & 5852 & 5618 & 6837 & 6547\\
2 & Gibbs & 709 & 561 & 848 & 504\\
2 & NUTS-Joint & 589 & 590 & 1110 & 1108\\
\addlinespace
2 & NUTS-Kalman & 5714 & 5440 & 5640 & 5106\\
3 & Gibbs & 662 & 589 & 681 & 553\\
3 & NUTS-Joint & 656 & 669 & 1218 & 1250\\
3 & NUTS-Kalman & 5682 & 5711 & 5834 & 5384\\
4 & Gibbs & 791 & 631 & 899 & 658\\
\addlinespace
4 & NUTS-Joint & 718 & 722 & 1311 & 1273\\
4 & NUTS-Kalman & 6186 & 6120 & 6512 & 6245\\
\bottomrule
\end{tabular}
\end{table}

\begin{table}
\centering
\caption{\label{r4:tab:eg4-latent-arp-time-total}Wall clock time in minutes for the multilevel latent AR(p) model of Section 4.4, with $N=100$ and $T=50$, over the same Monte Carlo samples as Table \ref{r4:tab:eg4-latent-arp-ess}. Dividing the effective sample sizes in that table by the times here recovers the efficiencies plotted in Section 4.4 up to the difference between the ratio of medians and the median of ratios, which is what the figures show.}
\centering
\begin{tabular}[t]{rlrr}
\toprule
\multicolumn{2}{c}{ } & \multicolumn{2}{c}{Time (minutes)} \\
\cmidrule(l{3pt}r{3pt}){3-4}
Lag & Method & Mean & Median\\
\midrule
1 & Gibbs & 93.2 & 92.9\\
1 & NUTS-Joint & 35.6 & 36.2\\
1 & NUTS-Kalman & 11.1 & 10.2\\
2 & Gibbs & 158.1 & 152.3\\
2 & NUTS-Joint & 83.4 & 80.3\\
\addlinespace
2 & NUTS-Kalman & 47.5 & 44.7\\
3 & Gibbs & 220.8 & 220.5\\
3 & NUTS-Joint & 158.0 & 137.0\\
3 & NUTS-Kalman & 110.6 & 104.8\\
4 & Gibbs & 264.9 & 265.1\\
\addlinespace
4 & NUTS-Joint & 248.5 & 228.4\\
4 & NUTS-Kalman & 187.6 & 180.6\\
\bottomrule
\end{tabular}
\end{table}

\clearpage

\section{Cross-Classified VAR(1) Model}

Figure \ref{r4:fig:eg5-rhat} shows box plots of $\hat{R}$ values. Figure \ref{r4:fig:eg5-agreement} compares the three algorithms in the same way as in Section \ref{r4:sec:eg1}, with the parameters grouped into families. Note that NUTS-Joint disagrees with Gibbs and NUTS-Kalman for several parameters while the latter two algorithms agree well. This behavior is consistent with the poor convergence of NUTS-Joint in this example. Table \ref{r4:tab:eg5-cross-classified-ess} shows effective sample sizes and Table \ref{r4:tab:eg5-cross-classified-time-total} shows wall times.

\begin{figure}
    \centering
    \includegraphics[width=\linewidth]{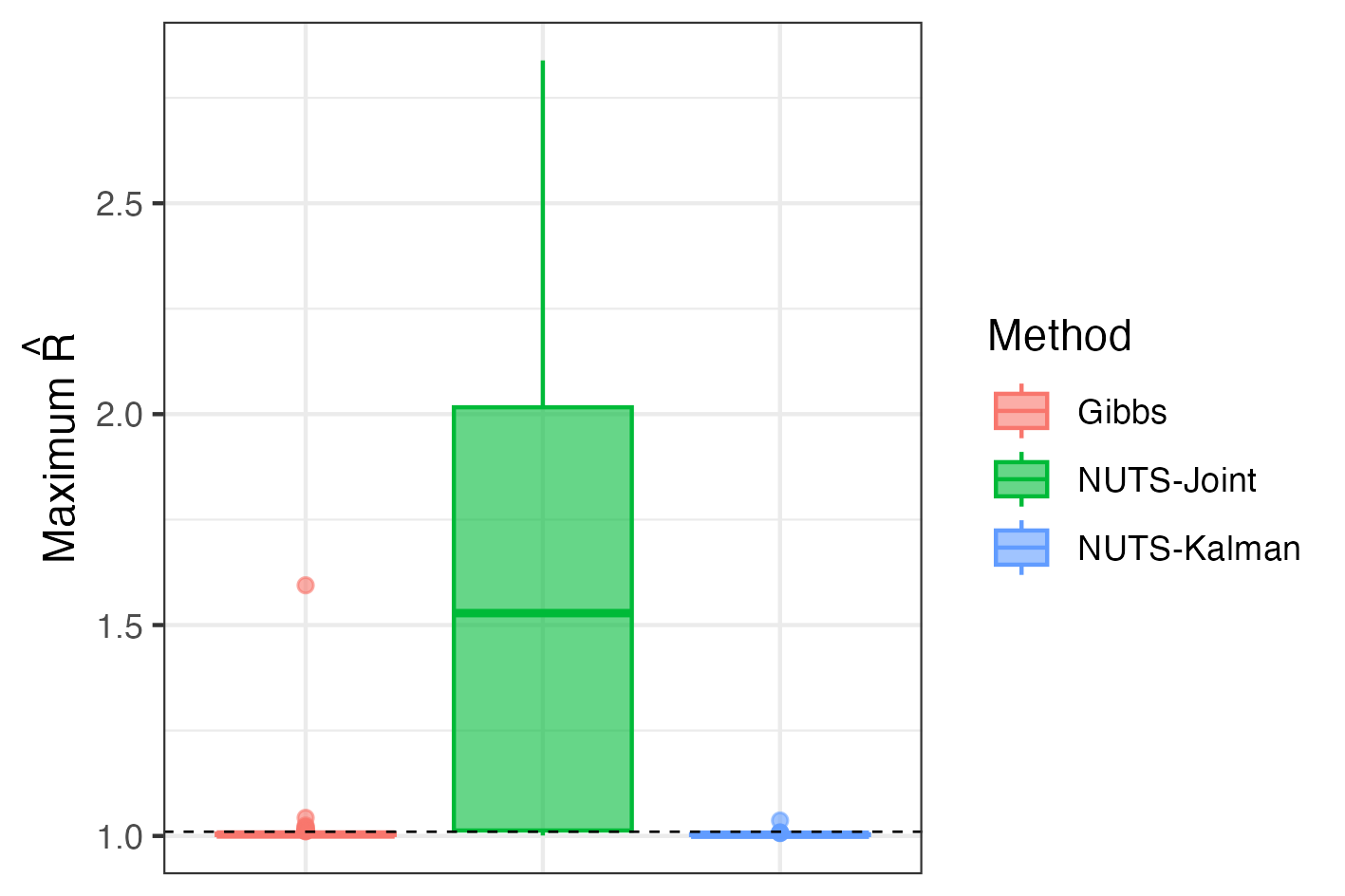}
    \caption{Box plots of $\hat{R}$ values for the three algorithms in the cross-classified VAR(1) model.}
    \label{r4:fig:eg5-rhat}
\end{figure}

\begin{landscape}
\begin{figure}[H]
    \centering
    \includegraphics[width=\linewidth, height=0.88\textheight, keepaspectratio]{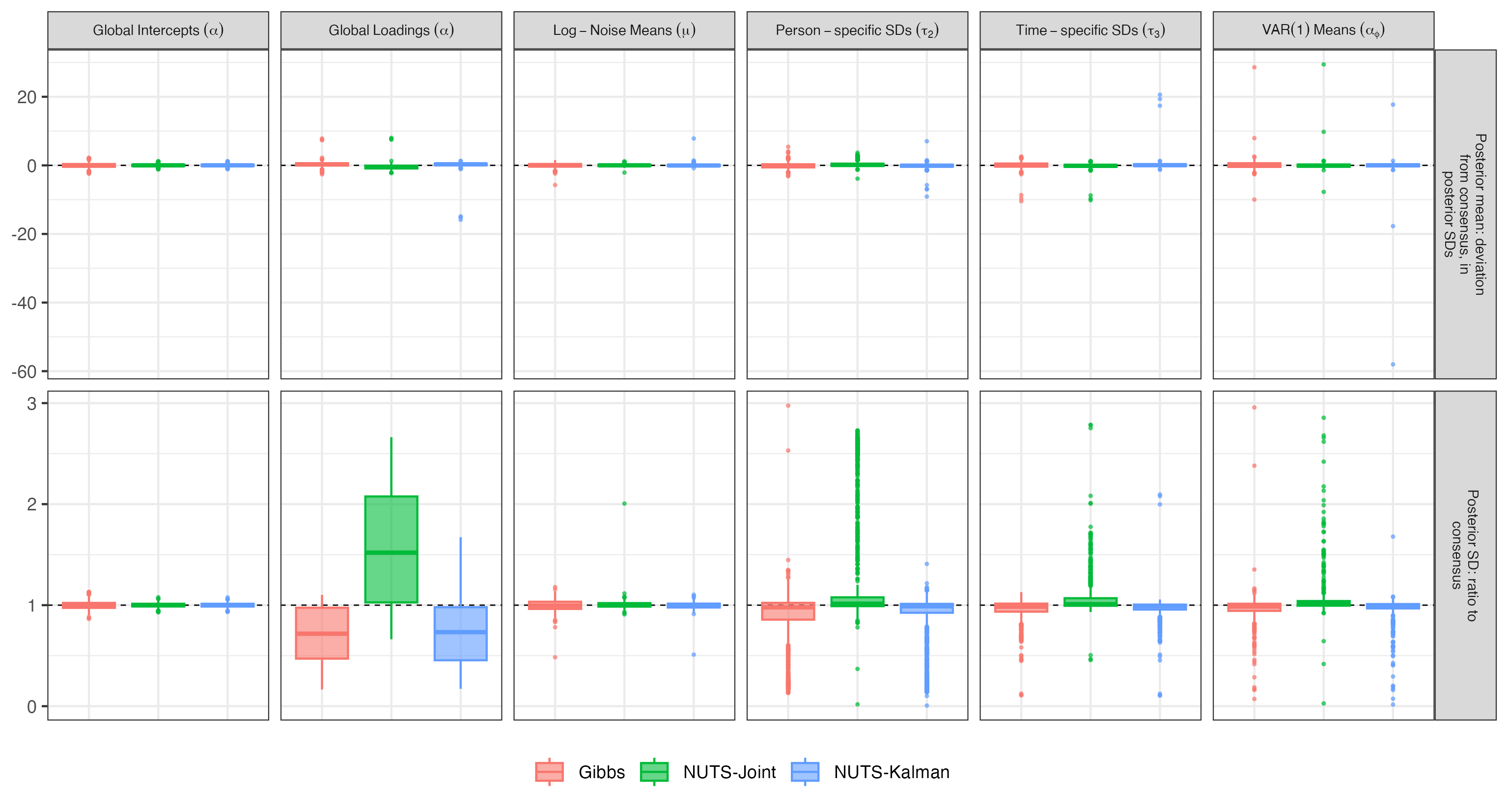}
    \caption{Agreement between the three algorithms for the cross-classified VAR(1) model. Each algorithm is compared with the average of the three within each Monte Carlo sample, so the comparison is paired. The upper row shows the deviation of the posterior mean from that average, in units of the posterior standard deviation; the lower row shows the ratio of the posterior standard deviation to the same average. Boxes centred on the dashed lines indicate that the algorithms target the same posterior, and their width reflects Monte Carlo error rather than disagreement. Parameters are grouped into families.}
    \label{r4:fig:eg5-agreement}
\end{figure}
\end{landscape}

\begin{table}
\centering
\caption{\label{r4:tab:eg5-cross-classified-ess}Effective sample size for the cross-classified VAR(1) model of Section 4.5, with $N=100$ and $T=150$. Within each Monte Carlo sample the minimum across the parameters of interest is taken, and the mean and median of those minima across Monte Carlo samples are reported.}
\centering
\begin{tabular}[t]{lrrrr}
\toprule
\multicolumn{1}{c}{ } & \multicolumn{2}{c}{Bulk-ESS} & \multicolumn{2}{c}{Tail-ESS} \\
\cmidrule(l{3pt}r{3pt}){2-3} \cmidrule(l{3pt}r{3pt}){4-5}
Method & Mean & Median & Mean & Median\\
\midrule
Gibbs & 1223 & 1302 & 2360 & 2523\\
NUTS-Joint & 336 & 7 & 737 & 26\\
NUTS-Kalman & 1269 & 1233 & 2719 & 2584\\
\bottomrule
\end{tabular}
\end{table}

\begin{table}
\centering
\caption{\label{r4:tab:eg5-cross-classified-time-total}Wall clock time in minutes for the cross-classified VAR(1) model of Section 4.5, with $N=100$ and $T=150$, over the same Monte Carlo samples as Table \ref{r4:tab:eg5-cross-classified-ess}. Dividing the effective sample sizes in that table by the times here recovers the efficiencies plotted in Section 4.5 up to the difference between the ratio of medians and the median of ratios, which is what the figures show.}
\centering
\begin{tabular}[t]{lrr}
\toprule
\multicolumn{1}{c}{ } & \multicolumn{2}{c}{Time (minutes)} \\
\cmidrule(l{3pt}r{3pt}){2-3}
Method & Mean & Median\\
\midrule
Gibbs & 3018.0 & 3045.0\\
NUTS-Joint & 853.4 & 832.4\\
NUTS-Kalman & 1077.9 & 1050.4\\
\bottomrule
\end{tabular}
\end{table}

\onlineresource{5}

This document contains supplementary information related to the application example in Section 5. Figure \ref{r5:fig:application-histograms} shows histograms of all responses for all items. Figure \ref{r5:fig:application-data} shows observations of the items ``happy'' and ``sad'' for three selected participants. Figures \ref{r5:fig:application-agreement-means} and \ref{r5:fig:application-agreement-sds} shows posterior means and standard deviation of Metropolis-within-Gibbs and NUTS-Joint plotted against the same quantities for NUTS-Kalman. 

\begin{figure}
    \centering
    \includegraphics[width=\linewidth]{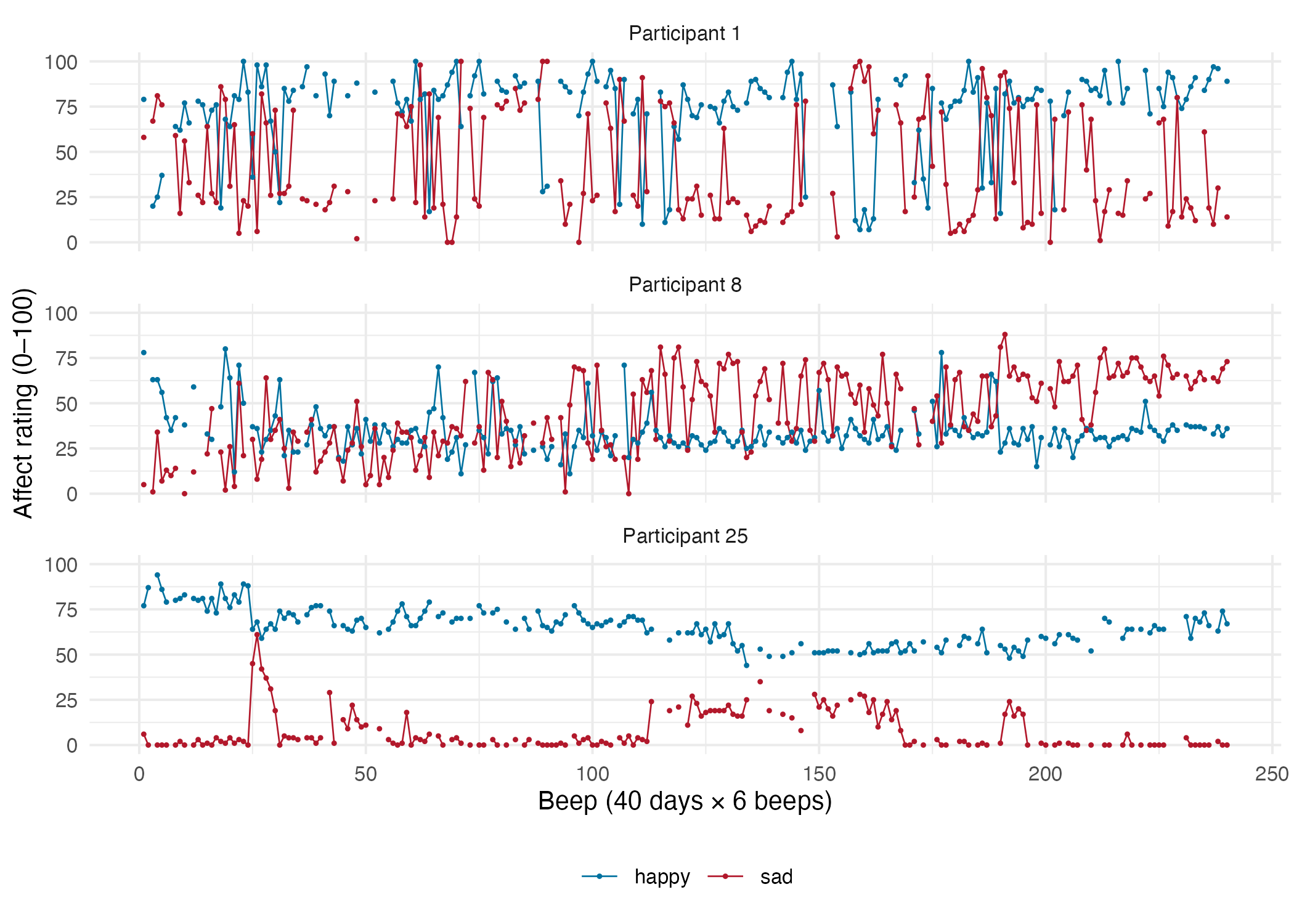}
    \caption{Data for participants 1, 25, and 8 for the items ``happy'' and ``sad''.}
    \label{r5:fig:application-data}
\end{figure}

\begin{figure}
    \centering
    \includegraphics[width=\linewidth]{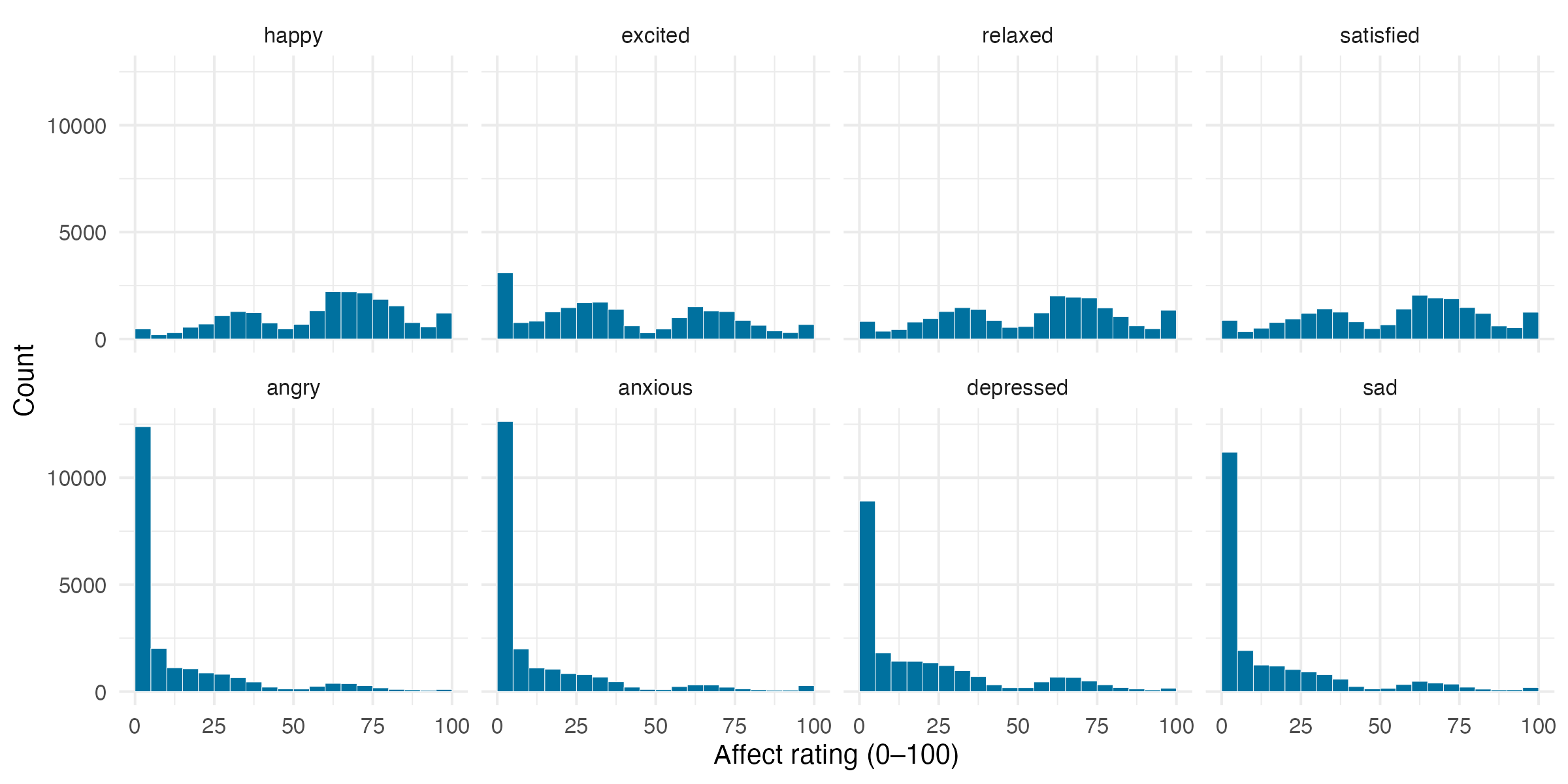}
    \caption{Histograms of responses for all items.}
    \label{r5:fig:application-histograms}
\end{figure}

\begin{figure}
    \centering
    \includegraphics[width=\linewidth]{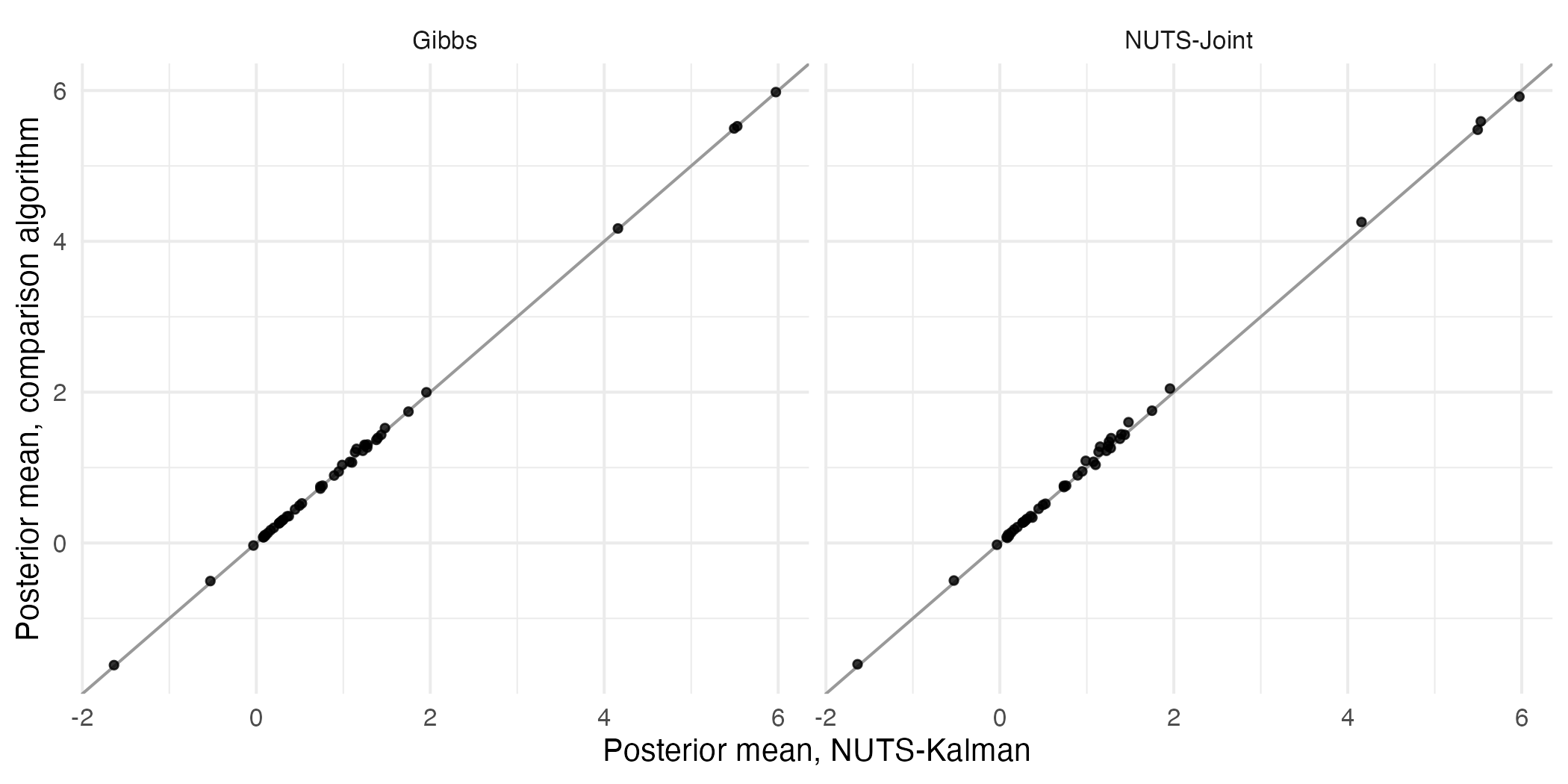}
    \caption{Posterior means of Gibbs and NUTS-Joint plotted against the posterior means of NUTS-Kalman, for all 45 parameters of interest.}
    \label{r5:fig:application-agreement-means}
\end{figure}

\begin{figure}
    \centering
    \includegraphics[width=\linewidth]{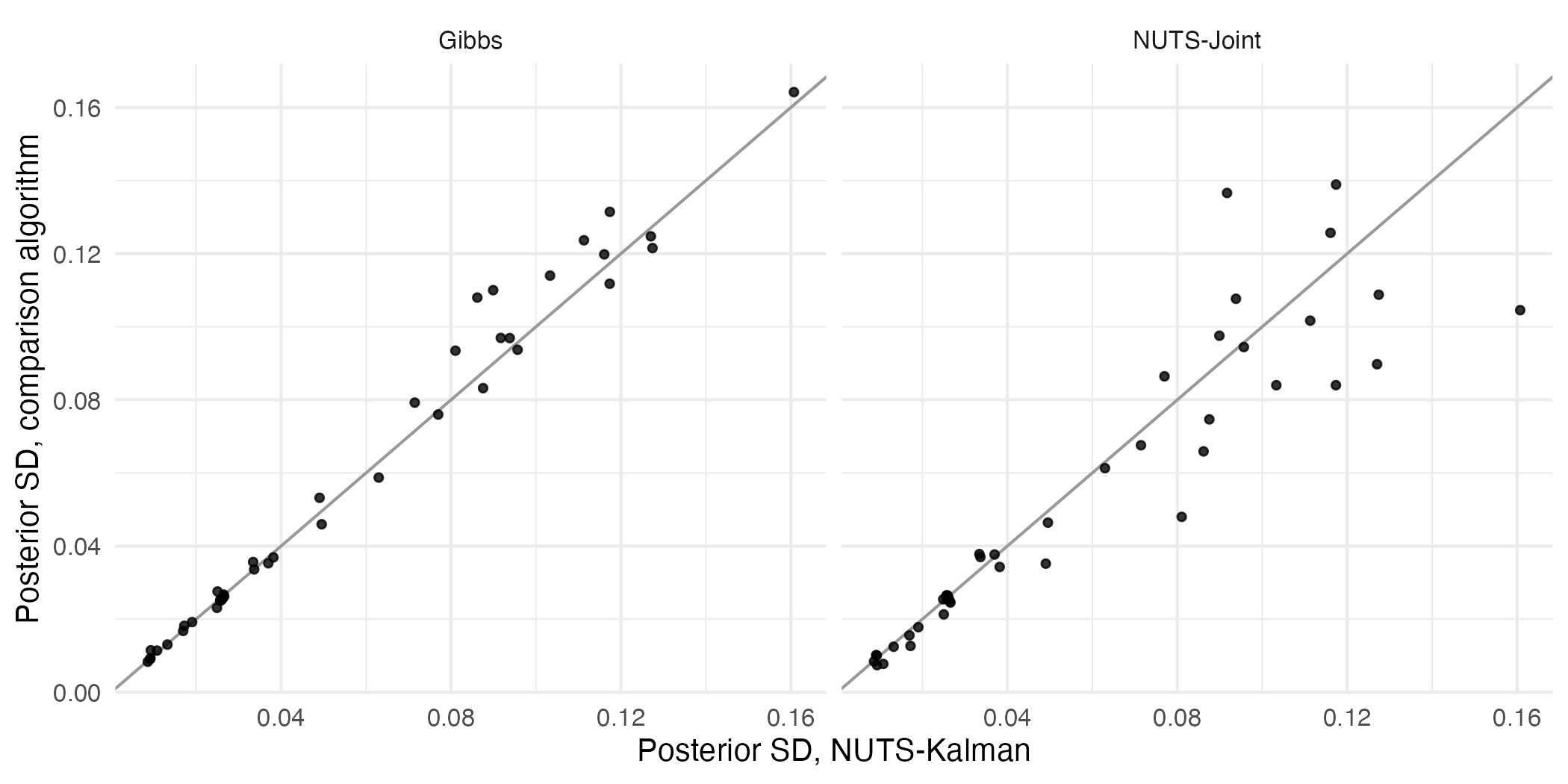}
    \caption{Posterior standard deviations of Gibbs and NUTS-Joint plotted against the posterior standard deviations of NUTS-Kalman, for all 45 parameters of interest.}
    \label{r5:fig:application-agreement-sds}
\end{figure}

\end{document}